\documentclass[
a4paper
]{article}

\usepackage[utf8]{inputenc}
\usepackage{xspace}
\usepackage{graphicx}
\usepackage{algorithm}
\usepackage[noend]{algpseudocode}

\usepackage{amsmath}
\usepackage{amssymb}
\usepackage{amsthm}

\usepackage{textcomp}
\usepackage{stmaryrd}
\usepackage{varwidth}
\usepackage{fullpage}

\usepackage{tikz}

\usepackage{multirow}

\allowdisplaybreaks

\usetikzlibrary{automata, graphs,positioning,chains,arrows,snakes,decorations.pathmorphing,calc,fit,shapes,shapes.geometric}

\makeatletter
\@fleqnfalse
\@mathmargin\@centering
\makeatother

\algdef{SE}[DOWHILE]{Do}{doWhile}{\algorithmicdo}[1]{\algorithmicwhile\ #1}%
\tikzset{
	ncbar angle/.initial=90,
	ncbar/.style={
		to path=(\tikztostart)
		-- ($(\tikztostart)!#1!\pgfkeysvalueof{/tikz/ncbar angle}:(\tikztotarget)$)
		-- ($(\tikztotarget)!($(\tikztostart)!#1!\pgfkeysvalueof{/tikz/ncbar angle}:(\tikztotarget)$)!\pgfkeysvalueof{/tikz/ncbar angle}:(\tikztostart)$)
		\tikztonodes
		-- (\tikztotarget)
	},
	ncbar/.default=0.5cm,
}

\tikzset{square left brace/.style={ncbar=-0.3cm}}
\tikzset{square right brace/.style={ncbar=0.3cm}}

\tikzset{round left paren/.style={ncbar=0.5cm,out=120,in=-120}}
\tikzset{round right paren/.style={ncbar=0.5cm,out=60,in=-60}}

\tikzset{rootnode/.style={circle,inner sep=1pt}}
\tikzset{anynode/.style={circle,minimum size=15pt}}
\tikzset{extnode/.style={circle,minimum size=12pt,draw}}
\tikzset{roundnode/.style={circle,inner sep=1pt}}
\tikzset{squarednode/.style={rectangle,inner sep=2pt}}
\tikzset{resetnode/.style={regular polygon, regular polygon sides=3, rotate=180,minimum size=15pt,anchor=mid,draw}}
\tikzset{clockchecknode/.style={draw,diamond,minimum size=15pt}}
\tikzset{emptynode/.style={minimum size=15pt,align=center,label=center:{\LARGE$\varnothing$}}}
\tikzset{unionnode/.style={minimum size=15pt,align=center,label=center:{\Large$\cup$}}}
\tikzset{grayunionnode/.style={minimum size=15pt,align=center,label=center:{\textcolor{gray}{\Large$\cup$}}}}

\usepackage
[disable]
{todonotes}

\newcommand{\cA}{\mathcal{A}}
\newcommand{\cS}{\mathcal{S}}

\newcommand{\cO}{\mathcal{O}}

\newcommand{\cT}{\mathcal{T}}
\newcommand{\cC}{\mathcal{C}}

\newcommand{\bbN}{\mathbb{N}}
\newcommand{\bbR}{\mathbb{R}}

\newcommand{\bbQplus}{\mathbb{Q}_{\geq 0}}
\newcommand{\bbQinf}{\mathbb{Q}_{\geq 0}^{\infty}}

\newcommand{\xset}{\mathbf{X}}
\newcommand{\zset}{\mathbf{Z}}

\newcommand{\sem}[1]{{\llbracket{}{#1}\rrbracket}}

\newcommand{\pow}[1]{2^{#1}}

\newcommand{\nequiv}{\not \equiv}

\newcommand{\tset}{\mathbf{T}} %
\newcommand{\aset}{\mathbf{A}} %
\newcommand{\dset}{\mathbf{D}} %
\newcommand{\etype}[1]{\operatorname{type}(#1)} 
 
\newcommand{\etypes}[1]{\operatorname{types}(#1)} 
\newcommand{\eatt}[1]{\operatorname{att}(#1)} 
\newcommand{\dom}[1]{\operatorname{dom}(#1)} 

\newcommand{\Stream}{\bar{S}} 
\newcommand{\StreamEx}{\bar{S}_{0}}

\newcommand{\pset}{\mathbf{P}} %
\newcommand{\psetbasic}{\mathbf{P}_{\operatorname{basic}}}
\newcommand{\eset}{\mathbf{E}} %

\newcommand{\true}{\texttt{true}} %

\newcommand{\kAS}{~\operatorname{ AS }~}
\newcommand{\kOR}{~\operatorname{ OR }~}
\newcommand{\kAND}{~\operatorname{ AND }~}

\newcommand{\kFILTER}{~\operatorname{ FILTER }~}

\newcommand{\PROJ}[1]{~\pi_{#1}~}

\newcommand{\cas}{\operatorname{AS}}
\newcommand{\cand}{\operatorname{AND}}
\newcommand{\ssq}{\,:\,}
\newcommand{\sq}{\,;\,}
\newcommand{\ks}{+\,}
\newcommand{\sks}{\oplus}
\newcommand{\xssq}[1]{\,:_{#1}\,}
\newcommand{\xsq}[1]{\,;_{#1}\,}
\newcommand{\xks}[1]{+_{#1}\,}
\newcommand{\ctimesq}[1]{;_{\, #1}}
\newcommand{\ctimessq}[1]{:_{\, #1}}
\newcommand{\timesq}[1]{\, \ctimesq{#1}\,}
\newcommand{\timessq}[1]{\ctimessq{#1}}
\newcommand{\timeks}[1]{+_{\,#1}\,}
\newcommand{\timesks}[1]{\oplus_{\,#1}\,}

\newcommand{\treset}{\operatorname{reset}}
\newcommand{\fract}{\operatorname{fract}}
\newcommand*\circled[1]{\tikz[baseline=(char.base)]{
		\node[shape=circle,draw,inner sep=2pt] (char) {#1};}}

\newcommand{\semaux}[1]{\sem{#1}}
\newcommand{\semauxop}[1]{{\llceil{#1}\rrfloor}}
\newcommand{\ctime}[0]{{\operatorname{interval}}}
\newcommand{\type}[0]{{\operatorname{type}}}
\newcommand{\cstart}[0]{{\operatorname{start}}}
\newcommand{\cend}[0]{{\operatorname{end}}}

\newcommand{\TRUE}{{\operatorname{true}}}
\newcommand{\FALSE}{{\operatorname{false}}}
\newcommand{\twin}[2]{\langle #1 \rangle_{#2}}

\newcommand{\nnode}{\texttt{n}}

\newcommand{\onode}{\texttt{o}}

\newcommand{\unode}{\texttt{u}}

\newcommand{\uleft}{\textsf{left}}
\newcommand{\uright}{\textsf{right}}

\newcommand{\maxreset}{\text{max-r}}

\newcommand{\maxset}{\textsc{MS}}
\newcommand{\gmaxset}[1]{\textsc{M}_{#1}}
\newcommand{\ulist}{\texttt{ul}}
\newcommand{\ecs}{\mathcal{E}}
\newcommand{\entry}{\texttt{entry}}
\newcommand{\exit}{\texttt{exit}}
\newcommand{\nodes}{\texttt{nodes}}

\newcommand{\linit}{\textsf{new-ulist}}
\newcommand{\linsert}{\textsf{insert}}
\newcommand{\lmerge}{\textsf{merge}}

\newcommand{\odepth}{\operatorname{odepth}}

\newcommand{\EvalProblem}{\textsc{TimedCEAEvaluation}}

\newcommand{\dbottom}{\textsf{new-bottom}}
\newcommand{\dunion}{\textsf{union}}
\newcommand{\dclockcheck}{\textsf{add-clock-check}}
\newcommand{\dreset}{\textsf{add-reset}}
\newcommand{\dextend}{\textsf{extend}}
\newcommand{\createbottom}{\textsf{create-bottom-node}}
\newcommand{\createunion}{\textsf{create-union-node}}
\newcommand{\createunions}{\textsf{create-union-nodes}}

\newcommand{\createextend}{\textsf{create-extend-node}}
\newcommand{\createempty}{\textsf{create-empty-node}}
\newcommand{\getgadget}{\textsf{get-gadget}}
\newcommand{\createclockcheckgadget}{\textsf{create-clock-check-gadget}}
\newcommand{\swap}{\textsf{swap}}
\newcommand{\createresetgadget}{\textsf{create-reset-gadget}}
\newcommand{\mergegadgets}{\textsf{merge-gadgets}}

\newcommand{\disbottom}{\textsf{is-bottom}}
\newcommand{\disunion}{\textsf{is-union}}

\newcommand{\disextend}{\textsf{is-extend}}
\newcommand{\disempty}{\textsf{is-empty}}

\newcommand{\EnumCAECS}{\texttt{EnumCAECS}}
\newcommand{\dleft}{\ell}
\newcommand{\dright}{r}

\newcommand{\newul}{\textsf{new-union-list}}
\newcommand{\insertul}{\textsf{insert}}
\newcommand{\mergeul}{\textsf{merge}}
\newcommand{\ulclockcheck}{\textsf{ul-clock-check}}
\newcommand{\ulreset}{\textsf{ul-reset}}

\newcommand{\yield}{\textsf{yield}}
\newcommand{\EOS}{\textsf{EOS}}
\newcommand{\DS}{\textsf{DS}}
\newcommand{\ckeys}{\textsf{keys}}
\newcommand{\cordkeys}{\textsf{ordered-keys}}
\newcommand{\hasnext}{\textrm{hasnext}}
\newcommand{\ustack}{\textrm{St}}

\newcommand{\bracketnode}[1]{$\displaystyle\left[\begin{matrix}#1\end{matrix}\right]$}

\renewcommand{\paragraph}[1]{\smallskip
	\noindent\textbf{#1}.}

\renewcommand{\operatorname}[1]{\mathsf{#1}}

\algrenewcommand\algorithmicindent{1.3em}%
\algnewcommand\algorithmicforeach{\textbf{for each}}
\algdef{S}[FOR]{ForEach}[1]{\algorithmicforeach\ #1\ \algorithmicdo}

\theoremstyle{plain}
\newtheorem{theorem}{Theorem}[section]
\newtheorem*{theorem*}{Theorem}
\newtheorem{proposition}[theorem]{Proposition}
\newtheorem*{proposition*}{Proposition}
\newtheorem{lemma}[theorem]{Lemma}
\newtheorem*{lemma*}{Lemma}

\newtheorem*{corollary*}{Corollary}

\newtheorem*{claim*}{Claim}

\theoremstyle{definition}

\newtheorem*{definition*}{Definition}

\theoremstyle{remark}

\newtheorem*{observation*}{Observation}
\newtheorem{example}[theorem]{Example}
\newtheorem*{example*}{Example}

\title{Complex event recognition under time constraints: \\ towards a formal framework for efficient query evaluation}

\usepackage{authblk}

\author[1]{Julián García}
\author[1]{Cristian Riveros}
\affil[1]{Pontificia Universidad Católica de Chile, \texttt{jgarcg@uc.cl}, \texttt{cristian.riveros@uc.cl}}

\date{} 

\begin{document}
	
	\maketitle

\begin{abstract}
Complex Event Recognition (CER) establishes a relevant solution for processing streams of events, giving users timely information. CER systems detect patterns in real-time, producing complex events and generating responses to them. In these tasks, time is a first-class citizen. Indeed, the time-sequence model distinguishes CER from other solutions, like data stream management systems. Surprisingly, until now, time constraints are usually included in CER query languages and models in a restricted way, and we still lack an understanding of the expressiveness and of efficient algorithms concerning this crucial feature: time.

This work studies CER under time constraints regarding its query language, computational models, and streaming evaluation algorithms. We start by introducing an extension of Complex Event Logic (CEL), called timed CEL, with simple time operators. We show that timed CEL aids in modeling CER query languages in practice, serving as a proxy to study the expressive power of such languages under time constraints.  
For this purpose, we introduce an automata model for studying timed CEL, called timed Complex Event Automata (timed CEA). This model extends the existing CEA model with clocks, combining CEA and timed automata in a single model. We show that timed CEL and timed CEA are equally expressive, giving the first characterization of CER query languages under time constraints. Then, we move towards understanding the efficient evaluation of timed CEA over streams concerning its determinization and efficient algorithms. We present a class of timed CEA that are closed under determinization; furthermore, we show that this class contains swg-queries, an expressive class of CER queries recently introduced by Kleest-Meißner et al. Finally, we present a streaming evaluation algorithm with constant update time and output-linear delay for evaluating deterministic monotonic timed CEA with a single clock, which have only less equal or greater equal comparisons.

\end{abstract} 	

	\section{Introduction}\label{sec:introduction}

Complex Event Processing (CEP) forms a set of solutions and technologies for processing data streams in real-time~\cite{cugola2012processing,GiatrakosAADG20}. CEP systems model their data sources as an infinite sequence of real-time events coming, for example, from sensor measurements, traffic reports, weather events, or social media messages. Complex Event Recognition (CER) is a subtask in CEP that processes event streams in real-time by detecting patterns, producing complex events, and generating responses to them as the events arrive. Users define these patterns through a structured query language that allows them to specify situations of interest in the real world.

Since data represents real-world events, \emph{time} is a first-class citizen in CER. As stated in~\cite{UgarteV18,cugola2012processing}, the time-sequence model is what distinguishes CER from data stream management systems, both in terms of query languages and evaluation algorithms. As early as the first CER systems~\cite{luckham1995event}, their query languages recognized the importance of time, incorporating special operators to include time constraints. Indeed, all CER systems include the \texttt{WITHIN} clause to specify a time window where the pattern must appear.
In other words, supporting time restrictions is a crucial feature of any CER query language today.

Surprisingly, until now, time constraints are usually included in CER query languages and models in a restricted way~\cite{cugola2012processing}. As a proof of fact, the \texttt{WITHIN} clause is included in all CER query languages but only for specifying a restriction over the overall pattern. In particular, one cannot use this operator internally or nested it inside a pattern. 
Although there are some proposals for more advanced time operators in CER query languages~\cite{li2005composite,GiatrakosAADG20}, there is no formal logic or computational model that helps us understand the expressiveness and efficient algorithms concerning this crucial feature: time.

This paper studies CER under time constraints regarding its query language, computational model, and streaming evaluation algorithms. 
Our first goal is to find a logic that serves as a proxy to study the expressive power of CER query languages under time constraints (Section~\ref{sec:timed-cel}).  
We propose to extend Complex Event Logic (CEL) with time windows and time operators, which we call \emph{timed CEL}. We show that timed CEL aids in modeling CER query languages, including most operators used in practice for declaring time constraints.

To understand the expressiveness of timed CEL, we look for the right automata model to express time constraints in CER (Section~\ref{sec:timed-cea}). Timed automata is an extension of finite state automata with clocks that has a rich theory and has been successful for formal verification of real-time systems. Our strategy is simple: to bridge the theory of timed automata with CER.
We introduce \emph{timed Complex Event Automata} (timed CEA), a model that extends the existing CER automata (called CEA) with clocks, combining CEA and timed automata in a single model. We show that timed CEA is the right model for CER by showing that timed CEL and timed CEA are equally expressive. To the best of our knowledge, this is the first characterization of the expressive power of CER query languages under time constraints.

With a logic and automata model for processing CER queries with time constraints, we move towards finding efficient evaluation algorithms for them. We start by studying the determinization of timed CEA, namely, to convert a timed CEA into a deterministic one (Section~\ref{sec:determinization}). This problem is crucial for evaluation since efficient enumeration algorithms over automata models usually require determinizing them to remove duplicates. Unfortunately, not all timed CEA admit an equivalent deterministic one, and thus, we look for classes of timed CEA that we can determinize efficiently. On this line, we present the class of \emph{synchronous resets} timed CEA, which is simple and easy to determinize. Furthermore, we show that this class contains swg-queries, an expressive class of CER queries recently introduced by Kleest-Meißner et al~\cite{Kleest-Meissner22}.

Having a class of deterministic timed CEA, we study algorithms to evaluate them (Section~\ref{sec:evaluation}). On this line, we present a streaming evaluation algorithm with constant update time and output-linear delay for evaluating monotonic deterministic timed CEA with a single clock, which have only less equal or greater equal comparisons. 

\paragraph{Related work} Here, we do a general review of the relation of this work to CER and timed automata. We discuss more connections throughout the paper whenever they are more suitable.

As stated above, time is a special feature in CER, and all CER systems include event timestamps in their models. Windows is the most common operator for declaring time constraints in query languages~\cite{cugola2012processing,GiatrakosAADG20}. Further, people have proposed efficient algorithms for evaluating queries over a sliding window~\cite{TangwongsanHS21,BucchiGQRV22}. Other proposals used timed restrictions between pairs of events~\cite{li2005composite,Kleest-Meissner22}. Nevertheless, we are not aware of a work in CER that has studied the expressiveness of such query languages, or efficient algorithms for evaluating queries with general time constraints. 

Timed automata~\cite{AlurD94} have a rich theory and applications in formal verification. This work is based on the theory of timed automata making a bridge with CER and, more generally, with data stream management systems. Logic for timed automata and real-time systems has been proposed in the past~\cite{koymans1990specifying,AsarinCM2002}. However, these logics are designed for formal verification of systems, rather than for streaming processing. In particular, they have not been used for CER: they are not defined to produce outputs (i.e., complex events) or studied for efficient evaluation over streams. Some works~\cite{chakraborty2005event,boubeta2017analysis} have used timed automata in CER; however, they did not study the same research questions as this work regarding expressive analysis, and efficient query answering.

	\section{Preliminaries}\label{sec:preliminaries}

\paragraph{Sets and intervals}  Given a set $A$, we denote by $\pow{A}$ the \emph{powerset} of $A$. We denote by $\bbN$ the natural numbers. Given $n,m \in \bbN$ with $n \leq m$, we denote by $[n]$ the set $\{1, \ldots, n\}$ and by $[n..m]$ the interval $\{n, n+1, \ldots, m\}$ over $\bbN$. 
We denote by $\bbQplus$ the non-negative rational numbers, and by $\bbQinf$ the set $\bbQplus \cup \{\infty\}$ where $q < \infty$ for every $q \in \bbQplus$. 
A subset $I \subseteq \bbQplus$ is an \emph{interval over $\bbQplus$} if for every $p, q \in I$ and $r \in \bbQplus$, if $p \leq r \leq q$, then $r \in I$.
We will use the standard notation for specifying intervals over $\bbQplus$, namely, as pairs in $\bbQplus \times \bbQinf$ with squared or circled brackets. For example, $(\frac{1}{2}, 3] = \{p \in \bbQplus \mid \frac{1}{2} < p \leq 3\}$, $(1, 2) = \{p  \in \bbQplus \mid 1 < p < 2\}$, and $[7, \infty) = \{p \in \bbQplus \mid 7 \leq p < \infty\}$ are intervals over $\bbQplus$ under this notation.

\paragraph{Events and streams} We fix a set $\tset$ of \emph{event types}, a set $\aset$ of \emph{attributes names}, and a set $\dset$ of \emph{data values} (e.g., integers, floats, strings). An \emph{event} $e$ is a partial mapping $e: \aset \rightarrow \dset$ that maps attributes names in $\aset$ to data values in $\dset$. We denote $\eatt{e}$ the domain of $e$, called the \emph{attributes of $e$}, and we assume that $\eatt{e}$ is finite. 
We denote by $e(A)$ the data value of the attribute $A \in \aset$ whenever $A \in \eatt{e}$. Further, each event $e$ has a \emph{type} in $\tset$ denoted by~$\etype{e}$. We define the \emph{size} $|e|$ of an event $e$ as the number of RAM registers to encode $e$. We write $\eset$ to denote the \emph{set of all events} over event types $\tset$, attributes names $\aset$, and data values $\dset$.
A \emph{stream} is an (arbitrary long) sequence $\Stream = e_1 e_2 \ldots e_{n}$ of events where $|\Stream| = n$ is the length of the stream. 

\begin{example} \label{ex:stream}
	As a running example (from~\cite[Section 3]{GrezRUV21}), suppose the scenario of a national park sensor network measuring the temperature and humidity of the park. For the sake of simplification, assume that there are two types of events: $T$ and $H$. A \emph{temperature event} $e$ has $\etype{e} = T$ and $\eatt{e} = \{temp\}$ where $e(temp)$ contains the temperature value in  degrees Celsius (e.g., $e(temp) = \text{20°}$). A \emph{humidity event} $e$ has $\etype{e} = H$ and $\eatt{e} = \{hum\}$ where $e(hum)$ contains the value as a percentage (e.g., $e(hum) = \text{40\%}$). In Figure~\ref{fig:stream}, we display a stream $\StreamEx$ for this scenario.
\end{example}

\begin{figure}
	\centering{
		{\small
			\begin{tabular}{|l|c|c|c|c|c|c|c|c|c|c}\hline
				type  &$H$&$T$&$H$&$H$&$T$&$T$&$T$&$H$&$H$ & \dots \\ \hline
				$temp$  &  & 45° &  &  & 40° & 42° & 25° & &  & \multirow{2}{*}{\dots} \\
				$hum$ & 25\% & & 20\%& 25\%& & & & 70\% & 18\%\\ \hline
				time (sec)  & 1.2 & 1.33 & 2.5 & 3.7 & 4.5 & 5.3 & 5.9 & 6.1 & 7.2 & \dots \\ \hline
				\multicolumn{1}{c}{} &
				\multicolumn{1}{c}{1} &  
				\multicolumn{1}{c}{2} &
				\multicolumn{1}{c}{3} &
				\multicolumn{1}{c}{4} &
				\multicolumn{1}{c}{5} &
				\multicolumn{1}{c}{6} &
				\multicolumn{1}{c}{7} &
				\multicolumn{1}{c}{8} & 
				\multicolumn{1}{c}{9} & 
				\multicolumn{1}{c}{\dots} 
		\end{tabular}}
		\caption{A timed stream $\StreamEx$ of temperatures ($T$) and humidities ($H$) with attributes $temp$ (degrees Celcius) and $hum$ (percentage), respectively. The timestamp is the third line, measured in seconds. The fourth line is the index position in the stream.}\label{fig:stream}
	}
\end{figure}

\paragraph{Complex events} Fix a finite set $\xset$ of \emph{variables} and assume that $\tset \subseteq \xset$, where $\tset$ is the set of event types as defined earlier. Let $\Stream$ be a stream of length~$n$. A \emph{complex event}\footnote{In~\cite{GrezRUV21,BucchiGQRV22},triples of the form $(i, j, \mu)$ are called valuations and complex events are triples $(i, j, D)$ where $D \subseteq [i..j]$. In this work, we prefer to call $(i, j, \mu)$ as complex events to simplify the presentation.} of $\Stream$ is a triple $(i, j, \mu)$ where $i, j \in [n]$, $i \leq j$, and $\mu:\xset \rightarrow \pow{[i..j]}$ is a function with finite domain. Intuitively, $i$ and $j$ marks the beginning and end of the interval where the complex event happens, and $\mu$ stores the events in the interval $[i..j]$ that fired the complex event.
In the following, we will usually use $C$ to denote a complex event $(i, j, \mu)$ of $\Stream$ and omit $\Stream$ if the stream is clear from the context. We will use $\ctime(C)$, $\cstart(C)$, and $\cend(C)$ to denote the interval $[i..j]$, the start $i$, and the end $j$ of $C$, respectively. Further, by some abuse of notation we will also use $C(X)$ for $X \in \xset$ to denote the set $\mu(X)$ of $C$.  

\begin{example}\label{ex:complex-events}
	Complex events of the stream $\StreamEx$ of Figure~\ref{fig:stream} can be $C_1=(5,9, \{X \mapsto \{5\}, Y \mapsto \{9\}, H \mapsto \varnothing, T \mapsto \varnothing\})$ and  $C_2=(2,9, \{X \mapsto \{2\}, Y \mapsto \{9\}, H \mapsto \varnothing, T \mapsto \varnothing\})$ where $X$ and $Y$ are variables in $\xset$.
\end{example}

The following operations on complex events will be useful throughout the paper. We define the \emph{union} of complex events $C_1$ and $C_2$, denoted by $C_1 \cup C_2$, as the complex event $C'$ such that $\cstart(C') = \min\{\cstart(C_1), \cstart(C_2)\}$, $\cend(C') = \max\{\cend(C_1), \cend(C_2)\}$, and $C'(X) = C_1(X) \cup C_2(X)$ for every $X \in \xset$. Further, we define the \emph{projection over $L$} of a complex event $C$, denoted by $\pi_L(C)$, as the complex event $C'$ such that $\ctime(C') = \ctime(C)$ and $C'(X) = C(X)$ whenever $X \in L$, and $C'(X) = \emptyset$, otherwise. Finally, we denote by $(i,j, \mu_\emptyset)$ the complex event with trivial mapping $\mu_\emptyset$ such that $\mu_\emptyset(X) = \emptyset$ for every $X \in \xset$.

\paragraph{Predicates of events} A \emph{predicate} is a possibly infinite set $P$ of events. 
For instance (see Example~\ref{ex:stream}), $P$ could be the set of all $T$-events $e$ such that $e(temp) \leq \text{20°}$.
In our examples, we will use the notation $A \sim a$ where $A \in \aset$ and $a \in \dset$ to denote the predicate $P = \{e \mid e(A) \sim a\}$ (e.g., $temp \leq \text{20°}$).
We say that an event $e$ satisfies predicate $P$, denoted $e \models P$, if, and only if, $e \in P$.
We generalize this notation from events to a set of events $E$ such that $E \models P$ if, and only if, $e \models P$ for every~$e \in E$. 

We assume a fixed \emph{set of predicates} $\pset$. Further, we assume that there is a basic set of predicates $\psetbasic \subseteq \pset$ (e.g., $temp \leq \text{20°}$) and $\pset$ is the closure of $\psetbasic$ under intersection and negation (i.e., $P_1 \cap P_2 \in \pset$ and $\eset \setminus P \in \pset$ for every $P, P_1, P_2 \in \pset$) where $\eset$ is a predicate in $\pset$, that we usually denote by $\true$. 
For every $P \in \pset$, we define a size of a predicate $|P|$ as follows: $|P| = 1$ if $P \in \psetbasic$ (i.e., constant size), $|P| = |P_1| + |P_2|$ if $P = P_1 \cap P_2$, and $|P| = |P'|+1$ if $P = \eset \setminus P'$.
For every $P \in \pset$, we assume that the time to check if $e \models P$ is in $\mathcal{O}(|e|)$.

\paragraph{Complex event logic} In this work, we use the \emph{Complex Event Logic} (CEL) introduced in~\cite{GrezRUV21} and implemented in CORE~\cite{BucchiGQRV22} as our basic query language for CER. The syntax of a CEL formula $\varphi$ is given by the grammar:
\[
	\varphi \ := \  R  \ \mid \ \varphi \kAS X \ \mid \ \varphi \kFILTER X[P] \ \mid \   \varphi \kOR \varphi \mid \ \varphi \kAND \varphi \ \mid \ \varphi \sq \varphi \ \mid \ \varphi \ssq \varphi \ \mid \ \varphi\ks \ \mid \ \varphi\sks \ \mid \ \pi_L(\varphi)
\]
where $R \in \tset$ is an event type, $X \in \xset$ is a variable, $P \in \pset$ is a predicate, and $L \subseteq \xset$ is a set of variables.  
We define the semantics of a CEL formula $\varphi$ over a stream $\Stream$, recursively, as a set of complex events over $\Stream$. In Figure~\ref{fig:cel-semantics}, we define the semantics of each CEL operator like in~\cite{BucchiGQRV22,GrezRUV21}. 
\begin{figure}[t]
	\small
	\centering
	\begin{align*}
		\sem{R}(\Stream) & \ = \ \{ (i,i, \mu) \ \mid \!\!
		\begin{array}[t]{l}
			\type(e_i) = R \ \wedge \ \mu(R) = \{i\} \ \wedge \  \forall X \neq R. \ \mu(X) = \emptyset\} 
		\end{array}  \\ 
		\sem{\varphi \kAS X}(\Stream) & \  = \  \{ C \ \mid \!\!
		\begin{array}[t]{l}
			\exists \, C' \in \sem{\varphi}(\Stream). \ \ \ctime(C) = \ctime(C') \ \wedge \ C(X) = \bigcup_{Y} C'(Y)   \\
			\wedge \ \forall Z \neq X. \  C(Z) = C'(Z) \} 
		\end{array} \\
		\sem{\varphi \kFILTER X[P]}(\Stream) & \ = \  \{ C \ \mid \!\!
		\begin{array}[t]{l}
			C \in \sem{\varphi}(\Stream)  \ \wedge \ C(X) \models P \} 
		\end{array} \\
		\sem{\varphi_1 \kOR \varphi_2}(\Stream) & \ = \ \sem{\varphi_1}(\Stream) \ \cup \  \sem{\varphi_2}(\Stream) \\
		\sem{\varphi_1 \kAND \varphi_2}(\Stream) & \ = \ \sem{\varphi_1}(\Stream) \ \cap \  \sem{\varphi_2}(\Stream) \\
		\sem{\varphi_1 \sq \varphi_2}(\Stream) & \ = \ \{ C_1 \cup C_2 \ \mid \!\!
		\begin{array}[t]{l}
			C_1 \in \sem{\varphi_1}(\Stream) \, \wedge \, C_2 \in \sem{\varphi_2}(\Stream) \, \wedge \, \cend(C_1) < \cstart(C_2)   \}
		\end{array}  \\
		\sem{\varphi_1 \ssq \varphi_2}(\Stream) & \ = \ \{ C_1 \cup C_2 \ \mid \!\!
		\begin{array}[t]{l}
			C_1 \in \sem{\varphi_1}(\Stream) \, \wedge \,  C_2 \in \sem{\varphi_2}(\Stream) \, \wedge \,
			\cend(C_1) + 1 = \cstart(C_2) \}
		\end{array}  \\
		\sem{\varphi\ks}(\Stream) & \ = \ \sem{\varphi}(\Stream) \ \cup \ \sem{\varphi \sq \varphi \ks}(\Stream) \\
		\sem{\varphi\sks}(\Stream) & \ = \ \sem{\varphi}(\Stream) \ \cup \ \sem{\varphi \ssq \varphi \sks}(\Stream) \\
		\sem{\PROJ{L}\!(\varphi)}(\Stream) & \ = \ \{ \pi_L(C) \ \mid \!\!
		\begin{array}[t]{l}
			C \in \sem{\varphi}(\Stream) \}
		\end{array}
	\end{align*}
	\vspace{-5mm}
	\caption{The semantics of CEL formulas defined over a stream $\Stream = e_1 e_2 \ldots e_n$ where each $e_i$ is an event.
	}
	\label{fig:cel-semantics}
\end{figure}
\begin{example}\label{ex:cel-query}
Coming back to Examples~\ref{ex:stream} and~\ref{ex:complex-events}, suppose we want to detect a temperature event $X$ above 40° followed by a humidity event $Y$ of less than 25\%, which represents a fire with high probability. In addition, we want to check that there is some $T$-event $e$ between $X$ and $Y$ with a temperature above 40° that confirms the high temperature (e.g., temperature measures are faulty in the sensor network). The following CEL formula defines this complex event:
\[
\varphi_1 \ := \  \pi_{X,Y}\big( (T \kAS X; T; H \kAS Y) \kFILTER (T[temp > \text{40°}] \wedge H[hum < \text{25\%}]) \big) 
\] 
Note that we use variables $X$ and $Y$ to mark the first temperature and the last humidity, respectively, removing the middle event $T$ from the output. 
In other words, we only want to check the existence of such $T$-event but do not want it in the final output. 
One can check that $C_1, C_2 \in \sem{\varphi_1}(\StreamEx)$ among others complex events.
\end{example}

	\section{Complex event logic under time constraints}\label{sec:timed-cel}

In this section, we introduced \emph{timed CEL}, an extension of CEL with time operators for modeling time constraints in complex event recognition (CER). 
We show how to extend the setting and CEL for modeling these constraints in practice. Further, we show through examples that these operators include most operators introduced in previous works.

\paragraph{Timed complex event logic} We first need to extend each event in a stream with a timestamp to include time constraints. Intuitively, the timestamp of an event models how time passes in the systems when events arrive continuously. These timestamps could represent when the event happened, when the event was measured, or when the event arrived in the system. In any possible scenario, we assume that these timestamps strictly increase with the arrival of events.\footnote{In practice, this assumption is not always valid since two events could arrive with the same timestamp or not necessarily in increasing order. Although these are relevant scenarios, we leave them for future work.}

Formally, recall that $\tset$ is a fixed set of event types, $\aset$ is a fixed set of attributes names, and $\dset$ is a fixed set of data values for our events $e:\aset \rightarrow \dset$ with types in $\tset$. A \emph{timed stream} is an (arbitrary long sequence) of pairs:
\[
\Stream \ = \ (e_1, t_1) (e_2, t_2) \ldots (e_n, t_n)
\] 
where $e_i$ is an event and $t_i \in \bbQplus$ is its timestamp for every $i \in \bbN$. We will usually call a pair $(e_i, t_i)$ a \emph{timed event} where $e_i$ is its event and $t_i$ is its timestamp. Further, we assume that timestamps are strictly increasing, namely, $t_1 < t_2 < \ldots < t_n$. The stream $\StreamEx$ of Figure~\ref{fig:stream} is also a timed stream where the last line displayed the timestamp of each event (in seconds).

Note that, although we could have modeled timestamps with a special attribute ``$\operatorname{ts}$'' such that $e(\operatorname{ts}) = t_i$, we preferred to model them as an additional component (i.e., as $(e_i, t_i)$). The reason is that timestamps are first citizens in this framework, and we wanted to clearly distinguish between the data (i.e., the events) and the time, which is constantly increasing. 

We define \emph{timed Complex Event Logic} (timed CEL for short) as an extension of CEL with additional operators for managing timestamps in time stream. 
Specifically, a timed CEL formula $\varphi$ follows the same grammar as a CEL formula %
where we additionally add the following syntax~rules:
\[
\begin{array}{rcll}
	\varphi & := &  \twin{\varphi}{I} & \text{(time window)} \\
	& \mid & \varphi \timesq{I} \varphi & \text{(time non-contiguous sequencing)}  \\
	& \mid & \varphi \timessq{I} \varphi & \text{(time contiguous sequencing)}  \\
	& \mid & \varphi \timeks{I} & \text{(time non-contiguous iteration)}  \\
	& \mid & \varphi \timesks{I} \ \ \ \ \ \ \ \ \ \  & \text{(time contiguous iteration)} 
\end{array}
\]
where $I$ is an interval over $\bbQplus$. We call the above operators the \emph{time operators} of timed CEL, whereas we call the other standard operators just \emph{CEL operators}. 
The previous syntax extends sequencing and iteration operators in CEL (both in their contiguous and non-contiguous versions) with a corresponding time-variant, where we additionally check a time constraint between complex events. Instead, $\twin{\varphi}{I}$ is a natural extension of time windows in CER~\cite{cugola2012processing,GiatrakosAADG20}. 
\begin{figure}[t]
	\small
	\centering
	\begin{align*}
		\sem{\twin{\varphi}{I}}(\Stream) & \ = \ \{ C \ \mid \!\!
		\begin{array}[t]{l}
			C \in \sem{\varphi}(\Stream) \ \wedge \ t_{\cend(C)} - t_{\cstart(C)} \in  I \\ 
		\end{array} \}\\
		\sem{\varphi_1 \timesq{I} \varphi_2}(\Stream) & \ = \ \{ C_1 \cup C_2 \ \mid \!\!
		\begin{array}[t]{l}
			C_1 \in \sem{\varphi_1}(\Stream) \,\wedge\, C_2 \in \sem{\varphi_2}(\Stream)  \\
			\wedge \, \cend(C_1) < \cstart(C_2) \, \wedge \, t_{\cstart(C_2)} - t_{\cend(C_1)} \in I \}
		\end{array}  \\
		\sem{\varphi_1 \timessq{I} \varphi_2}(\Stream) & \ = \ \{ C_1 \cup C_2 \ \mid \!\!
		\begin{array}[t]{l}
			C_1 \in \sem{\varphi_1}(\Stream) \,\wedge\, C_2 \in \sem{\varphi_2}(\Stream) \\
			\wedge \, \cend(C_1) +1 =  \cstart(C_2) \, \wedge \, t_{\cstart(C_2)} - t_{\cend(C_1)} \in I \}
		\end{array}  \\
		\sem{\varphi\timeks{I}}(\Stream) & \ = \ \sem{\varphi}(\Stream) \ \cup \ \sem{\varphi \timesq{I} (\varphi \timeks{I})}(\Stream) \\
		\sem{\varphi\timesks{I}}(\Stream) & \ = \ \sem{\varphi}(\Stream) \ \cup \ \sem{\varphi \timessq{I} (\varphi \timesks{I})}(\Stream) 
	\end{align*}
	\vspace{-5mm}
	\caption{The semantics of time operators of timed CEL for every timed CEL formulas $\varphi, \varphi_1, \varphi_2$, a timed stream $\Stream = (e_1, t_1) \ldots (e_n, t_n)$, and $I$ is an interval over $\bbQplus$. Here, recall that $t_{\cstart(C)}$ is the timestamp of the event in index $\cstart(C)$. 
	}
	\label{tab-time-semantics}
\end{figure}

Timed CEL is an extension of CEL and, as such, we can naturally extend the semantics of CEL operators from streams to timed streams where we just omit the timestamp of timed events (see Figure~\ref{fig:cel-semantics} in Section~\ref{sec:preliminaries}). For the time operators of timed CEL, we define their semantics recursively in Figure~\ref{tab-time-semantics}. %
In the sequel, we will use the notation $\sim\! c$ in timed CEL formulas to specify an interval $I = \{p \in \bbQplus \mid p \sim c\}$ over $\bbQplus$ for any $c \in \bbQplus$ and $\sim \, \in \{\leq, \geq, <, >, =\}$. For example, we write $\twin{A\timesq{> 2} B}{\leq 5}$ to denote the formula $\twin{A\timesq{(2, \infty)} B}{[0, 5]}$. Next, we show some examples on how to use and compose the time operators of timed CEL to define relevant patterns over timed streams.

\begin{example}[Continued Example~\ref{ex:cel-query}] \label{ex:timed-cel1}
	For finding relevant complex events in a fire scenario, we also want to restrict that the time difference between the high temperature $X$ and low humidity $Y$ is less than five seconds. Further, we expect that the confirmation of the high temperature (i.e., the middle event $T$) occurs in less than one second from $X$. We can define these time restrictions in timed CEL as follows:
	\[
	\varphi_1' \ := \  \pi_{X,Y}\big( \twin{T \kAS X\timesq{\leq 1} T; H \kAS Y}{\leq 5} \kFILTER (T[temp > \text{40°}] \wedge H[hum < \text{25\%}]) \big) 
	\] 
	We can check that now $C_1 \in \sem{\varphi_1'}(\StreamEx)$ but $C_2 \notin \sem{\varphi_1'}(\StreamEx)$. 
\end{example}
\begin{example} \label{ex:timed-cel2}
As an example of using time iteration, suppose that we want to see how temperature changes when humidity increases from less than 30\% to over 30\%~\cite[Ex.3]{GrezRUV21}. Further, the time difference between consecutive events must be less than one second. We can define this query as:
\[
\varphi_2 \ := \ \pi_{X,Y,T}\big(\,(H \kAS X \timessq{\leq 1} (T \timesks{\leq 1}) \timessq{\leq 1} H \kAS Y) \kFILTER (X[hum < 30\%] \wedge Y[hum > 30\%]) \,\big)
\]
Note that we use the contiguous sequencing and iteration to capture all consecutive events between $X$ and $Y$. Furthermore, time iteration is useful to impose a repetitive time restriction over a sequence of contiguous events. One can check that $(4, 8, \{X \mapsto \{4\}, Y \mapsto \{8\}, T \mapsto \{5,6,7\}, H \mapsto \varnothing\}) \in \sem{\varphi_2}(\StreamEx)$.
\end{example}

\paragraph{Timed CEL in practice} Arguably, the operators introduced in timed CEL cover most of the use of time restrictions in CER. Below, we discuss how timed CEL can model previous proposals, dividing them into three groups. Our goal is not to be exhaustive but to show how different flavors of timed constraints appear in CER and coincide with timed CEL.

\smallskip
\noindent \emph{Time windows.} Most previous CER query languages introduce time restriction as a form of a sliding window~\cite{GiatrakosAADG20,verwiebe2023survey}, declared by a \texttt{WITHIN} clause. For example, consider the SASE+ query in~\cite{agrawal2008efficient}:
	\begin{center}
		\begin{tabular}{l}
			\texttt{PATTERN SEQ(Shelf a, $\sim$(Register b), Exit c)} \\
			\texttt{WHERE skip\_till\_next\_match(a, b, c) \{ [tag\_id] \} } \\
			\texttt{WITHIN 12 hours}
		\end{tabular}
	\end{center}
	Here, the \texttt{WITHIN} clause forces that the pattern must appear in a time window of the last 12 hours. Similarly, several systems use the \texttt{WITHIN} clause (see, e.g., \cite{PoppeLRM17,MaLPR22,BucchiGQRV22,cugola2012complex}). Specifically, if a user wants to find the pattern $\varphi$ between $c$ hours (or any time scale) for some value $c$, we can express it in timed CEL as $\twin{\varphi}{\leq c}$.
	
	Notice that we limit the above discussion to use time-based windows, contrary to count-based windows where the \texttt{WITHIN} clause specifies how many events must be inside the windows (e.g., \texttt{WITHIN 100 events}). For count-based windows, one can easily use the AND operator to declare that the matched substream must have some specific length. Also, we compared timed CEL with operators in CER for imposing time restrictions over a pattern. Specifically, CER systems and, in general, stream data management systems also use window operators to restrict the amount of data where it must evaluate the query, which we call here a \emph{data window}. For example, \emph{tumbling windows} is a data window that splits the stream into consecutive sequences of equal size (i.e., the same number of events), and the query is evaluated over each sequence. We refer the reader to~\cite{verwiebe2023survey} that surveys several data windows types used in stream processing systems. Of course, timed CEL cannot necessarily express data windows, given that the time operators are designed to impose restrictions related to time and not associated with data. 
	
\smallskip	
\noindent \emph{Time sequencing and time iteration.} These operators are less common in the literature. One example of a system that uses time sequencing is Padres~\cite{li2005composite} which uses $\ctimesq{I}$ for checking a time constraint with two relevant events (not necessarily contiguous). As an illustration, consider the following (simplified) example from~\cite{cugola2012processing}:
\[
\texttt{(B(Y=10);}_{\texttt{[timespan:5]}} \  \texttt{C(Z<5))}_{\texttt{[within:15]}}
\]
Intuitively, this query in Padres means that the user wants to find an event of type $\texttt{B}$ with attribute $\texttt{Y=10}$, followed by an event of type $\texttt{C}$ with attribute $\texttt{Z<5}$. Further, the operator $\texttt{;}_{\texttt{[timespan:5]}}$ restricts that the timespan between $B$ and $C$ has to be greater than $5$ seconds, and $\texttt{(...)}_{\texttt{[within:15]}}$ to be lower than $15$ seconds. This query is equivalent to $\twin{(B \kFILTER B[Y\!=\!10]) \timesq{\geq 5} (C \kFILTER C[Z\!<\!5])}{\leq 15}$.

To the best of our knowledge, we are not aware of CER systems that use time iteration like $\timeks{I}$ or $\timesks{I}$. Probably, the reason behind this limitation is that CER systems struggle with iteration (i.e., without time) and, usually, they provide a limited use of this operator~\cite{DemersGHRW06,agrawal2008efficient}. One can see time iteration as the natural generalization of time sequencing, and for this reason, it is natural to add them as a primary operator of timed CEL. 

\smallskip
\noindent \emph{Combining all time operators.} We are not aware of a CER system or CER query language with the same expressibility as timed CEL. Still, there are query languages that combine time windows and time sequencing. One example is Padres~\cite{li2005composite}, which we previously discussed above. Another more recent proposal is \emph{subsequence-queries with wildcards and gap-size constraints} (swg-queries for short) introduced in~\cite{Kleest-Meissner22}.
A swg-query $Q$ is a triple $(s, w, c)$ where $s = P_1 \ldots P_k$ is a sequence of predicates (called \emph{types} in~\cite{Kleest-Meissner22}), $w \in \bbQplus \cup \{\infty\}$ is a global window size, and $c = (I_1, \ldots, I_{k-1})$ is a tuple of local gap-size constraints where each $I_i$ is an interval over $\bbQplus$ for every $i \in [k-1]$.
A time stream $\Stream = (e_1, t_1) (e_2, t_2) \ldots (e_n, t_n)$ satisfies $Q$ if we can find a subsequence  $(e_1', t_1') \ldots (e_k', t_k')$ of $\Stream$ such that $e_i' \models P_i$, the total time is less than $w$ (i.e., $t_k - t_1 \leq w$), and each consecutive pair of timed events $(e_i', t_i')$ and $(e_{i+1}', t_{i+1}')$ satisfies the local gap-size constraint $I_i$ (i.e., $t_{i+1}' - t_{i}' \in I_i$).
The matches of $Q$ over $\Stream$ are all subsequences $(e_1', t_1') \ldots (e_k', t_k')$ that witness that $\Stream$ satisfies $Q$.

We must note that our definition of swg-queries slightly differs from \cite{Kleest-Meissner22} to fit our purpose better. First, \cite{Kleest-Meissner22} uses window size and gap-sizes in $\bbN \cup\{\infty\}$, since they model time discretely. Instead, we present above the natural extension of swg-queries from $\bbN$ to $\bbQplus$. Second, \cite{Kleest-Meissner22} considers variables (called wildcards) in the sequence $P_1 \ldots P_k$ where each variable repetition must match the same type of event. This feature corresponds to some form of a join operation (called correlation in CER), a feature that we do not include in this work for the sake of simplification (see Section~\ref{sec:conclusions}). 

One can easily see that any swg query $Q = (s, w, c)$ like above can be represented in timed CEL as follows. Let $\tset = \{R_1, \ldots, R_m\}$ be the set of all event types, and define the CEL formula $\varphi_{\tset} :=  R_1 \kOR R_2 \kOR \ldots \kOR R_m$ that its true for every event. 
Then, one can easily check that the following timed CEL formula $\varphi_Q$ is equivalent to $Q$:
\begin{equation} \label{eq:swg-form}
\varphi_Q \ :=  \ \twin{(\varphi_{\tset} \kAS X_1) \kFILTER X_1[P_1]  \ \timesq{I_1} \cdots \timesq{I_{k-1}} (\varphi_{\tset} \kAS X_k) \kFILTER X_k[P_k]}{\leq w} \tag{*}
\end{equation}
Therefore, one can consider swg-queries as a (strict) subfragment of timed CEL. For this reason, in the following we call a timed CEL formula $\varphi$ a \emph{swg query} if it has the form of (\ref{eq:swg-form}).

	\section{Timed complex event automata}

To understand the expressiveness of timed CEL and design efficient algorithms, we introduce here its automata counterpart, which we call timed Complex Event Automata (timed CEA). This model is the natural combination of complex event automata~\cite{GrezRUV21,BucchiGQRV22} (CEA) and timed automata, namely, we extend CEA with clocks. 

In the following, we introduce the model of timed CEA and study its expressive power with timed CEL. Interestingly, we show that both models are equally expressive, providing evidence of the robustness of timed CEL as a logical language for CER with time constraints.
We start by first recalling the notion of CEA to introduce later timed CEA.

\paragraph{Complex Event Automata}
A \emph{Complex Event Automata} (CEA) is a tuple $\cA = (Q, \pset, \xset, \Delta, q_0, F)$ where $Q$ is a finite set of states, $\pset$ is the set of predicates, $\xset$ is a finite set of variables, $\Delta \subseteq Q \times \pset \times 2^\xset \times Q$ is a finite relation (called the transition relation), $q_0 \in Q$ is the initial state, and $F$ is the set of final states.
A \emph{run $\rho$ of $\cA$ over the stream $\Stream = e_1 e_2 \ldots e_{n}$ from position $i$ to $j$} is a sequence:
\[
\rho \ := \  p_i \xrightarrow{P_{i}, L_{i}} p_{i+1} \xrightarrow{P_{i+1}, L_{i+1}} p_{i+2} \xrightarrow{P_{i+2}, L_{i+2}} \dots \xrightarrow{P_{j}, L_{j}} p_{j+1}
\]
where $p_i = q_0$, $(p_k, P_{k}, L_{k}, p_{k+1}) \in \Delta$, and $e_k \models P_k$ for all $k \in [i..j]$. We say that the run is \emph{accepting} if $p_{j+1} \in F$. A run $\rho$ from positions $i$ to $j$ like above  defines the complex event $C_\rho = (i, j, \mu_\rho)$ such that $\mu_\rho(X) = \{k \in [i..j] \mid X \in L_k\}$ for every $X \in \xset$. Note that the starting and ending positions $i,j$ of the run define the interval of the complex event, and the labels $L_k \subseteq \xset$ define the mapping $\mu_\rho$ of~$C_\rho$. We define the set of all complex events of $\cA$ over~$\Stream$~as:
\[
\sem{\cA}(\Stream) = \{C_\rho \mid \rho \text{ is an accepting run of $\cA$ over $\Stream$}\}
\]

\paragraph{Timed CEA}
The model of timed CEA is the natural extension of CEA with \emph{clocks}, similarly as timed automata is the extension of finite automata~\cite{AlurD94}. First we introduce clock conditions to then introduce the model. 

Let $\zset$ be a finite set of \emph{clocks}. A \emph{clock condition} over $\zset$ is an expression following the grammar: 
\[
\gamma \ \  := \ \TRUE  \ \mid \ z = c \ \mid \ z < c \ \mid \ z \leq c \  \mid  \ z \geq c \ \mid \ z > c \ \mid \ \gamma \land \gamma \ \mid \ \gamma \lor \gamma
\] 
where $z \in \zset$ and $c \in \bbQplus$. We define the size $|\gamma|$ of a clock condition $\gamma$ as the number of operations in $\gamma$.
We denote by $\mathcal{C}_\zset$ the set of \emph{all clock conditions} over $\zset$. A \emph{clock valuation} $\nu : \zset \rightarrow \mathbb{R}_{\ge 0}$ is a partial mapping that assigns to a clock $z \in \zset$ a timestamp $\nu(z)$. If $z \notin \dom{\nu}$, it represents that $z$ \emph{has not been initialized}.
Given a clock condition $\gamma$ and a clock valuation $\nu$, we say that $\nu$ \emph{satisfies}~$\gamma$, denoted by $\nu \models \gamma$, if all variables in $\gamma$ are in $\dom{\nu}$ and the expression resulting by replacing every clock $z$ by $\nu(z)$ in $\gamma$ evaluates to true.  For example, the clock valuation $\nu := \{x \mapsto 5, y \mapsto \frac{1}{2}\}$ satisfies  $\gamma_1 := x \geq 4 \vee y < \frac{1}{2}$ but does not satisfy $\gamma_2 := z \leq \frac{5}{3} \wedge y \geq 1$ (i.e., $z \notin \dom{\nu}$).

The following operations on clock valuations will simplify the presentation of timed CEA.
Let $\nu$ be a clock valuation.
For $t \in \bbQplus$, we define the \emph{update of $\nu$ by $t$} as the clock valuation $\nu +t$ such that, for every $z \in \zset$,   $[\nu+t](z) = \nu(z)+t$ whenever $z \in \dom{\nu}$, and undefined, otherwise.
Furthermore, for a set $Z \subseteq \zset$ of clocks we define the \emph{reset of $\nu$ on $Z$} as the clock valuation $\treset_Z(\nu)$ such that $\dom{\treset_Z(\nu)} = \dom{\nu} \cup Z$ where $[\treset_Z(\nu)](z) = 0$ when $z \in Z$ and $[\treset_Z(\nu)](z) = \nu(z)$ when $z \in \dom{\nu} \setminus Z$. 

A \emph{timed Complex Event Automata} (timed CEA for short) is a tuple:
\[
\cT \ = \  (Q, \pset, \xset, \zset, \Delta, q_0, F)
\] 
where $Q$, $\pset$, $\xset$, $q_0$, and $F$ are exactly as for CEA. In addition, $\zset$ is the set of clocks and $\Delta \subseteq Q \times \pset \times \mathcal{C}_\zset \times 2^\xset \times 2^\zset \times Q$ is the finite transition relation. 
The two new components $\gamma$ and $Z$ in a transition $\tau = (p, P, \gamma, L, Z, q)$ represent the condition on the clocks that the automata must satisfy before firing $\tau$ and the set of clocks that are reset after firing $\tau$, respectively. We define the size $|\cT|$ of $\cT$ as the sum of the number of states and transitions, namely, $|\cT| = |Q| + \sum_{(p, P, \gamma, L, Z, q)} |\gamma| + |P| + |L| + |Z|$.

Let $\Stream \ = \ (e_1, t_1) \ldots (e_n, t_n)$ be a timed stream.
A configuration of $\cT$ over $\Stream$ is a pair $(q, \nu)$ where $q \in Q$ is the current state of $\cT$ and $\nu$ is a clock valuation representing the current values of the clocks in $\dom{\nu}$. We define the \emph{time difference at position $i \in [n]$} as $\delta t_i = t_i - t_{i-1}$ for every $i \in [n]$  where $\delta t_1 = t_1$ when $i = 1$.
Then, a \emph{run $\rho$ of $\cT$ over $\Stream$ from position $i$ to $j$} is a sequence:
\[
\rho \ := \ (p_i, \nu_i) \xrightarrow{P_{i}, \gamma_i / L_{i}, Z_i} (p_{i+1}, \nu_{i+1}) \xrightarrow{P_{i+1}, \gamma_{i+1} / L_{i+1}, Z_{i+1}} \dots \xrightarrow{P_{j}, \gamma_j / L_{j}, Z_j} (p_{j+1}, \nu_{j+1})
\] 
where $p_i = q_0$, $\nu_i$ is the \emph{trivial} clock valuation (i.e., $\dom{\nu_i} = \emptyset$), and, for every $k \in [i..j]$, the following four conditions are satisfied: (1) $(p_k, P_k, \gamma_k, L_k, Z_k, p_{k+1}) \in \Delta$, (2) $e_k \models P_k$, (3) $\nu_k + \delta t_k \models \gamma_k$, and (4) $\nu_{k+1} = \treset_{Z_k}(\nu_k + \delta t_k)$.
Conditions (1) and (2) are verbatim from CEA, (3) checks that the clock values must satisfy the clock conditions at each position, and (4) specifies that clocks in $Z_k$ are reset after updating them with the time difference at position $k$.  
We say that $\rho$ like above is an \emph{accepting run} if $p_{j+1} \in F$. Similar to CEA, we define the complex event $C_\rho = (i,j,\mu_\rho)$ of $\rho$ where $\mu_\rho(X) = \{k \in [i..j] \mid X \in L_k\}$ for every $X \in \xset$, and the \emph{output set} of $\cT$ over $\Stream$ as $\sem{\cT}(\Stream) = \{C_\rho \mid \rho \text{ is an accepting run of $\cT$ over $\Stream$}\}$.

Note that a run of a timed CEA must start with all clocks empty (i.e., clocks are not initialized by default). Then, to use a clock $z$, the automata first has to reset $z$ in a run. This technical detail is different from timed automata~\cite{AlurD94}, and it is necessary to be aligned with the principles behind the semantics of CEL. If not, the automata could have the power to check a property with the event just before the interval $[i..j]$ (i.e., $e_{i-1}$), and thus, the complex event $C_\rho$ will not depend only on $[i..j]$.
\begin{example} \label{ex:timed-cea1}
	The following timed CEA uses a single clock $z$ and set of variables $\xset = \{X, Y\}$:
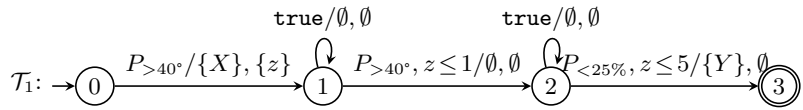
	\begin{center}
		\small
		\begin{tikzpicture}[->,>=stealth, semithick, auto, initial text= {}, initial distance= {3mm}, accepting distance= {4mm}, node distance=3cm, semithick]
			\tikzstyle{every state}=[draw=black,text=black,inner sep=0pt, minimum size=5mm]
			
			\node[initial,state] (0) {$0$};
			\node[state] (1) [right of=0]	{$1$};
			\node[state] (2) [right of=1]	{$2$};
			\node[state,accepting] (3) [right of=2]	{$3$};
			
			\draw[->] (0) edge node {$P_{> \text{40°}} / \{X\}, \{z\}$} (1);
			\draw[->] (1) edge[loop above] node {$\texttt{true} / \emptyset, \emptyset$} (1);
			\draw[->] (1) edge node {$P_{> \text{40°}}, z \!\leq\! 1 / \emptyset, \emptyset$} (2);
			\draw[->] (2) edge[loop above] node {$\texttt{true} / \emptyset, \emptyset$} (2);
			\draw[->] (2) edge node {$P_{< \text{25\%}}, z \!\leq\! 5 / \{Y\}, \emptyset$} (3);
			
			\node at ($(0) + (-0.9,0.05)$) {$\cT_1$:};
			
		\end{tikzpicture}
	\end{center}
	where $P_{> \text{40°}} := temp > \text{40°}$, $P_{< \text{25\%}} := hum < \text{25\%}$, and $\texttt{true}$ denotes the predicate $\eset$. One can check that $\cT_1$ defines the same query as formula $\varphi_1'$ from Example~\ref{ex:timed-cel1}.
\end{example}

\begin{example}\label{ex:cea-overlap}
To illustrate the advantages of more clocks, suppose that to reinforce formula $\varphi_1$ (see Example~\ref{ex:cel-query}), we want to detect two pairs $(X,Y)$ (i.e., high temperature / low humidity) where each pair must be between less than five seconds. Further, to confirm the high temperatures and low humidities, we expect to see the temperatures first and then the humidities. For defining this query with a timed CEA, we need states to check the pattern $T;T;H;H$ and two clocks $z_1$ and $z_2$ to check the time differences for each pair $(X,Y)$, respectively. The timed CEA $\cT_2$ specifies this query:
	\begin{center}
		\small
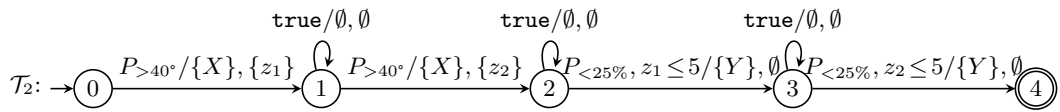
		\begin{tikzpicture}[->,>=stealth, semithick, auto, initial text= {}, initial distance= {3mm}, accepting distance= {4mm}, node distance=3cm, semithick]
			\tikzstyle{every state}=[draw=black,text=black,inner sep=0pt, minimum size=5mm]
			
			\node[initial,state] (0) {$0$};
			\node[state] (1) [right of=0]	{$1$};
			\node[state] (2) [right of=1]	{$2$};
			\node[state] (3) [right of=2,node distance=3.2cm]	{$3$};
			\node[state,accepting] (4) [right of=3,node distance=3.2cm]	{$4$};
			
			\draw[->] (0) edge node {$P_{> \text{40°}} / \{X\}, \{z_1\}$} (1);
			\draw[->] (1) edge[loop above] node {$\texttt{true} / \emptyset, \emptyset$} (1);
			\draw[->] (1) edge node {$P_{> \text{40°}} / \{X\}, \{z_2\}$} (2);
			\draw[->] (2) edge[loop above] node {$\texttt{true} / \emptyset, \emptyset$} (2);
			\draw[->] (2) edge node {$P_{< \text{25\%}}, z_1 \!\leq\! 5 / \{Y\}, \emptyset$} (3);
			\draw[->] (3) edge[loop above] node {$\texttt{true} / \emptyset, \emptyset$} (3);
			\draw[->] (3) edge node {$P_{< \text{25\%}}, z_2 \!\leq\! 5 / \{Y\}, \emptyset$} (4);
			
			\node at ($(0) + (-0.89,0.05)$) {$\cT_2$:};
		\end{tikzpicture}
	\end{center}
\end{example}

\paragraph{The expressive power of timed CEL} We proceed by measuring the expressiveness of timed CEL with its automata counterpart, timed CEA. We start by showing that any timed CEL can be compiled into a timed CEA that defines the same query as expected.  
 
\begin{proposition} \label{pro:cel2cea}
	For every timed CEL formula $\varphi$ there exists a timed CEA $\cT_\varphi$ such that $\sem{\varphi}(\Stream) = \sem{\cT_\varphi}(\Stream)$ for every stream $\Stream$. 
\end{proposition}

\begin{proof}[Proof sketch.]
The proof goes by induction over the formula, showing how to compile each operator into a timed CEA. The standard CEL operators follow a similar construction like in~\cite{GrezRUV21}, but this time one also needs to take care of clocks, keeping disjoint sets of clocks for each subautomaton. 
For the new timed operators, one can use a new clock for each operator, for checking the condition of the time window, time sequencing, or time iteration. We present the proof in Appendix~\ref{app:timed-cea}.
\end{proof}

For the other direction, it is less clear whether any timed CEA can be defined by some timed CEL formula. Indeed, timed CEA can check overlapping timed restrictions (like in Example~\ref{ex:cea-overlap}) while timed CEL can only check nested time constraints with the time window operator (i.e., $\twin{ \cdot }{I}$) or immediate time constraints with the timed sequencing (i.e.,  $\timesq{I}$ or $\timessq{I}$) or timed iteration (i.e., $\timeks{I}$ or $\timesks{I}$). Interestingly, timed CEL formulas can simulate overlapping time constraints by using variables ($\cas$), projection ($\pi_L$), and conjunction ($\cand$). For instance, we can define the query of Example~\ref{ex:cea-overlap} with the following timed CEL query:
\[
\begin{array}{l}
 \varphi_3 \, := \, \pi_{X,Y}\Big( \big[ \twin{T \kAS X; T \kAS X; H \kAS Y}{\leq 5} ; H \kAS Y \ \ \  \kAND \\
	\hspace{20mm} T \kAS X; \twin{T \kAS X; H \kAS Y ; H \kAS Y}{\leq 5} \ \big]
	\kFILTER (X[temp > \text{40°}] \wedge Y[hum < \text{25\%}]) \Big) 
\end{array}
\]
As the following result shows, we can follow a similar strategy for every timed CEA.
\begin{proposition} \label{pro:cea2cel}
	For every timed CEA $\cT$ there exists a timed CEL formula $\varphi_\cT$ such that $\sem{\cT}(\Stream) = \sem{\varphi_\cT}(\Stream)$ for every stream $\Stream$. Furthermore, $\varphi_\cT$ does not use any operator in $\{;, +, \timesq{I}, \timessq{I}, \timeks{I}, \timesks{I}\}$. 
\end{proposition}
\begin{proof}[Proof sketch.]
	The proof starts by reducing the construction from a timed CEA $\cT$ with multiple clocks to a timed CEA with a single clock $\cT_1$. For each clock $z$, one can construct a timed CEA $\cT_z$ that simulates $\cT$, but it only verifies the clock $z$ and marks with temporal variables all the transitions that are used during a run. Then, for all the clocks $z_1, \ldots, z_k$ used by $\cT$, we can take the conjunction of $\cT_{z_1}, \ldots, \cT_{z_k}$ and then project out the temporal variables. Therefore, the problem reduces to encode with timed CEL the behavior of a timed CEA $\cT_1$ with a single clock. 
	
	For encoding a timed CEA $\cT_1$ with a single clock $z$, we use two automata constructions to reduce the complexity of how the clock $z$ is used. For the first construction, we show that $\cT_1$ can be reduced to a timed automaton $\cT_1'$ such that $z$ is checked at most $k$-times after it is reset, where $k$ is a fixed constant. $\cT_1'$ is called a \emph{$k$-check bounded timed CEA}. This reduction helps to limit the number of times a clock is used after a reset, similar to a timed CEL formula where each clock is used exactly once (i.e., after a reset). For the second construction, we show that the $k$-check bounded automaton $\cT_1'$ can be specified as the conjunction of two $k$-check bounded automaton $\cT_1''$ and $\cT_1'''$ such that each automaton can either reset a clock or check a clock during a transition, but not both at the same time. We call this condition a \emph{simplified timed CEA}. 
	
	Finally, we prove that any simplified and $k$-check bounded timed CEA with a single clock can be defined by a timed CEL formula. This construction follows a similar strategy to the Kleene Theorem (for regular expressions), where the formula is constructed by induction on the set of states and the number of checks that are used during a run. More details of the proof in Appendix~\ref{app:timed-cea}.
\end{proof}

By combining Proposition~\ref{pro:cel2cea}~and~\ref{pro:cea2cel}, we can show that both formalism are equally expressive. 
\begin{theorem}\label{theo:main}
	Timed CEL and timed CEA are equally expressive, namely, one can be defined by the other and vice versa. 
\end{theorem}

It is important to note that an analog result was proven in~\cite{AsarinCM2002}, showing that timed automata are equally expressive than Timed Regular Expressions (TREs). TREs are an extension of regular expressions with time windows, renaming, conjunction, and the so-called absorbing concatenation and absorbing iteration. The paper aimed to extend the expressive power of regular expressions to capture timed automata. Although Proposition~\ref{pro:cea2cel} is inspired by~\cite{AsarinCM2002}, arguably Theorem~\ref{theo:main} is novel, both in its statement and proof, for the following reasons. First, for reaching Proposition~\ref{pro:cea2cel} we only have to extend CEL with time windows, an operator already presented in practical proposals in CER. In particular, CEL already included all other operators needed for the characterization. In contrast, TREs needed several new operators for regular expressions to capture the expressive power of time automata. Second, although one can make a parallel between the operators between timed CEL and TREs, they are different in syntax and semantics. Indeed, timed CEL includes outputs (i.e., complex events) and operators like projection and filtering, which are features that are not present in TREs. For this reason, the proof of Proposition~\ref{pro:cea2cel} is not a direct consequence of the results in~\cite{AsarinCM2002}. Finally, the proof of Proposition~\ref{pro:cea2cel} differs from the one in~\cite{AsarinCM2002}. The proof strategy of reducing the problem to single-clock automata is similar; however, the proof differs from there. Indeed, we cannot follow the same approach as in~\cite{AsarinCM2002} given the semantics of timed CEL, namely, by using the so-called \emph{absorbing} concatenation and iteration and constructing quasilinear equations. We cannot simulate these operators in timed CEL, and therefore, we have to construct the formula for a single-clock timed CEA directly, as we explained in the proof sketch of Proposition~\ref{pro:cea2cel}.  %
\label{sec:timed-cea}
	
	\section{On the determinization of timed complex event automata}

\newcommand{\varphinodet}{\varphi_{\text{non-det}}}

Now that we understand the expressive power of timed CEL concerning timed CEA, we move next to the efficient evaluation of timed CEL or, more concretely, of timed CEA. A critical aspect for efficient query evaluation is whether timed CEA are determinizable or not. We say that a timed CEA $\cT = (Q, \pset, \xset, \zset, \Delta, q_0, F)$ is \emph{deterministic} if for every pair of transitions $\tau_1 = (p, P_1, \gamma_1, L_1, Z_1, q_1)$ and $\tau_2 = (p, P_2, \gamma_2, L_2, Z_2, q_2)$ in $\Delta$, if $P_1 \cap P_2 \neq \emptyset$ and $\gamma_1 \wedge \gamma_2$ is satisfiable (i.e., there exists $\nu$ such that $\nu \models \gamma_1$ and $\nu \models \gamma_2$), then $L_1 \neq L_2$.
In other words, for any configuration $(p, \nu)$ of $\cT$, an event $e$ can trigger transitions $\tau_1$ and $\tau_2$ simultaneously only if they produce different outputs.
\begin{example}\label{ex:det-cea}
$\cT_2$ in Example~\ref{ex:cea-overlap} is a deterministic timed CEA, but $\cT_1$ in Example~\ref{ex:timed-cea1} is not since the predicates of the outgoing transitions of state \circled{1} intersect and produce the same output.
\end{example}

The importance of determinism for a timed CEA $\cT$ is that for every timed stream $\Stream$ and complex event $C \in \sem{\cT}(\Stream)$, there exists a unique run $\rho$ of $\cT$ over $\Stream$ that produces $C$, namely, $C = C_\rho$. This implies that there exists a correspondence between runs and outputs, which is very useful for designing efficient streaming (and non-streaming) enumeration algorithms. Indeed, most of the algorithms in the query evaluation literature assumed that the automata model that represents a query is deterministic~(e.g., \cite{Bagan06,FlorenzanoRUVV18,SchmidS21}). Therefore, being deterministic is desirable for the efficient evaluation of timed CEA and the enumeration of their outputs. Nevertheless, it is important to say that the size of the determinized automaton is exponential. However, if we think about data complexity, the blow-up of the automata is in the size of the query (which is small) and not in the size of the stream (which is big).

Unfortunately, not all timed CEA admit an equivalent deterministic one. Specifically, we say that a timed CEA $\cT$ is \emph{determinizable} if there exists a deterministic timed CEA $\cT'$ equivalent to $\cT$ (i.e., $\sem{\cT}(\Stream) = \sem{\cT'}(\Stream)$). This issue of non-determinizable timed CEA already holds for timed automata, where it is well-known that they are not always determinizable~\cite{AlurD94,AlurFH99,Baier2009}. Moreover, determining if a timed automaton is determinizable is an undecidable problem~\cite{Finkel2006}, and that property is easy to extend to timed CEA. As an example, for the following timed CEL formula, there exists no deterministic timed CEA that defines it:
\[
\varphinodet := \pi_{\emptyset}(A; A \timesq{=1} A).
\]
Note that the above formula is analog to the non-determinizable timed automata in~\cite[Figure 12]{AlurD94}.
Intuitively, $\varphinodet$ looks for complex events $(i, j, \mu_\emptyset)$ of three $A$-events $e_i$, $e_k$, and $e_j$ with $k \in [i..j]$ such that the time difference between $e_k$ and $e_j$ is exactly $1$. 
A deterministic timed CEA must keep track of the timestamp of all $A$-events after $e_i$ to verify that such $e_k$ exists when the last event $e_j$ arrives. Thus, we will need an unbounded number of clocks to define~$\varphinodet$.

Although timed CEA are not always determinizable, we strive to find a relevant class of timed CEA where each automaton is determinizable. Formally, we say that a class $\mathcal{C}$ of timed CEA is \emph{closed under determinization} if, and only if, every $\cT \in \mathcal{C}$ is determinizable. In the following, we present and study a class of timed CEA that is closed under determinization. Furthermore, we show that this class already contained an interesting class of timed CEL formulas that admit a representation as a deterministic timed CEA. 

\paragraph{Synchronous resets timed CEA} A timed CEA $\cT = (Q, \pset, \xset, \zset, \Delta, q_0, F)$ has \emph{synchronous resets} if for every 
$\Stream \ = \ (e_1, t_1) \ldots (e_n, t_n)$ and for every two runs of $\cT$ over $\Stream$ from position $i$ to $j$:
\[
(p_i, \nu_i) \xrightarrow{P_{i}, \gamma_i / L_{i}, Z_i} \dots \xrightarrow{P_{j}, \gamma_j / L_{j}, Z_j} (p_{j+1}, \nu_{j+1}) \ \ \ \ \ \text{and}  \ \ \ \ \ 
(p_i', \nu_i') \xrightarrow{P_{i}', \gamma_i' / L_{i}, Z_i'} \dots \xrightarrow{P_{j}', \gamma_j' / L_{j}, Z_j'} (p_{j+1}', \nu_{j+1}')
\]
it holds that $Z_{k} = Z_{k}'$ for every $k \in [i..j]$. Note that the two runs must have the same output $L_i, \ldots, L_j$, and they are not necessarily accepting runs; thus, the condition must hold for every run (i.e., accepting or not) or extension of them. We define the class of \emph{synchronous resets timed CEA} as all timed CEA that have synchronous resets. 
\begin{example} \label{ex:synch-cea}
	One can check that, although $\cT_1$ is not deterministic (Example~\ref{ex:det-cea}), it has synchronous resets since the clock is only reset at the beginning of the run. 
\end{example}

One can easily check that every deterministic timed CEA has synchronous resets but not necessarily the other way around. Interestingly, we can show that the class of synchronous resets timed CEA is closed under determinization and, thus, $\cT_1$ is determinizable.

\begin{theorem}\label{theo:determinization}
	For every synchronous resets timed CEA $\cT = (Q, \pset, \xset, \zset, \Delta, q_0, F)$, there exists a deterministic timed CEA $\cT'$ such that $\sem{\cT}(\Stream) = \sem{\cT'}(\Stream)$.  Furthermore, $|\cT'| \in \mathcal{O}(2^{|Q| + 2|\Delta|}\cdot|\cT|)$.
\end{theorem}

\begin{proof}[Proof sketch]
	The construction of $\cT'$ follows the standard subset construction where $2^Q$ are the states of $\cT'$. One has to additionally solve the use of predicates and time conditions in the transitions. Given that both are closed under conjunction and negation, we can use \emph{types} to represent subsets of them. Finally, given that the resets of competing runs %
	are synchronized, we can reset without considering conflicts between runs. The full proof is in Appendix~\ref{app:determinization}. 
\end{proof}

A natural question at this point is whether one can decide if a timed CEA has synchronous resets or not. The next result shows that this property is decidable in $\textsc{PSPACE}$.
\begin{theorem}\label{theo:algo-synch-resets}
	The problem of deciding if a timed CEA $\cT$ has synchronous resets is $\textsc{PSPACE}$-complete.
\end{theorem} 
\begin{proof}[Proof sketch]
	For the $\textsc{PSPACE}$ upper bound, we explore non-deterministically all of the different pairs of runs of the timed CEA. We can do so by grouping all of the runs that go through the same clock regions of the automaton, a technique introduced originally in~\cite{AlurD94}. That way, the number of possible next runs is polynomial. If any two runs do not have the same clock resets, then the timed CEA does not have synchronized resets.
	
	For the $\textsc{PSPACE}$-hardness, we reduce from the time automaton \emph{emptiness} problem. Given a timed automaton $\cA$, we construct a timed CEA $\cT_{\cA}$ where we mark all of the transitions with different variables so that the resulting timed CEA is deterministic. Then, we add two self-loops resetting different sets of clocks in the final states so that a run in $\cA$ reaches a final state if, and only if, two runs reset different sets of clocks in $\cT_{\cA}$. We present the full proof in Appendix~\ref{app:determinization}.
\end{proof}

The determinization of timed automata is a relevant topic in the area. Indeed, one can translate classes of determinizable timed automata to timed CEA, and vice versa, by reinterpreting the notion of output and predicates. For instance, synchronous resets are more powerful than \emph{event-clock automata} introduced in~\cite{AlurFH99} because every event-clock automata has synchronous resets. Other classes of determinizable timed automata, like timed automata with integer resets~\cite{SumanPKM08}, non-Zeno timed automata~\cite{asarin1998controller}, or history-deterministic timed automata~\cite{henzinger2022history}, are incomparable with synchronous resets. Theorem~\ref{theo:determinization} can also be recovered from general determinization procedures (see~\cite{Baier2009,lorber2017bounded}), by unfolding the timed CEA with synchronous resets into an infinite tree and determinizing it afterwards, because the infinite tree would be $|Z|$-bounded, where $|Z|$ is the number of clocks.
All in all, synchronous resets timed CEA are incomparable to previous proposals in terms of the setting and the use of clocks. Moreover, although the determinization of synchronous resets may be recovered from more general proposals, one could argue that the determinization procedure of synchronous resets is simpler, computable \emph{on-the-fly} (i.e., one can compute the next state in linear time), and one can decide the property in $\textsc{PSPACE}$ (i.e., for \cite{Baier2009,lorber2017bounded} it is undecidable). Finally, as we show next, the class of synchronous resets contains a relevant class of timed CEL formulas.

We say that a timed CEL formula $\psi$ is \emph{simple} if $\psi$ does not use projection $\pi_L(\cdot)$ or time window~$\twin{\cdot}{I}$ as a subformula. Further, a timed CEL formula $\phi$ is \emph{windowed} if it respects the following~grammar:
\[
\phi := \varphi \mid \psi \mid \phi \kAS X \mid \phi \kFILTER X[P] \mid \phi \kOR \phi \mid \phi \kAND \phi \mid\twin{\phi}{I} 
\]
where $\varphi$ is a standard CEL formula (i.e., no time operators), $\psi$ is a simple timed CEL formula, $P$ is a predicate, and $I$ is an interval. In other words, windowed CEL formulas have two levels. The first level allows either the free use of all timed CEL operators except for projection and time windows or the use of no time operator at all. The second level allows the use of sets of operators and filters plus time windows.
For instance, all swg queries (i.e., queries of the form~(\ref{eq:swg-form})) are windowed. 
Interestingly, one can compile any windowed timed CEL formula into timed CEA with synchronous resets, and, therefore, they are determinizable. 

\begin{theorem}\label{theo:windowed-sync}
	Let $\phi$ be a windowed timed CEL formula. Then, there exists a two-clock timed CEA with synchronous resets $\cT$ such that $\sem{\cT}(\Stream) = \sem{\phi}(\Stream)$.
\end{theorem}

The previous result pinpointed a timed CEL fragment that generalizes swg queries and single windows queries (e.g., SASE+\cite{agrawal2008efficient}, CORE\cite{BucchiGQRV22}). Moreover, each formula can be compiled into a deterministic timed CEA. Therefore, windowed timed CEL is an expressive fragment and a good candidate for efficient streaming evaluation. Although this does not mean that windowed timed CEL can be evaluated efficiently, at least one could exploit the determinism to process~them.

 \label{sec:determinization}
	
	\section{Towards the efficient evaluation of timed complex event automata}

In this last section, we will look for efficient algorithms for evaluating (deterministic) timed CEA in the streaming scenario of CER. Here, we will provide a partial answer by showing the first algorithm for evaluating a restricted class of deterministic timed CEA with a single clock and monotonic clock conditions. We show that a streaming enumeration algorithm for this class is already non-trivial and, further, that it already includes an interesting subclass of queries. 

The plan for this section is to first present the problem with the algorithmic guarantees we strive for to solve it. Then, we define the class of monotonic timed CEA with the statement of the main algorithmic result. Finally, we provide a overview of the algorithm and some of its consequences. Given space restrictions, we move all the technical details of the algorithm to Appendix~\ref{app:evaluation}. 

\paragraph{The evaluation problem} A timed CEA $\cT$ defines a function from a timed stream $\Stream$ to a set of complex events $\sem{\cT}(\Stream)$. In practice, the stream is arbitrarily long  and, thus, users want results as soon as they occur. Then, it is natural to define the restriction of $\sem{\cT}(\Stream)$ at position $j \in [|\Stream|]$ as
$
\sem{\cT}_j(\Stream) = \{C \in \sem{\cT}(\Stream) \mid \cend(C) = j\}
$, namely, all complex events that end at position $j$. Given the semantics of timed CEL and timed CEA, if a complex event starts and ends at positions $i$ and $j$, respectively, then it only depends on the substream between $i$-th and $j$-th positions. Therefore, as soon as the $j$-th timed event arrives, we can safely return all complex events $\sem{\cT}_j(\Stream)$.

Let $\cC$ be a class of timed CEA. The main evaluation problem that we aim to solve, parameterized by $\cC$, can be formalized as follows:
\begin{center}
		\begin{tabular}{rl}
			\hline\\[-2ex]
			\textbf{Problem:} & $\EvalProblem[\cC]$\\
			\textbf{Input:} & A timed CEA $\cT \in \cC$ and a timed stream $\Stream$ \\\
			\textbf{Output:} & Enumerates all complex events $\sem{\cT}_j(\Stream)$ at each position $j \in [|\Stream|]$ \\\\[-2.2ex]
			\hline
		\end{tabular}
\end{center}
Our main goal is to find algorithms for $\EvalProblem[\cC]$ that evaluates $\cT$ reading $\Stream = (e_1, t_1) \ldots (e_n,t_n)$ event by event and enumerates  $\sem{\cT}_j(\Stream)$ right after reading the $j$-th events. For this, we assume that one can only access $\Stream$ through a method $\yield(\Stream)$ such that the $j$-th call to this function retrieves $(e_j, t_j)$ if $j \leq n$, or \EOS~(End-Of-Stream) otherwise. 

A \emph{streaming evaluation algorithm} for $\EvalProblem[\cC]$ is an algorithm that receives as input $\cT \in \cC$ and consumes $\Stream$ by calling $\yield(\Stream)$, maintaining a data structure $\DS$ such that, after the $j$-th call to $\yield(\Stream)$, it receives $(e_j, t_j)$ and operates over $\DS$ in two phases. The first phase, called the \emph{update phase}, updates $\DS$ with $(e_j, t_j)$. The second phase, called the \emph{enumeration phase}, enumerates the complex events in $\sem{\cT}_j(\Stream)$ one by one using $\DS$. After the second phase is ready, it proceeds to the next call to $\yield(\Stream)$, and so on.

We consider two parameters to measure the efficiency of a streaming evaluation algorithm $\cS$. First, we say $\cS$ has $f$-\emph{update time} for some function $f$ if the time during the update phase of $(e_j, t_j)$ is at most $\cO(f(\cT, e_j))$. Second, we say that $\cS$ has \emph{output-linear delay} if there exists a constant $d \in \bbN$ such that, for every $j$, during the enumeration phase of $\sem{\cT}_j(\Stream) = \{C_1, \ldots, C_m\}$, the time taken between retrieving $C_k$ and $C_{k+1}$ (i.e., the \emph{delay}) is bounded by $d \cdot |C_{k+1}|$. This implies that the time to output the first complex event is at most $d \cdot |C_1|$ and constant time after the last output. In particular, $\cS$ takes constant time to end the enumeration phase if there is no output (i.e., $\sem{\cT}_j(\Stream) = \emptyset$). 

For $\EvalProblem[\cC]$, a streaming evaluation algorithm with $f$-update time and output linear delay is a strong guarantee for efficiency, and it is the standard notion used in previous work~\cite{BerkholzKS17,BucchiGQRV22,MunozR22}. Note that if one considers that $\cT$ is small or constant (i.e., data complexity), the update time per event is proportional to the next event. Further, output-linear delay (see also constant-delay~\cite{Segoufin13}) is the gold standard for enumeration problems, where the delay between outputs only depends on the next one. As it is standard in the literature, for these algorithms we assume the usual RAM model with unit costs and logarithmic size registers (see e.g., \cite{aho1974design,GrandjeanJ22}).

\paragraph{A class of timed CEA that admits a streaming evaluation algorithm} The main open problem is to pinpoint a class $\cC$ of timed CEA such that $\EvalProblem[\cC]$ admits a streaming evaluation algorithm with $f$-update timed and output-linear delay for some function $f$. This problem is already open even for time automata (see, e.g., \cite{GrezMPPR22}), where the output is true or false after each event.
Here, we present a first class of deterministic timed CEA that admits a streaming evaluation, improving our understanding of the efficient evaluation of timed CEA. 

For $\sim \ \in \{\leq, \geq\}$, we say that a timed CEA $\cT = (Q, \pset, \xset, \zset, \Delta, q_0, F)$ has \emph{$\sim$-clock conditions} if all clock conditions are of the form $\bigwedge_{z \in Z} z \sim c_z$ for some $Z \subseteq \zset$ and $c_z \in \bbQplus$. Further, we say that a timed CEA is \emph{monotonic} if it either has $\leq$-clock conditions or $\geq$-clock conditions.

\begin{theorem}\label{theo:algorithm}
	Let $\cC^*$ be the class of all deterministic monotonic timed CEA with a single-clock. Then, $\EvalProblem[\cC^*]$ admits a streaming evaluation algorithm with $\cO(|\cT|^2 \cdot |e|)$-update time and output-linear enumeration, for a timed CEA $\cT \in \cC^*$ and next event $e$ of the input stream.
\end{theorem}
\begin{proof}[Proof sketch]
	The streaming evaluation algorithm of Theorem~\ref{theo:algorithm} is a non-trivial extension of the algorithm presented in~\cite{BucchiGQRV22} for evaluating CEL over a sliding window. We dedicate this proof sketch to explain the technical novelty of this algorithm with respect to~\cite{BucchiGQRV22}. Given space restrictions, we present all the technical details of the algorithm in the appendix. 
	
	First, Theorem~\ref{theo:algorithm} strictly generalizes the algorithm of~\cite{BucchiGQRV22}: it can manage a more expressive set of deterministic timed CEA. Indeed, although \cite{BucchiGQRV22} does not use timed CEA, the queries in~\cite{BucchiGQRV22} are of the form $\twin{\varphi}{\leq c}$ for a standard CEL formula $\varphi$. This formula can be converted into a deterministic monotonic timed CEA and, thus, can be evaluated with the algorithm of Theorem~\ref{theo:algorithm}. Further, our class of time automata can deal with more expressive queries like Example~\ref{ex:monotonic} (below). 
	
	For managing multiple clock constraints, we extend the data structure of~\cite{BucchiGQRV22}, called timed Enumerable Compact Set (tECS), by introducing two new nodes that we called \emph{reset nodes} and \emph{clock-check nodes}, which are in charge of the intermediate resets and checks during a run. These nodes allow the restriction of the enumeration to a subset of outputs that satisfy a reset or check condition. Adding these two classes of nodes requires adapting the conditions on the data structure and the enumeration procedure to ensure output-linear delay. 
	More importantly, we rethink the operations over the data structure, like union or extend, to ensure that the conditions are satisfied for output-linear enumeration of the results. Specifically, the reset nodes and check nodes do not produce output during the enumeration. Then, we need to enforce that they are not stacked during operations over the data structure. For this purpose, we introduce a \emph{gadget} that combines resets and checks, which can be simplified whenever two of them are combined. 
	
	Finally, we need to revisit the evaluation algorithm of~\cite{BucchiGQRV22} and the use of the so-called \emph{union lists}, a secondary data structure for managing lists of nodes. We generalize the conditions and invariants maintained over states and union lists to include the resets and checks in the transitions. This is the main purpose of the monotonic condition over the timed CEA, which allows us to generalize the main algorithmic ideas in~\cite{BucchiGQRV22}.
\end{proof}

\begin{example}\label{ex:monotonic}
One can check that $\cT_1$ is a monotonic timed CEA. Furthermore, it is determinizable (by Theorem~\ref{theo:determinization}), and the determinization has one clock and is monotonic. By Theorem~\ref{theo:algorithm}, we can efficiently evaluate the formula $\varphi_1$. Similarly for the formula $\varphi_2$ of Example~\ref{ex:timed-cel2}. 
\end{example}

We leave open to find a more general class of timed CEA that admits a streaming evaluation algorithm. To the best of our knowledge, the algorithm of Theorem~\ref{theo:algorithm} is the first streaming evaluation strategy that combines automata evaluation with general time constraints and efficient~enumeration.

 \label{sec:evaluation}
	
	\section{Conclusions and future work}\label{sec:conclusions}

This paper explored the expressibility and efficient evaluation of time constraints in complex event recognition. We showed that timed CEL has the same expressive power as timed CEA, meaning that the time operators suffice to write any CER query with time constraints. We also found a subset of timed CEL that is definable by deterministic timed CEA with two clocks, allowing the processing of the queries without removing duplicate runs. Finally, we constructed an algorithm to process monotonic timed CEAs with one clock efficiently, setting grounds for future real-life implementation of the framework.

This paper leaves several open problems regarding the efficient evaluation of timed CEL queries. First, it is an open question whether one can evaluate any timed CEL formula with $f$-update time and output-linear delay for some function $f$. Even the case of a single clock is open (i.e., without the monotonic restriction). In~\cite{GrezMPPR22}, the streaming evaluation of timed automata with a single clock was solved; however, it is unclear how to extend the approach to more clocks for boolean output or to single-clock timed CEA. Even the case of swg queries would be an interesting problem to study. Second, finding more expressible classes of determinizable timed CEL formulas is a relevant problem for these algorithms. Finally, it will be interesting to consider time restrictions together with other features of CEP, like correlation (i.e., joins).

	\newpage
	\bibliographystyle{abbrv}
	\bibliography{biblio}
	
	\newpage
	\appendix

	\section{Proofs from Section~\ref{sec:timed-cea}} \label{app:timed-cea}

\subsection*{Proof of Proposition~\ref{pro:cel2cea}}

\begin{proof}
	We prove this result by constructing the timed CEA 
	$\cT_\varphi = (Q_\varphi, \pset, \xset, \zset_\varphi, \Delta_\varphi, q_{0\varphi},F_\varphi)$
	 by induction over the syntax of formula $\varphi$ as follows. Note that with this construction, there are no transitions from any state to the initial state, and every clock is reset before it is used. 
	 Both assumptions are important for the next constructions.
	\begin{itemize}\itemsep3mm
		\item If $\varphi = R$, then $\cT_\varphi$ is defined as $\cT_\varphi = (\{q_1,q_2\},\pset, \xset , \varnothing, \{(q_1,P_R,\TRUE, \{R\}, \varnothing, q_2)\}, q_1,\{q_2\})$, where $P_R$ is the predicate containing all events of type $R$. We will demonstrate that $\semaux{\cT_\varphi}(\Stream) = \semaux{R}(\Stream)$. If $C\in \semaux{\cT_\varphi}(\Stream)$, then an accepting run over $\Stream$ is $\rho: q_1 \xrightarrow{P_R, \TRUE /  \{R\}, \varnothing} q_2$, then $C_\rho = \{i, i,\ \mu\}$ where $\mu = \{(R, \{i\})\}$. Following the semantics that were defined earlier, $C_\rho \in \semaux{R}(\Stream)$. If $C \in \semaux{R}(\Stream)$, then $C = (i, i, \{(R, \{i\})\})$. A run of the automaton over $\Stream$ from position $i$ to $j$ is $q_0\xrightarrow{P_R, \TRUE /  \{R\}, \varnothing}q_1$. Then, $C \in \semaux{\cT_{\varphi}}(\Stream)$.
		
		\item If $\varphi = \psi \kAS X$, where $X$ is a new variable, then $\cT_\varphi = (Q_\psi, \pset, \xset, \zset_\psi, \Delta_\varphi, q_{0\psi},F_\psi)$ where $\Delta_\varphi$ is the result of adding variable $X$ to all marking transitions of $\Delta_\psi$, i.e., $\Delta_\varphi = \{(p,P,\gamma,L,Z,q)\in\Delta_\psi \mid L = \emptyset\}\cup\{(p,P,\gamma,L\cup\{X\},Z,q)\mid (p,P,\gamma,L,Z,q)\in \Delta_\psi \land L \ne \varnothing\}$. The proof of correctness is straightforward.
		
		\item If $\varphi = \psi \kFILTER X[P]$, where $X$ is an existing variable and $P\in\pset$ is a predicate, then $\cT_\varphi = (Q_\psi, \pset,\xset, \zset_{\psi}, \Delta_\varphi, q_{0\psi}, F_{\psi})$, where $\Delta_{\psi}$ is the result of filtering all transitions marked by the variable $X$ with the predicate $P$, i.e., $\Delta_\varphi = \{(p,P',\gamma,L,Z,q)\in\Delta_\psi \mid X \notin L\}\cup\{(p,P' \cap P,\gamma,L,Z,q)\mid (p,P',\gamma,L,Z,q)\in \Delta_\psi \land X \in L\}$. Again, the proof here is straightforward.
		
		\item If $\varphi = \psi_1 \kOR \psi_2$, then $\cT_\varphi$ is the union between $\cT_{\psi_1}$ and $\cT_{\psi_2}$, that is,  $\cT_\varphi = (Q_{\psi_1} \cup Q_{\psi_2}\cup\{q_{0\varphi}\},\pset, \xset, \zset_{\psi_1}\cup \zset_{\psi_2},\Delta_{\varphi}, q_{0\varphi}, F_{\psi_1} \cup F_{\psi_2})$,where $q_{0\varphi}$ is a fresh new state, and $\Delta_{\varphi} = \Delta_{\psi_1}\cup\Delta_{\psi_2}\cup \{(q_{0\varphi}, P, \gamma, L, Z, q)\mid (q_{0\psi_1}, P, \gamma, L, Z, q)\in \Delta_{\psi_1} \lor (q_{0\psi_2}, P, \gamma, L, Z, q)\in \Delta_{\psi_2}\}$. Here, we assume w.l.o.g. that $\cT_{\psi_1}$ and $\cT_{\psi_2}$ have disjoint sets of states.
		We omit the proof of correctness given that it is direct from the construction. 
		
		\item If $\varphi = \psi_1 \kAND \psi_2$, then $\cT_{\varphi} = (Q_{\psi_1} \times Q_{\psi_2}, \pset, \xset, \zset_{\psi_1} \cup \zset_{\psi_2},  \Delta_\varphi, (q_{0\psi_1}, q_{0\psi_2}), F_{\psi_1} \times F_{\psi_2})$, and:
		\[
		\begin{array}{rcl}
			\Delta_\varphi & = & \{((p_1, p_2), P_1 \cap P_2, \gamma_1 \land \gamma_2, L, Z_1 \cup Z_2, (q_1, q_2)) \mid \\
			& & \hspace{2cm} (p_1, P_1,  \gamma_1, L, Z_1, q_1) \in \Delta_{\psi_1} \land (p_2, P_2, \gamma_2, L, Z_2, q_2) \in \Delta_{\psi_2}\}.
		\end{array}
		\]
		We assume that $\cT_{\psi_1}$ and $\cT_{\psi_2}$ have disjoint sets of clocks. Similarly, the proof of correctness is direct here.
		
		\item If $\varphi = \psi_1 \sq \psi_2$, then $\cT_\varphi = (Q_{\psi_1} \cup Q_{\psi_2},\pset,\xset, \zset_{\psi_1}\cup \zset_{\psi_2}, \Delta_\varphi,  q_{0\psi_1}, F_{\psi_2})$ where $\Delta_\varphi = \Delta_{\psi_1} \cup \Delta_{\psi_2} \cup \{(p,P,\gamma,L,Z,q_{0\psi_2}) \mid (p,P,\gamma,L,Z,q_f)\in \Delta_{\psi_1} \land q_f\in F_{\psi_1}\} \cup \{(q_{0\psi_2},P,\TRUE,\varnothing,\varnothing,q_{0\psi_2})\}$. Here, we assume that $\cT_{\psi_1}$ and $\cT_{\psi_2}$ have disjoint sets of states, disjoint sets of clocks, and exploit the fact that there are no transitions from any state to the initial state. We assume by induction that $\semaux{\psi_1}(\Stream) = \semaux{\cT_{\psi_1}}(\Stream)$ and $\semaux{\psi_2}(\Stream) = \semaux{\cT_{\psi_2}}(\Stream)$, and will prove that $\semaux{\psi_1 \sq \psi_2}(\Stream) = \semaux{\cT_{\varphi}}(\Stream)$. 
		
		Let $(i,j, \mu)\in \semaux{\cT_{\varphi}}(\Stream)$ be a complex event. Then, there exists an accepting run over $\Stream$, $\rho: (q_0, \nu_0) \xrightarrow{P_i, \gamma_i / L_i, Z_i} \dots \xrightarrow{P_j, \gamma_j/ L_j, Z_j} (q_{f}, \nu_f)$. By construction, we know that the run has the following structure: $\rho: (q_{0\psi_1}, \nu_0) \xrightarrow{P_i, \gamma_i/ L_i, Z_i} \dots \xrightarrow{P_k, \gamma_k/ L_k, Z_k} (q_{0\psi_2}, \nu_k) \dots (q_{0\psi_2}, \nu_k') \xrightarrow{P_{k'}, \gamma_{k'} / L_{k'}, Z_{k'}} \dots \xrightarrow{P_j, \gamma_j / L_j, Z_j} (q_{f}, \nu_f)$. By construction also, there exists an accepting run of $\cT_{\psi_1}$ over $\Stream$, $\rho_1: (q_{0\psi_1}, \nu_0)\xrightarrow{P_i, \gamma_i/ L_i, Z_i} \dots \xrightarrow{P_k, \gamma_k/ L_k, Z_k} (q_{f\psi_1}, \nu_{f\psi_1})$. and a run of $\cT_{\psi_2}$ over $\Stream$, $\rho_2: (q_{0\psi_2}, \nu_0) \xrightarrow{P_{k'}, \gamma_{k'} / L_{k'}, Z_{k'}} \dots \xrightarrow{P_j, \gamma_j / L_j, Z_j} (q_{f\psi_2}, \nu_{f\psi_2})$. Summarizing, we know that $(i, k, \mu_1) \in \semaux{\cT_{\varphi_1}}(\Stream)$ and $(k', j, \mu_2) \in \semaux{\cT_{\varphi_2}}(\Stream)$. By assumption, we know that $(i, k, \mu_1) \in \semaux{\varphi_1}(\Stream)$ and $(k', j, \mu_2) \in \semaux{\varphi_2}(\Stream)$, with $k < k'$, then $(i, j, \mu) \in \semaux{\varphi}(\Stream)$.
		
		In the other direction, let $(i, j, \mu) \in \semaux{\varphi}(\Stream)$. Then, there exist $(i, k, \mu_1) \in \semaux{\varphi_1}(\Stream)$ and $(k', j, \mu_2) \in \semaux{\varphi_2}(\Stream)$, with $k < k'$. By assumption, we know that there exist runs $\rho_1$ of $\cT_{\varphi_1}$ over $\Stream$ and $\rho_2$ of $\cT_{\varphi_2}$ over $\Stream$, where $\rho_1: (q_{0\psi_1}, \nu_0)\xrightarrow{P_i, \gamma_i/ L_i, Z_i} \dots \xrightarrow{P_k, \gamma_k/ L_k, Z_k} (q_{f\psi_1}, \nu_{f\psi_1})$. and $\rho_2: (q_{0\psi_2}, \nu_0) \xrightarrow{P_{k'}, \gamma_{k'} / L_{k'}, Z_{k'}} \dots \xrightarrow{P_j, \gamma_j / L_j, Z_j} (q_{f\psi_2}, \nu_{f\psi_2})$. By construction, there exists a run $\rho$ of $\cT_\varphi$ over $\Stream$, where
		$\rho: (q_0, \nu_0)\xrightarrow{P_i, \gamma_i/ L_i, Z_i}\dots \xrightarrow{P_j, \gamma_j / L_j, Z_j} (q_f, \nu_f)$ and $q_f\in F$. Then, $(i, j, \mu)\in \semaux{\cT_{\varphi}}(\Stream)$.
		
		\item If $\varphi = \psi_1 \ssq \psi_2$, then $\cT_\varphi = (Q_{\psi_1} \cup Q_{\psi_2},\pset, \xset, \zset_{\psi_1}\cup \zset_{\psi_2}, \Delta_\varphi, q_{0\psi_1}, F_{\psi_2})$ where $\Delta_\varphi = \Delta_{\psi_1} \cup \Delta_{\psi_2} \cup \{(p,P,\gamma,L,Z,q_{0\psi_2}) \mid (p,P,\gamma,L,Z,q_f)\in \Delta_{\psi_1} \land q_f\in F_{\psi_1}\}$. Here, we assume that $\cT_{\psi_1}$ and $\cT_{\psi_2}$ have disjoint sets of states and clocks. This proof is analogous to the proof of operator $\sq$.
		
		\item If $\varphi = \psi \ks$, then $\cT_\varphi = (Q_\psi \,\cup\, \{q_{new}\}, \pset, \xset, \zset_{\psi}, \Delta_\varphi, q_{0\psi}, F_\psi)$ where $q_{new}$ is a fresh state, and 
		\[
		\begin{array}{rl}
			\Delta_\varphi = &\Delta_\psi \,\cup\, \{(p,P,\gamma,L,Z,q_{new}) \mid \exists q' \in F_\psi.\,(p,P,\gamma,L,Z,q') \in \Delta_\psi\}\,\cup\\
			&\{(q_{new},P,\TRUE,\varnothing,\varnothing,q_{new})\}\, \cup\\
			&\{(q_{new},P,\gamma,L,Z,q) \mid (q_{0\psi},P,\gamma,L,Z,q) \in \Delta_\psi\}\,\cup\\
			&\{(q_{new},P,\gamma,L,Z,q_{new}) \mid \exists q_f . \, (q_{0\psi},P,\gamma,L,Z,q_f) \in \Delta_\psi\}
		\end{array}
		\]
		In the previous definition, it is important that clocks are reset before they are used.
		Let $\varphi = \psi \ks$. We assume by induction that $\semaux{\psi}(\Stream) = \semaux{\cT_{\psi}}(\Stream)$, and will prove that $\semaux{\psi\ks}(\Stream) = \semaux{\cT_{\varphi}}(\Stream)$. Let $(i, j, \mu)\in \semaux{\psi\ks}(\Stream)$. Then, by construction, either $(i, j, \mu)\in \semaux{\psi}(\Stream)$ or $(i, j, \mu)\in \semaux{\psi\sq\psi\ks}(\Stream)$. We will continue this proof using induction. As a base case, if $(i, j, \mu)\in \semaux{\psi}(\Stream)$, then by initial assumption $(i, j, \mu)\in \semaux{\cT_\varphi}(\Stream)$, because the transitions from $\cT_{\psi}$ are preserved in $\cT_{\varphi}$. We will then assume that if $(i, k, \mu_1) \in \semaux{\psi}(\Stream)$ and $(k', j, \mu_2) \in \semaux{\psi\ks}(\Stream)$, then $(i, j, \mu) \in \semaux{\cT_{\varphi}}(\Stream)$. We then know that there exists runs $(q_0, \nu_0) \xrightarrow{P_i, \gamma_i / L_i, Z_i} \dots \xrightarrow{P_k,\gamma_k / L_k, Z_k} (q_f, \nu_k)$ over $\cT_{\psi}$ and $(q_0, \nu_0)\xrightarrow{P_{k'+1}, \gamma_{k'+1} /L_{k'+1}, Z_{k'+1}}\dots\xrightarrow{P_j, \gamma_j / L_j, \gamma_j} (q_f, \nu_j)$ over $\cT_{\varphi}$, with $k < k'$. Finally, if $(i, j, \mu)\in \semaux{\psi\sq\psi\ks}(\Stream)$, then there exists a run $\rho: (q_0, \nu_0) \xrightarrow{P_i, \gamma_i / L_i, Z_i} \dots \xrightarrow{P_k,\gamma_k / L_k, Z_k} (q_{new}, \nu_k)\xrightarrow{P_{k+1}, \gamma_{k+1} /L_{k+1}, Z_{k+1}}\dots\xrightarrow{P_{k'}, \gamma_{k'} /L_{k'}, Z_{k'}}(q_{k'}, \nu_{k}')\xrightarrow{P_{k'+1}, \gamma_{k'+1} /L_{k'+1}, Z_{k'+1}}\dots\xrightarrow{P_j, \gamma_j / L_j, \gamma_j} (q_f, \nu_j)$. Therefore $(i, j, \mu)\in \semaux{\cT_{\psi\sq\psi\ks}}(\Stream)$.
		
		On the other hand, let $C=(i, j, \mu)\in \semaux{\cT_{\varphi}}(\Stream)$. Then, there exists a run $\rho: (q_0, \nu_0) \xrightarrow{P_i, \gamma_i / L_i, Z_i} \dots \xrightarrow{P_j, \gamma_j / L_j, Z_j} (q_f, \nu_f)$. We will prove the result by induction. Let's first assume that the run doesn't pass through the state $q_{new}$. Then, the automaton passes only through states that were present in $\cT_\psi$, so it is clear that $C \in \semaux{\psi}(\Stream) \subseteq \semaux{\varphi}(\Stream)$. Now let's assume the run over $\Stream$ passes $n$ times through the new state $q_{new}$. We can take from the stream $n+1$ substreams, such that each sub-event $(k, k', \mu') \in \semaux{\cT_\psi}(\Stream)$ doesn't pass through $q_{new}$ (by construction). Let's call $k$ the last index of the first substream accepted by $\cT_\psi$. We can then take two complex events from $C$, the first one being $C_1 = (i, k, \mu_1)$ and the second one being $C_2 = (k', j, \mu_2)$, where $k'$ is the first index of the second substream. As such, we know that $C_1\in \semaux{\cT_{\psi}}(\Stream)$ and $C_2\in \semaux{\cT_{\varphi}}(\Stream)$, but the second run passes at most $n-1$ times through $q_{new}$. We conclude that $C \in \semaux{\varphi}(\Stream)$.
		
		\item If $\varphi = \psi\sks$, then $\cT_\varphi = (Q_\psi \cup \{q_{new}\}, \pset, \xset, \zset_{\psi}, \Delta_\varphi,  q_{0\psi}, F_\psi)$ where $q_{new}$ is a fresh state, $\Delta_\varphi = \Delta_\psi \cup \{(p,P,\gamma,L,Z,q_{new}) \mid \exists q' \in F_\psi.(p,P,\gamma,L,Z,q') \in \Delta_\psi\} \cup  \{(q_{new},P,\gamma,L,Z,q) \mid (q_{0\psi},P,\gamma,L,Z,q) \in \Delta_\psi\} \cup \{(q_{new},P,\gamma,L,Z,q_{new}) \mid \exists q_f\,.\, (q_{0\psi},P,\gamma,L,Z,q_f) \in \Delta_\psi\}$. The proof of this operator is analogous to the one for operator $\ks$.
		
		\item If $\varphi = \pi_L(\psi)$ for some $L \subseteq \xset$, then $\cT_\varphi = (Q_\psi,\pset, \xset, \zset_{\psi}, \Delta_\varphi, q_{0\psi},F_\psi)$ where $\Delta_\varphi$ is the result of intersecting the labels of each transition in $\Delta_\psi$ with $L$.
		Formally, that is $\Delta_\varphi = \{(p, P, \gamma,L \cap L',Z,q) \mid (p,P,\gamma,L',Z,q) \in \Delta_\psi\}$.
		The correctness is direct from the construction.
		
		\item If $\varphi = \twin{\psi}{I}$, then $\cT_\varphi$ is the result of starting a clock when the automata begins execution and checking it when the automata finishes execution: $\cT_{\varphi} = (Q_{\psi}\cup \{q_f\}, \pset, \xset, \zset_\psi\cup \{z_N\},\Delta_\varphi,  q_{0\psi}, \{q_f\})$, with $z_N$ a new clock, where 
		\[
		\begin{array}{rcl}
			\Delta_\varphi & = & \{(q_{0\psi}, P, \gamma, L, Z\cup \{z_N\}, q)\mid (q_{0\psi}, P, \gamma, L, Z, q)\in \Delta_{\psi}\land q\notin F_\psi\} \ \cup  \\
			& & \{(p, P, \gamma, L, Z, q)\in \Delta_{\psi}\mid p\ne q_{0\psi}\} \ \cup \\
			& & \{(p, P, \gamma\land \gamma_I, L, Z, q_f)\mid (p, P, \gamma, L, Z, q)\in \Delta_{\psi}\land p\ne q_{0\psi}\land q\in F_\psi\} \ \cup \\
			& & \{(q_{0\psi}, P, \gamma\land (0\in I), L, Z, q_f)\mid (q_{0\psi}, P, \gamma, L, Z, q)\in \Delta_{\psi} \land q\in F_\psi\}
		\end{array}
		\]
		where $\gamma_I$ is the clock constraint that checkes if the clock $z_N$ is part of the interval $I$, and $(0 \in I) := \TRUE$ if the expression evaluates to true and $\FALSE$ otherwise.
		\medskip
		
		For the proof of correctness, we assume by induction that $\semaux{\psi}(\Stream) = \semaux{\cT_{\psi}}(\Stream)$ and we will demonstrate that $\semaux{\cT_\varphi}(\Stream) = \semaux{\twin{\psi}{I}}(\Stream)$. Let $(i, j, \mu)\in \semaux{\cT_\varphi}(\Stream)$. Then, there exists a run over $\Stream$ $\rho: (q_0, \nu_0)\xrightarrow{P_i, \gamma_i / L_i, Z_i}\dots \xrightarrow{P_j, \gamma_j / L_j, Z_j}(q_f, \nu_f)$ where $C = \{(X, \{k \mid X\in L_k\})\mid X\in \xset\}$. We also know that $\nu_1(z_N) = 0$ (because the clock $z_N$ is resetted in the first transition) and that $\nu_f(z_N) = t_j-t_i$ because the clock $z_N$ is not resetted after the beginning of the run. We are presented with two types of accepting runs over the automaton: In the first type, the run has only one step, and goes directly from the initial state to the final state $\rho: (q_0, \nu_0)\xrightarrow{P_i, \gamma_i / L_i, Z_i}(q_1, \nu_1)$. In the second type of run, the run goes through more than two states, $\rho: (q_0, \nu_0)\xrightarrow{P_i, \gamma_i / L_i, Z_i} \dots \xrightarrow{P_j, \gamma_j / L_j, Z_j} (q_f, \nu_f)$. In both cases, there exists by definition an accepting run over $\cT_\psi$, in the first case it is $\rho': (q_0, \nu_0)\xrightarrow{P_i, \gamma_i / L_i, Z_i} (q_f, \nu_f)$ (with $i=j$) and in the second case it is $\rho': (q_0, \nu_0)\xrightarrow{P_i, \gamma_i / L_i, Z_i}\dots \xrightarrow{P_j, \gamma_j / L_j, Z_j} (q_f, \nu_f)$. Then, the valuation over this run $(i, j, \mu) \in \semaux{\cT_\psi}(\Stream)$, and by assumption $(i, j, \mu) \in \semaux{\psi}(\Stream)$. By the semantic definition of $\twin{\psi}{I}(\Stream)$, $C\in \semaux{\twin{\psi}{I}}(\Stream)$ if and only if $C\in \semaux{\psi}(\Stream)$ and $t_{\cend(C)} - t_{\cstart(C)}\in I$. Because $\rho$ was an accepting run, then we know that $t_{\cend(C)} - t_{\cstart(C)}  \in I$. Then $C\in \semaux{\twin{\psi}{\in I}}(\Stream)$. On the other side, be $C = (i, j, \mu)\in \semaux{\twin{\psi}{\in I}}(\Stream)$. Then, $t_{\cend(C)} - t_{\cstart(C)} \in I$. We know by construction that $(i, j, \mu)\in \semaux{\psi}(\Stream)$, then by assumption $(i, j, \mu)\in \semaux{\cT_\psi}(\Stream)$. Then there exists an accepting run of $\cT_\psi$ over $\Stream$, namely $\rho: (q_{i-1}, \nu_{i-1})\xrightarrow{P_i, \gamma_i / L_i, Z_i} \dots \xrightarrow{P_j, \gamma_j / L_j, Z_j} (q_j, \nu_j)$ where $q_j\in F$. By construction, there exists a run of $\cT_{\varphi}$ over $\Stream$, namely $\rho': (q_{i-1}, \nu_{i-1})\xrightarrow{P_i, \gamma_i / L_i, Z_i} \dots \xrightarrow{P_j, \gamma_j / L_j, Z_j} (q_f, \nu_f)$. The clock $z_N$ only is resetted at the first transition of the run, so we know that $\nu_1(z_N) = 0$ and $\nu_f(z_N) - \nu_1(z_N) = t_{\cend(C)} - t_{\cstart(C)}$. As the last transition only differs of the previous automaton by adding a clock constraint $\gamma_I$, and because $\nu_f(z_N) = t_{\cend(C)} - t_{\cstart(C)} \in I$, the run is accepting and the valuation is part of $\semaux{\cT_{\varphi}}(\Stream)$.
		
		\item Let $\varphi = \psi_1\xsq{I}\psi_2$, then $\cT_\varphi = (Q_{\psi_1} \cup Q_{\psi_2}\cup \{q_{new}\}, \pset, \xset, \zset_{\psi_1}\cup \zset_{\psi_2} \cup \{z_X\}, \Delta_\varphi, q_{0\psi_1}, F_{\psi_2} )$ where $z_X$ is a new clock, $q_{new}$ is a new state, and $\Delta_\varphi = \Delta_{\psi_1} \cup \Delta_{\psi_2} \cup \{(p,P,\gamma,L,Z\cup \{z_X\},q_{new}) \mid (p,P,\gamma,L,Z,q_f)\in \Delta_{\psi_1} \land q_f\in F_{\psi_1}\} \cup \{(q_{new},P,\TRUE,\varnothing,\varnothing,q_{new})\}\cup \{(q_{new}, P, \gamma \land (\gamma_I), L, Z, q) \mid \exists (q_{0\psi_2}, P, \gamma, L, Z, q) \in \Delta_{\psi_2}\}$, where $\gamma_I$ is the clock constraint that checks if $z_X$ is inside the interval $I$. This automaton is basically the same than the automaton for $\psi_1 \sq \psi_2$, except that it adds a time constraint between the two formulas.
		As before, we assume that the set of states and clocks are disjoint.
		\smallskip
		
		We assume by induction that $\semaux{\psi_1}(\Stream) = \semaux{\cT_{\psi_1}}(\Stream)$ and $\semaux{\psi_2}(\Stream) = \semaux{\cT_{\psi_2}}(\Stream)$, will prove that $\semaux{\psi_1\xsq{I}\psi_2} = \semaux{\cT_{\varphi}}(\Stream)$.
		
		Let $(i, j, \mu) \in \semaux{\cT_{\varphi}}(\Stream)$. Then, there exists an accepting run over $\Stream$, $\rho = (q_{0}, \nu_0) \xrightarrow{P_i, \gamma_i / L_i, Z_i} \dots \xrightarrow{P_j, \gamma_j / L_j, Z_j}(q_f, \nu_f)$. By construction, we know that the described run has the following structure: $\rho: (q_{0\psi_1}, \nu_0)\xrightarrow{P_i, \gamma_i / L_i, Z_i} \dots \xrightarrow{P_k, \gamma_k / L_k, Z_k} (q_{new}, \nu_k)\dots(q_{new}, \nu_{k'-1})\xrightarrow{P_{k'}, \gamma_{k'} / L_{k'}, Z_{k'}}(q_{k'}, \nu_{k'})\dots \xrightarrow{P_j, \gamma_j / L_j, Z_j} (q_f, \nu_f)$, such that $\nu_{k'}(z_X)-\nu_{k}(z_X) \in I$. Then, we know that there exist accepting runs $\rho_1: (q_{0\psi_1}, \nu_0)\xrightarrow{P_i, \gamma_i / L_i, Z_i} \dots \xrightarrow{P_k, \gamma_k / L_k, Z_k} (q_{f}, \nu_f)$ over $\psi_1$ and $(q_{0\psi_2}, \nu_0)\xrightarrow{P_{k'}, \gamma_{k'} / L_{k'}, Z_{k'}}(q_{k'}, \nu_{k'})\dots \xrightarrow{P_j, \gamma_j / L_j, Z_j} (q_f, \nu_f)$ over $\psi_2$. As such, we know that $C_1 = (i, k, \mu_1) \in \semaux{\psi_1}(\Stream)$ and $C_2 = (k', j, \mu_2) \in \semaux{\psi_2}(\Stream)$. As $t_{\cstart(C_2)} - t_{\cend(C_1)} \in I$, then $C \in \semaux{\varphi}(\Stream)$.
		
		On the other side, let's assume that $C = (i, j, \mu) \in \semaux{\varphi}(\Stream)$. Then we know that there exist $C_1 = (i, k, \mu_1) \in \semaux{\psi_1}(\Stream)$ and $C_2 = (k', j, \mu_2) \in \semaux{\psi_2}(\Stream)$ such that $t_{\cstart(C_2)} - t_{\cend(C_1)} \in I$. Then, there exist accepting runs $\rho_1: (q_{0\psi_1}, \nu_0) \xrightarrow{P_i, \gamma_i / L_i, Z_i}\dots \xrightarrow{P_k, \gamma_k / L_k, Z_k}(q_{k}, \nu_k)$ over $\cT_{\psi_1}$ and $\rho_2: (q_{0\psi_2}, \nu_0) \xrightarrow{P_{k'}, \gamma_{k'} / L_{k'}, Z_{k'}}\dots \xrightarrow{P_j, \gamma_j / L_j, Z_j}(q_{j}, \nu_j)$ over $\cT_{\psi_2}$. As we know that $t_{\cstart(C_2)} - t_{\cend(C_1)} \in I$, then we know that the run over $\Stream$ $\rho: (q_{0\psi_1}, \nu_0) \xrightarrow{P_i, \gamma_i / L_i, Z_i} \dots (q_{new}, \nu_k) \dots (q_{new}, \nu_{k'})$ $ \dots \xrightarrow{(P_j, \gamma_j / L_j, Z_j)} (q_f, \nu_j)$ is an accepting run over $\cT_{\varphi}$.
		
		\item Let $\varphi = \psi\xssq{I}\psi$, then $\cT_\varphi = (Q_{\psi_1} \cup Q_{\psi_2}\cup \{q_{new}\}, \pset, \xset, \zset_{\psi_1}\cup \zset_{\psi_2} \cup \{z_X\}, \Delta_\varphi,  q_{0\psi_1}, F_{\psi_2} )$ where $z_X$ is a new clock, $q_{new}$ is a new state, and $\Delta_\varphi = \Delta_{\psi_1} \cup \Delta_{\psi_2} \cup \{(p,P,\gamma,L,Z\cup \{z_X\},q_{new}) \mid (p,P,\gamma,L,Z,q_f)\in \Delta_{\psi_1} \land q_f\in F_{\psi_1}\} \cup \{(q_{new}, P, \gamma \land \gamma_I, L, Z, q) \mid \exists (q_{0\psi_2}, P, \gamma, L, Z, q) \in \Delta_{\psi_2}\}$. This automaton is basically the same than the automaton for $\psi_1 \ssq \psi_2$, except that it adds a time constraint between the two formulas.  The proof for this operator is analogous to the one for operator $\xsq{I}$.
		\smallskip
		
		\item Let $\varphi = \psi\xks{I}$, then $\cT_\varphi = (Q_\psi \cup \{q_{new}\}, \pset, \xset, \zset_{\psi}\cup \{z_X\}, \Delta_\varphi,  q_{0\psi}, F_\psi)$ where $q_{new}$ is a fresh state, $z_X$ is a new clock, $\Delta_\varphi = \Delta_\psi \cup \{(p,P,\gamma,L,Z\cup \{z_X\},q_{new}) \mid \exists q' \in F_\psi.(p,P,\gamma,L,Z,q') \in \Delta_\psi\} \cup \{(q_{new},P,\TRUE,\varnothing,\varnothing,q_{new})\} \cup \{(q_{new},P,\gamma\land \gamma_I,L,Z,q) \mid (q_{0\psi},P,\gamma,L,Z,q) \in \Delta_\psi\} \cup \{(q_{new},P,\gamma \land (0 \in I),L,Z\cup \{z_X\},q_{new}) \mid \exists q_f\,.\, (q_{0\psi},P,\gamma,L,Z,q_f) \in \Delta_\psi\}$, $\gamma_I$ is the clock condition that checks if the value of the clock $z_X$ is in the interval $I$ and $(0 \in I)$ is $\TRUE$ is the expression evaluates to true and $\FALSE$ otherwise.
		
		\smallskip
		
		Let $\varphi = \psi \xks{I}$. We assume by induction that $\semaux{\psi}(\Stream) = \semaux{\cT_{\psi}}(\Stream)$, and will prove that $\semaux{\psi\xks{I}}(\Stream) = \semaux{\cT_{\varphi}}(\Stream)$. Let $(i, j, \mu)\in \semaux{\psi\xks{I}}(\Stream)$. Then, by construction, either $(i, j, \mu)\in \semaux{\psi}(\Stream)$ or $(i, j, \mu)\in \semaux{\psi\xsq{I}\psi\xks{I}}(\Stream)$. We will continue this proof using induction. As a base case, if $(i, j, \mu)\in \semaux{\psi}(\Stream)$, then by initial assumption $(i, j, \mu)\in \semaux{\cT_\varphi}(\Stream)$, because the transitions from $\cT_{\psi}$ are preserved in $\cT_{\varphi}$. We will then assume that if $(i, k, \mu_1)\in \semaux{\psi}(\Stream)$ and $(k', j, \mu_2) \in \semaux{\psi\xks{I}}(\Stream)$, then $(i, j, \mu) \in \semaux{\cT_{\varphi}}(\Stream)$. We then know that there exists runs $(q_0, \nu_0) \xrightarrow{P_i, \gamma_i / L_i, Z_i} \dots \xrightarrow{P_k, \gamma_k / L_k, Z_k} (q_f, \nu_k)$ over $\cT_{\psi}$ and $(q_0, \nu_0)\xrightarrow{P_{k'+1},  \gamma_{k'+1} / L_{k'+1}, Z_{k'+1}}\dots\xrightarrow{P_j, \gamma_j / L_j, Z_j} (q_f, \nu_j)$ over $\cT_{\varphi}$, with $k < k'$, and that $t_{k'} - t_{k} \in I$. Finally, if $(i, j, \mu)\in \semaux{\psi\xsq{I}\psi\xks{I}}(\Stream)$, then there exists a run $\rho: (q_0, \nu_0) \xrightarrow{P_i, \gamma_i/ L_i, Z_i} \dots \xrightarrow{P_k, \gamma_k / L_k, Z_k} (q_{new}, \nu_k)\xrightarrow{P_{k+1},\gamma_{k+1} / L_{k+1}, Z_{k+1}}\dots\xrightarrow{P_{k'}, \gamma_{k'}/ L_{k'}, Z_{k'}}(q_{k'}, \nu_{k}')\xrightarrow{P_{k'+1}, \gamma_{k'+1} / L_{k'+1}, Z_{k'+1}}\dots\xrightarrow{P_j, \gamma_j/ L_j, Z_j} (q_f, \nu_j)$. Therefore $(i, j, \mu)\in \semaux{\cT_{\psi\xsq{I}\psi\xks{I}}}(\Stream)$.
		
		On the other hand, let $C=(i, j, \mu)\in \semaux{\cT_{\varphi}}(\Stream)$. Then, there exists a run $\rho: (q_0, \nu_0) \xrightarrow{P_i, \gamma_i/ L_i, Z_i} \dots \xrightarrow{P_j, \gamma_j / L_j, Z_j} (q_f, \nu_f)$. We will prove the result by induction. Let's first assume that the run doesn't pass through the state $q_{new}$. Then, the automaton passes only through states that were present in $\cT_\psi$, so it is clear that $(i, j, \mu) \in \semaux{\psi}(\Stream) \subseteq \semaux{\varphi}(\Stream)$. Now let's assume the run over $\Stream$ passes $n$ times through the new state $q_{new}$. We can take from the stream $n+1$ substreams, such that each subevent $(k, k', \mu') \in \semaux{\cT_\psi}(\Stream)$ doesn't pass through $q_{new}$ (by construction). Let's call $k$ the last index of the first subevent and $k'$ the first index of the last subevent. We can then take two sub-events from $C$ , the first one being $C_1 = (i, k, \mu_1)$ and the second one being $C_2 = (k', j, \mu_2)$, where $k < k'$. We know by construction that $\nu_{k'-1}(z_X) - \nu_k(z_X) + \delta t_{k'} \in I$. As such, we know that $C_1 \in \semaux{\cT_{\psi}}(\Stream)$ and $C_2 \in \semaux{\cT_{\varphi}}(\Stream)$, but its run passes at most $n-1$ times through $q_{new}$. That said, $C_1 \in \semaux{\psi}(\Stream)$, $C_2\in \semaux{\psi\xks{\in I}}(\Stream)$, and $t_{k'}-t_{k} \in I$. We conclude that $C \in \semaux{\varphi}(\Stream)$.

		\item If $\varphi = \psi \timesks{I}$, then $\cT_\varphi = (Q_\psi \cup \{q_{new}\}, \pset, \xset, \zset_{\psi}\cup \{z_X\}, \Delta_\varphi, q_{0\psi}, F_\psi)$ where $q_{new}$ is a fresh state, $z_X$ is a new clock, $\Delta_\varphi = \Delta_\psi \cup \{(p,P,\gamma,L,Z\cup \{z_X\},q_{new}) \mid \exists q' \in F_\psi.(p,P,\gamma,L,Z,q') \in \Delta_\psi\} \cup \{(q_{new},P,\gamma\land \gamma_I,L,Z,q) \mid (q_{0\psi},P,\gamma,L,Z,q) \in \Delta_\psi\} \cup \{(q_{new},P,\gamma \land (0 \in I),L,Z\cup \{z_X\},q_{new}) \mid \exists q_f\,.\, (q_{0\psi},P,\gamma,L,Z,q_f) \in \Delta_\psi\}$, $\gamma_I$ is the clock constraint that checks if the valuation of the clock $z_X$ is in $I$ and $(0 \in I)$ is $\TRUE$ if the expression evaluates to true and $\FALSE$ otherwise. The proof for this operator is analogous to the one for operator $\xks{I}$.
	\end{itemize}	
	
\end{proof}

\subsection*{Proof of Proposition~\ref{pro:cea2cel}}
We will establish a method to build a timed CEL query equivalent to the language of a timed CEA $\cT = (Q, \pset, \xset, \zset, \Delta, q_0, F)$. For that matter, we will begin introducing the definitions of \emph{disjunction-free} and \emph{strongly-deterministic}, adapted from the definitions introduced in \cite{AsarinCM2002}. A timed CEA $\cT = (Q, \pset, \xset, \zset, \Delta, q_0, F)$ is \emph{disjunction-free} if for every transition $(p, P, \gamma, L, Z, q) \in \Delta$, we have that $\gamma$ is a formula with no disjunctions. A timed CEA $\cT = (Q, \pset, \xset, \zset, \Delta, q_0, F)$ is \emph{strongly-deterministic} if for every two transitions $(p, P_1, \gamma_1, L_1, Z_1, q_1)$, $(p, P_2, \gamma_2, L_2, Z_2, q_2) \in \Delta$ where $P_1 \cap P_2 \ne \varnothing$, we have that $L_1 \ne L_2$, namely, the determinism only depends on the value of the next tuple no matter the values of the clocks.

\begin{lemma}
	Any timed CEA $\cT = (Q, \pset, \xset, \zset, \Delta, q_0, F)$ can be converted into an equivalent disjunction-free automaton.
\end{lemma}

\begin{proof}
	The first step is to convert every transition guard to disjunctive normal form (DNF) $\gamma = \gamma_1 \lor\dots\lor \gamma_k$ where each $\gamma_i$ is a conjunction. We then replace every transition $\tau = (q, P, \gamma, L, Z, q')$ by k transitions of the form $(q, P, \gamma_i, L, Z, q')$, for $i \in [1..k]$. Furthermore, every transition will have a clock constraint equivalent to $z_X \le \varepsilon_1 \land z_X \ge \varepsilon_2$, where $\varepsilon_1 \in \bbQplus \cup \{+\infty\}$ and $\varepsilon_2 \in \bbQplus \cup \{-\infty\}$, which we will use instead of the complex clock constraint. This assumption will be use later in other constructions where $z_X \leq +\infty$ or $-\infty \leq z_X$ are always true.
\end{proof}

From now on, in this section, we will only work with disjunction-free automata. To successfully convert a timed CEA into a timed CEL formula, we will start by converting a single-clock timed CEA into a timed CEL formula. This will take several steps, that we will go through next.

Let $\rho$ be a run over $\cT$, and $\tau_1\tau_2\dots \tau_n$ be the transitions taken in that run. A single clock timed CEA is k-check-bounded if for every consecutive transitions $\tau_i, \dots \tau_j$ with $k+1$ transitions where the clock condition is not trivial, then there exists $i \le l < j$ such that $\tau_l$ resets the clock.

Next, we prove that we are able to convert the single-clock timed CEA into a k-check-bounded timed CEA.

\begin{lemma}\label{lemma:check-bounded}
	For every single-clock timed CEA $\cT = (Q, \pset, \xset, \zset, \Delta, q_0, F)$, there exists an equivalent $k$-check-bounded timed CEA for some $k$.
\end{lemma}

\begin{proof}
	The main idea behind this demostration is the following: Let $\rho$ be a run of $\cT$ over $w = w_1 \dots w_n$. Let $i$ be an index where the clock $z_X$ is resetted, and $j$ be the next index where it is resetted. Let $z_X \le \varepsilon_1, \dots, z_X \le \varepsilon_m$ be the constraints checked, in order of appearance. If $\varepsilon_k \le \varepsilon_{k'}$, with $k' < k$, then a run that satisfies $\varepsilon_k$ also satisfies $\varepsilon_{k'}$. This is because the clocks can only increase their value and there are no resets. Then, we can do the following procedure: We define $m_1 = \min \{\varepsilon_i \mid i \in [1, m]\}$, the lower clock constraint. We then define $p_1 = \max\{i \mid \varepsilon_i = m_1\}$, the latest appeareance of $\varepsilon_1$. This is the first clock constraint that we will keep in our updated automaton, because it being true implies all of the previous clock constraints. For our next clock constraint, we will define $m_2 = \min \{\varepsilon_i \mid i \in (p_1, m]\}$, the lower clock constraint which is checked after the index $p_1$, and $p_2 = \max \{i \mid \varepsilon_i = m_2\}$. This will be the second clock constraint that we will check in our run, because it is the next one that provides information that is not redundant. This way, we can continue to trim our clock constraints, finally getting a run that has a limited number of constraint checks before a next clock reset.
	
	To build our equivalent automaton, we need to define our sets of clock constraints:
	
	\begin{itemize}
		\item $C_\leq = \{\varepsilon \mid z_X < \varepsilon \text{ or } z_X \le \varepsilon \text{ is a subformula of a clock condition in $\Delta$}\}$
		\item $C_\geq = \{\varepsilon \mid z_X > \varepsilon \text{ or } z_X \ge \varepsilon \text{ is a subformula of a clock condition in $\Delta$}\}$
	\end{itemize}
	
	Note that equal clock conditions are not added. That is because a clock condition $z_X = c$ is treated as a conjunction of clock conditions, $z_X \le c \land z_X \ge c$. That being defined, we will construct the automaton. $\cT_{cb} = (Q', \pset, \xset, \zset, \Delta', q_0', F')$, with: 
	
	\begin{itemize}
		\item $Q' = Q \times (C_{\leq} \cup \{-\infty\}) \times \{<, \le\} \times (C_{\leq} \cup \{\infty\})\times \{<, \le\} \times (C_{\geq} \cup \{-\infty\}) \times \{>, \ge\}$
		\item $q_0' = (q_0, -\infty, <, \infty, \le, -\infty, \ge)$, and
		\item $F' = F \times (C_{\leq}\cup \{-\infty\}) \times \{<, \le\} \times \{\infty\} \times \{\le\} \times (C_{\geq}\cup \{-\infty\})\times \{>, \ge\}$
	\end{itemize}
	
	Each state is added three clock condition bounds, where the first corresponds to the last checked lower-equal clock condition, the second corresponds to the last skipped lower-equal clock condition and the third corresponds to the last checked greater-equal clock condition. We add these guards so we can skip the redundant transitions. The transitions next are for the greater equal and lower equal cases separately, but it is easy to merge them into one automaton. The full proof is omitted because of readability.
	
	We define the operation $(\varepsilon_1, \sim_1) < (\varepsilon_2, \sim_2)$ if either $\varepsilon_1 < \varepsilon_2$, or $\varepsilon_1 = \varepsilon_2$, $\sim_1$ is $<$ and $\sim_2$ is $\le$.	
	Also, $(\varepsilon_1, \sim_1) > (\varepsilon_2, \sim_2)$ if either $\varepsilon_1 > \varepsilon_2$, or $\varepsilon_1 = \varepsilon_2$, $\sim_1$ is $>$ and $\sim_2$ is $\ge$.
	
	We will begin defining $\Delta_\le$. The transitions it contains take six forms. The first ones are transitions of the form $((p, \varepsilon_c, \sim_c, \varepsilon_m, \sim_m), P, \true, L, \varnothing, (q, \varepsilon_c, \sim_c, \varepsilon_{m}', \sim_m'))$, where $(p, P, z_X \le \varepsilon, L, \varnothing, q) \in \Delta$, $(\varepsilon_{c}, \sim_c) < (\varepsilon, \le)$, and $(\varepsilon_{m}', \sim_m') = \min \{(\varepsilon, \le), (\varepsilon_{m}, \sim_m)\}$.
	The second kinds of transitions are transitions of the form $((p, \varepsilon_c, \sim_c, \varepsilon_m, \sim_m), P, \true, L, \varnothing, (q, \varepsilon_c, \sim_c, \varepsilon_{m}', \sim_m'))$, where $(p, P, z_X < \varepsilon, L, \varnothing, q) \in \Delta$, $\varepsilon_{c} < \varepsilon$, and $(\varepsilon_{m}', \sim_m') = \min \{(\varepsilon, <), (\varepsilon_{m}, \sim_m)\}$. These types of transitions skip the clock constraint, but save it so that a later one implies it. 
	The third are of the form $((p, \varepsilon_{c}, \sim_c, \varepsilon_{m}, \sim_m), P, z_X \le \varepsilon, L, \varnothing, (q, \varepsilon, \le, \infty, \le))$, where $(p, P, z_X \le \varepsilon, L, \varnothing, q) \in \Delta$, where $(\varepsilon_{c}, \sim_c) < (\varepsilon, \le)$, $(\varepsilon, \le) \le (\varepsilon_m, \sim_m)$.
	The fourth are of the form $((p, \varepsilon_{c}, \sim_{c}, \varepsilon_{m}, \sim_{m}), P, z_X < \varepsilon, L, \varnothing, (q, \varepsilon, <, \infty, \le))$, where $(p, P, z_X < \varepsilon, L, \varnothing, q) \in \Delta$, where $\varepsilon_{c} < \varepsilon \le \varepsilon_{m}$. These types of transitions force the clock constraint, ensuring that it implies the skipped constraints and clear the skipped constraint. 
	The fifth ones are of the form $((p, \varepsilon_{c}, \sim_c, \varepsilon_{m}, \sim_{m}), P, z_X \le \varepsilon, L, \{z_X\}, (q, -\infty, <, \infty, \le))$, where $(p, P, z_X \le \varepsilon, L, \{z_X\}, q) \in \Delta$, $(\varepsilon_{c}, \sim_c) < (\varepsilon, \le)$, $(\varepsilon, \le) \le (\varepsilon_m, \sim_m)$. 
	The sixth ones are of the form $((p, \varepsilon_{c}, \sim_c, \varepsilon_{m}, \sim_m), P, z_X < \varepsilon, L, \{z_X\}, (q, -\infty, <, \infty, \le))$ where $(p, P, z_X < \varepsilon, L, \{z_X\}, q) \in \Delta$, where $\varepsilon_{c} < \varepsilon\le \varepsilon_{m}$. These types of transitions are the same than the previous ones, but also resets the clock and restores the stored constraints.
	
	Now we will define $\Delta_\ge$. The transitions it contains take eight forms. The first are transitions of the form $((p, \varepsilon_{c}, \sim_c), P, z_X \ge \varepsilon, L, \varnothing, (q, \varepsilon, \ge))$, where $(p, P, z_X \ge \varepsilon, L, \varnothing, q) \in \Delta$ and $\varepsilon > \varepsilon_{c}$.
	The second are transitions of the form $((p, \varepsilon_{c}, \sim_c), P, z_X > \varepsilon, L, \varnothing, (q, \varepsilon, >))$, where $(p, P, z_X > \varepsilon, L, \varnothing, q) \in \Delta$, $(\varepsilon, >) > (\varepsilon_{c}, \sim_c)$. These types of transitions check the constraint and store it so that the next constraint is forced to be greater.
	The third are transitions of the form $((p, \varepsilon_{c}, \sim_c), P, \true, L, \varnothing, (q, \varepsilon_{c}, \sim_c))$, where $(p, P, z_X \ge \varepsilon, L, \varnothing, q) \in \Delta$, and $\varepsilon\le \varepsilon_{c}$.
	The fourth are transitions of the form $((p, \varepsilon_{c}, \sim_{c}), P, \true, L, \varnothing, (q, \varepsilon_{c}, \sim_c))$, where $(p, P, z_X > \varepsilon, L, \varnothing, q) \in \Delta$, and $(\varepsilon, >) \le (\varepsilon_{c}, \sim_c)$. These types of transitions skip the clock constraint, at it is already implied by the previous checks.
	The fifth are transitions of the form $((p, \varepsilon_{c}, \sim_c), P, z_X \ge \varepsilon, L, \{z_X\}, (q, - \infty, \ge))$, where $(p, P, z_X \ge \varepsilon, L, \{z_X\}, q) \in \Delta$ and $\varepsilon > \varepsilon_{c}$.
	The sixth are transitions of the form $((p, \varepsilon_{c}, \sim_c), P, z_X > \varepsilon, L, \{z_X\}, (q, - \infty, \ge))$, where $(p, P, z_X > \varepsilon, L, \{z_X\}, q) \in \Delta$, $(\varepsilon, >) > (\varepsilon_{c}, \sim_c)$.
	The seventh are transitions of the form $((p, \varepsilon_{c}, \sim_c), P, \true, L, \{z_X\}, (q, - \infty, \ge))$, where $(p, P, z_X \ge \varepsilon, L, \{z_X\}, q) \in \Delta$, and $\varepsilon\le \varepsilon_{c}$.
	The eighth are transitions of the form $((p, \varepsilon_{c}, \sim_{c}), P, \true, L, \{z_X\}, (q, - \infty, \ge))$, where $(p, P, z_X > \varepsilon, L, \{z_X\}, q) \in \Delta$, and $(\varepsilon, >) \le (\varepsilon_{c}, \sim_c)$. These are analogous to the previous transitions, but also reset the saved constraint because of the clock reset.
	
	One can check that $\cT_{cb}$ is $k$-check bounded with $k \leq |C_\leq| + |C_\geq|$ and the proof of correctness follows from the construction.
\end{proof}

By Lemma~\ref{lemma:check-bounded}, in the sequel we will only work with \emph{k-check-bounded} timed CEA when they are single-clock. Next, we will introduce the notion of \emph{reset-distinct}. A reset-distinct timed CEA is a timed CEA that does not have transitions that reset a clock and check a clock condition at the same time. We will show that any reset-distinct single-clock timed CEA can be converted into a formula in CEL. 
\begin{lemma}
	A reset-distinct and $k$-check bounded single-clock timed CEA $\cT$ has an equivalent formula in timed CEL $\phi$.
\end{lemma}
\begin{proof}
	We will create a timed CEL formula from the timed CEA described earlier. Let $\cT = (Q, \pset, \xset, \zset, \Delta, q_0, F)$ be a reset-distinct and $k$-check bounded timed CEA. We will introduce the following variables.
	First, we have the variable $S_{pq}^O$, that creates an expression in timed CEL starting from state $p$ and going to state $q$, passing only through states in $O \subseteq Q$, without any resets or clock-checks. The expressions are build recursively from the following formulas. The base case, for just one transition, $S_{q_iq_j}^{\varnothing} = \kOR_{(q_i, P, \true, L, \varnothing, q_j)\in \Delta}\PROJ{L}(P \kAS L)$, and the inductive case, $S_{q_iq_j}^{O \cup \{q_l\}} = S_{q_iq_l}^O \ssq S_{q_lq_l}^O\sks \ssq S_{q_lq_j}^O \kOR S_{q_iq_l}^O \ssq S_{q_lq_j}^O \kOR S_{q_iq_j}^O$. Here, $P \kAS L$ is an abuse of notation, denoting the use of operator $\kAS$ with every variable in $L$.
	
	Next, we introduce the formula for the subsets of the automaton that start from state $p$, pass first through with a transition that has a reset, and then pass through $h$ time windows, $R^h_{pq}$, arriving to state $q$. The base case ($0$ time windows) is $R_{q_iq_j}^0 = \kOR_{(q_i, P, \true, L, \{z\}, q_l)\in \Delta} \PROJ{L}(P \kAS L)\ssq S_{q_lq_j}^Q$, and the inductive case is divided into the cases that pass through $h+1$ time windows and end with a non-time-window transition, $R_{n,q_iq_j}^{h+1} = \kOR_{(q_l, P, z \in I, L, \varnothing, q_l)\in \Delta} \left ( \twin{R_{q_iq_l}^{h} \ssq \PROJ{L}(P \kAS L)}{I} \ssq S_{q_lq_j}^{Q} \right ) $ and into the cases that pass through $h+1$ time windows and end with a time-window transition, $R_{t,q_iq_j}^{h+1} = \kOR_{(q_l, P, z \in I, L, \varnothing, q_j)\in \Delta} \left ( \twin{R_{q_iq_l}^{h} \ssq \PROJ{L}(P \kAS L)}{I} \right )$. Then, $R_{q_iq_j}^{h+1} = R_{n,q_iq_j}^{h+1} \kOR R_{t,q_iq_j}^{h+1}$.
	
	Next, we introduce the formula for one or more of the previously introduced reset-starting parts of the automaton. The formula $F_{pq}^O$ denotes the contiguous reset-starting parts of the automata starting in $p$, ending in $q$ and having as start or end states of each subpart only states in $O$. The base case is $F_{q_iq_j}^\varnothing = \kOR_{h=0}^{k}R_{q_iq_j}^{h}$, the union of every combination of time windows, and the inductive case is $F_{q_iq_j}^{O \cup \{q_l\}} = F_{q_iq_l}^O \ssq F_{q_lq_l}^O\sks \ssq F_{q_lq_j} \kOR F_{q_iq_l}^O \ssq F_{q_lq_j}\kOR F_{q_iq_j}^O$, similar as the $S_{pq}^O$ formula. Finally, the formula $\varphi$ for the full automaton is stated as 
	\[
	\varphi = \kOR_{q_l\in Q} \left (\kOR_{q_f \in F} \left (S_{q_0q_l}^Q \ssq F_{q_lq_f}\right ) \right ) \kOR \kOR_{q_f \in F}\left ( S_{q_0q_f} \kOR F_{q_0q_f}\right )
	\]
\end{proof}

Now that we know that any reset-distinct single-clock timed CEA can be converted into a CEL formula, we will convert the \emph{timed CEA} into a conjunction of reset-distinct timed CEAs. This transformation will allow us to easily convert the \emph{timed CEA} into a timed CEL formula.

\begin{lemma}\label{lemma:reset-distinct-AND}
	Let $\cT = (Q, \pset, \xset, \{z\}, \Delta, q_0, F)$ be a single-clock timed CEA, where $\xset'$ is the set of variables used by $\cT$. There exist two reset-distinct timed CEAs $\cT_1$ and $\cT_2$ such that, if $\phi_1$ and $\phi_2$ are the equivalent timed CEL formulas of $\cT_1$ and $\cT_2$, respectively, then $L(\PROJ{\xset'}(\phi_1\kAND \phi_2)) = L(\cT)$.
\end{lemma}

\begin{proof}
	First we will create a new strongly-deterministic timed CEA $\cT'$ based on automaton $\cT$. We will do so by replacing any transition $(q, P, \gamma, L, Z, q')$ by $(q, P, \gamma, L\cup \{X_i\}, Z)$, where $X_i$ is a new variable, choosing a different $i$ for every transition going out of the same state $q$. This new automaton is not equivalent to the initial automaton, but we can obtain it easily by applying the operation $L(\PROJ{\xset'}(\phi')) = L(\cT)$, where $\phi'$ is the formula for automaton $\cT'$ and $\xset'$ is the set of variables of the initial automaton.
	
	Having our new strongly-deterministic timed CEA $\cT'$, we will separate it into two different timed CEAs $\cT_1$ and $\cT_2$ with respective formulas $\phi_1$ and $\phi_2$, such that $L(\phi_1 \kAND \phi_2) = L(\phi')$. 
	The idea is to separate the resets of $\cT$ between even and odd resets during a run. Then $\cT_1$ will do the resets and checks at even positions and $\cT_2$ will do the same but at odd positions. 
	The first automaton is $\cT_1 = (Q\times \{1, 2\}, \pset, \xset, \{z\}, \Delta_1, (q_0, 1), F \times \{1, 2\})$, where:
	\[
	\begin{array}{rclc}
		\Delta_1 & = & \{((p, 1), P, \gamma, L, \varnothing, (q, 2))\mid (p, P, \gamma, L, \{z\}, q) \in \Delta\} &  \cup \\
		& & \{((p, 2), P, \true, L, \{z\}, (q, 1))\mid (p, P, \gamma, L, \{z\}, q) \in \Delta\} & \cup \\
		& & \{((p, 1), P, \gamma, L, \varnothing, (q, 1))\mid (p, P, \gamma, L, \varnothing, q) \in \Delta\} & \cup \\
		& &  \{((p, 2), P, \true, L, \varnothing, (q, 2))\mid (p, P, \gamma, L, \varnothing, q) \in \Delta\}.
	\end{array}
	\] 
	This automaton only checks time-windows from resets that are in an even position on the run.
	
	The second automaton is $\cT_2 = (Q\times \{1, 2\}, \pset, \xset, \{z\}, \Delta_2,(q_0, 1), F \times \{1, 2\})$, where:
	\[
	\begin{array}{rclc}
		\Delta_2 & = & \{((p, 1), P, \true, L, \{z\}, (q, 2))\mid (p, P, \gamma, L, \varnothing, q) \in \Delta\} & \cup \\
		& & \{((p, 2), P, \gamma, L, \varnothing, (q, 1))\mid (p, P, \gamma, L, \{z\}, q) \in \Delta\} & \cup \\
		& & \{((p, 1), P, \true, L, \varnothing, (q, 1))\mid (p, P, \gamma, L, \varnothing, q) \in \Delta\} & \cup \\
		& &  \{((p, 2), P, \gamma, L, \varnothing, (q, 2))\mid (p, P, \gamma, L, \varnothing, q) \in \Delta\}. 
	\end{array}
	\]
	This automaton only checks time-windows from resets that are in an odd position on the run. These new timed CEAs are reset-distinct because there are no transitions that have a clock reset and a non-trivial clock condition, and the proof of equivalence is straightforward.
\end{proof}

It is important to note that the resulting reset-distinct automata are still k-check-bounded, strongly deterministic, have a single clock, and are disjunction free, because the transformations do not alter these properties.

The previous lemmas allow us to introduce the following lemma on transforming timed CEAs to timed CEL formulas:

\begin{lemma}
	Any single-clock timed CEA $\cT$ can be converted into an equivalent timed CEL formula.
\end{lemma}

The proof is straightforward. We will now transform the several-clock timed CEA into a timed CEL formula. Finally, using all the previous lemmas, we will prove the Proposition~\ref{pro:cea2cel}.

\begin{proof}[Proof of Proposition~\ref{pro:cea2cel}]
	We begin by transforming the timed CEA $\cT$ into a strongly deterministic automaton $\cT'$, by using the same method introduced earlier in Lemma~\ref{lemma:reset-distinct-AND}. Note that if $\semaux{\cT} = \semaux{\phi_\cT}$ and $\semaux{\cT'} = \semaux{\phi_{\cT'}}$, then $\semaux{\PROJ{\xset'}(\phi_{\cT'})} = \semaux{\phi_\cT}$, where $\xset'$ are the variables used in the automaton $\cT$. To be able to transform the automaton into a timed CEL formula, we will separate the new automaton $\cT'$ into several single-clock timed CEA $\cT_1'\dots \cT_n'$, such that the ``conjunction'' of $\sem{\cT_1'}, \dots, \sem{ \cT_n'}$ (i.e., $\kAND$) is equivalent to $\sem{\cT'}$. This can be done by separating the automaton into $n$ single-clock timed CEA, one for each different clock. Each automaton $\cT_i'$ will change all checks over a clock different to $z_i$ by $\TRUE$, leaving only the checks using the clock $z_i$. It is clear that the intersection of every automata has the same semantics as $\cT'$.
	
	We know that each single-clock automaton $\cT_1, \dots \cT_n$ can be converted into a timed CEL formula $\varphi_1, \dots \varphi_n$. So finally, the timed CEL formula for the full automaton will be:
	\[
	\varphi = \PROJ{\xset'}(\kAND_{i=1}^{n} \ \ \varphi_i)
	\]
	Note that the operators used for the construction of the formula are $\kOR$, $\PROJ{L}$, $\ssq$, $\sks$, $\kAS$, $\twin{...}{I}$ and $\kAND$. This does not use any of the operators in $\{;, +, \timesq{\sim c}, \timessq{I}, \timeks{I}, \timesks{I}\}$. 
\end{proof} 	
	\section{Proofs from Section~\ref{sec:determinization}} \label{app:determinization}

\subsection*{Proof of Theorem~\ref{theo:determinization}}
First we will will introduce a new concept, analogous to the definition proposed in \cite{GrezRUV21}. Let $\cT= (Q, \pset, \xset, \zset, \Delta, q_0, F)$ be a timed CEA, and let $\pset_\Delta = \{P_1, \dots, P_n\}$ be the set of all predicates used in transitions of the automaton. Then, for $S \subseteq [1..n]$, we define the new CEA predicate $P_S = \bigcap_{i \in S} P_i \cap \bigcap_{i \notin S} \eset \backslash {P_i}$. We also name $\etypes{\pset_\Delta} = \{P_S \mid S \subseteq [1..n]\}$. Note that the predicates in $\etypes{P}$ are pairwise exclusive and partition the space.

Furthermore, we will recall an equivalence relation between clock constraints. Let $A = \{\gamma_1, \dots, \gamma_n\}$ be the set of all clock constraints used in the transitions of the automaton. Then, $(\equiv, \mathcal{C}_\zset)$ is an equivalence relation for clock conditions such that for all $\alpha, \beta \in \mathcal{C}_\zset$, $\alpha \equiv \beta$ if and only if for every valuation $\nu$, $\nu \models \alpha \Leftrightarrow \nu \models \beta$. Because of the equivalence relationship properties, we know that if $\alpha \nequiv \beta$, then $[\alpha] \cap [\beta] = \varnothing$.

Afterwards, we will introduce the set $\Gamma = \{\bigwedge_{i=1}^{n} \alpha_i \mid \forall i\ \alpha_i \in \{\gamma_i, \bar{\gamma_i}\}\}$, where $\bar{\gamma_i}$ is the clock condition that satisfies exactly all valuations that are not satisfied by $\gamma_i$. Note that $|\bar{\gamma_i}| \le 2 |\gamma_i|$, as $\bar{\gamma_i}$ can be obtained by changing every $\land$ by a $\lor$ and $\lor$ by $\land$, every $z \le c$ (resp. $z \ge c$) by $z > c$ (resp. $z<c$) and every $z=c$ by $z>c \lor z<c$.
Also note that similar to the set $\etypes{\pset_\Delta}$, the elements of $\Gamma / \equiv$ are pairwise exclusive and partition the space. Next we will proceed to demonstrate Theorem~\ref{theo:determinization}.

\begin{proof}[Proof of Theorem~\ref{theo:determinization}]
	We can build a deterministic timed CEA $\cT'$ based on the timed CEA $\cT$, which we will assume to be disjunction-free (see the proof of Proposition~\ref{pro:cea2cel}). We will construct the equivalent deterministic timed CEA as the following tuple: $\cT' = (Q_{det}, \pset, \xset, \zset, \Delta_{det}, Q_0, F_{det})$, where  $Q_{det} = 2^Q$, $Q_0 = (\{q_0\})$, and $F_{det} = \{D \in Q_{det} \mid \exists q_f \in F .\ q_f \in D\}$.
	For the transition relation, the elements are $(D, P, \alpha, L, Z, D')\in \Delta_{det}$, where $P \in \etypes{\pset_\Delta}$, $[\alpha] \in \Gamma / \equiv$ and $D' = \{q \mid \exists p \in D.\ \exists (p, P', \gamma, L, Z, q) \in \Delta.\ \alpha \models \gamma \land P \subseteq P'\}$. For the clock resets $Z$, we can take any of the clock resets $Z'$ of any of the transitions $(p, P', \gamma, L, Z', q)\in\Delta$ such that $p\in D$,  $\alpha \models \gamma$ and $P \subseteq P'$. As the automaton has synchronous resets, we know that these sets of clocks must be the same for every transition that follow those constraints.
	
	The automaton is deterministic because, as said earlier, $\etypes{\pset_\Delta}$ and $\Gamma / \equiv$ partition space and have pairwise exclusive elements.
	
	Let $\rho: (q_0, \nu_0) \xrightarrow{P_1, \gamma_1 / L_1, Z_1} \dots \xrightarrow{P_n, \gamma_2 / L_n, Z_n} (q_n, \nu_n)$ be an accepting run of a stream $\Stream$ over $\cT$, using transitions $(q_i, P_{i+1}, \gamma_{i+1}, L_{i+1}, Z_{i+1}, q_{i+1})$. By construction, we know that there exists a run (not necessarily accepting) $\rho': (Q_0, \nu_0') \xrightarrow{P_1', \gamma_1'/ L_1, Z_1'} \dots \xrightarrow{P_n', \gamma_n' / L_n, Z_n'} (Q_n, \nu_n')$, using transitions of the form $(Q_i, P_{i+1}', \alpha_{i+1}, Z_{i+1}', Q_{i+1})$.	
	
	Now we will prove using induction that for all $i\le n$, $q_i \in Q_i$ and $\nu_i = \nu_i'$. As a base case, $q_0 \in Q_0$ and $\nu_0 = \nu_0'$. Then, let's assume $q_i \in Q_i$ and $\nu_i = \nu_i'$, and let's prove that $q_{i+1}\in Q_{i+1}$ and $\nu_{i+1} = \nu_{i+1}'$. We know that there exists a transition $(q_i, P_{i+1}, \gamma_{i+1}, L_{i+1}, Z_{i+1}, q_{i+1}) \in \Delta$. Based on this fact, there exists a transition $(Q_i, P_{i+1}', \alpha_{i+1}, L_{i+1}, Z_{i+1}, Q_{i+1})$ such that $P_{i+1}' \subseteq P_{i+1}$, $\alpha_{i+1} \models \gamma_{i+1}$ and $\nu_i' + \delta t_{i+1} \models \alpha_{i+1}$. Using this transition, $q_{i+1} \in Q_{i+1}$. Furthermore, $Z_{i+1}' = Z_{i+1}$, so $\nu_{i+1} = \nu_{i+1}'$, that is, the two runs have the same clock valuations. Finally, $Q_n \in F_{det}$ then the run is an acceptance run.
	
	Let $\rho: (Q_0, \nu_0) \xrightarrow{P_1', \gamma_1' / L_1, Z_1'} \dots \xrightarrow{P_n', \gamma_n' / L_n, Z_n'} (Q_n, \nu_n)$ be an accepting run of a stream $\Stream$ over $\cT'$, using transitions $(Q_i, , P_{i+1}', \alpha_{i+1}, Z_{i+1}', Q_{i+1})$. We will prove that for every $i \le n$ and for every $q_i \in Q_i$, there exists a run $(q_0, \nu_0') \xrightarrow{P_1, \gamma_1/ L_1, Z_1} \dots \xrightarrow{P_i, \gamma_i / L_i, Z_i} (q_i, \nu_i')$ over $\cT$ such that $\nu_i' = \nu_i$.
	
	As a base case, $q_0 \in Q_0$ and $\nu_0' = \nu_0$. Then, we will assume that for every $q_i \in Q_i$, there exists a run $(q_0, \nu_0') \xrightarrow{P_1, \gamma_1 / L_1, Z_1} \dots \xrightarrow{P_i, \gamma_i /L_i, Z_i} (q_i, \nu_i')$ such that $\nu_i' = \nu_i$, and will prove it for $i+1$.
	
	We know that for any $q_{i+1} \in Q_{i+1}$, there exists a transition $(q, P_{i+1}, \gamma_{i+1}, L_{i+1}, Z_{i+1}, q_{i+1})$ such that $q\in Q_i$, $P_{i+1}' \subseteq P_{i+1}$, $\alpha_{i+1}\models\gamma_{i+1}$, and $Z_{i+1} = Z_{i+1}'$. Following that logic, as by assumption $\nu_i = \nu_i'$, it is true that $\nu_i' + \delta t_{i+1}\models \alpha_{i+1} \models \gamma_{i+1}$. Furthermore, the resets are the same, so $\nu_{i+1} = \nu_{i+1}'$, so there is a run $(q_0, \nu_0) \xrightarrow{P_1, \gamma_1/ L_1, Z_1} \dots \xrightarrow{P_i, \gamma_i / L_i, Z_i} (q_i, \nu_i)\xrightarrow{P_{i+1}, \gamma_{i+1}/ L_{i+1}, Z_{i+1}} (q_{i+1}, \nu_{i+1})$.
	
	Finally, as there exists $q_f\in F$ such that $q_f\in Q_n$, then there  exists a run $(q_0, \nu_0) \xrightarrow{P_1, \gamma_1 / L_1, Z_1} \dots \xrightarrow{P_n, \gamma_n/ L_n, Z_n} (q_f, \nu_n)$ over $\cT$, such that $q_f\in F$, concluding that the run is an accepting run.
	
	The size of this automaton is $|\cT'| = |Q_{det}| + \sum_{(D, P, \gamma, L,Z, D')\in \Delta_{det}} |\gamma| + |P|+ |L| + |Z|$. We know that $|Q_{det}| = 2^{|Q|}$. We also know that each one of the transitions $(D, P, \alpha, L, Z, D')\in \Delta_{det}$ have a size $|P| \le \sum_{(p, P', \gamma', L', Z', q)\in \Delta}|P'|+1$ and $|\alpha| \le \sum_{(p, P', \gamma', L', Z', q)\in \Delta}|\gamma'|+1$. Therefore, we can check that:
	\begin{alignat*}{2}
		\sum_{(D, P, \alpha, L,Z, D')\in \Delta_{det}} |\alpha| + |P|+ |L| + |Z| & \le && \sum_{(D, P, \alpha, L,Z, D')\in \Delta_{det}} \ \  \sum_{(p, P', \gamma', L', Z', q)\in \Delta}|\gamma'|+1 \\
		& && \qquad\qquad\qquad+ \sum_{(p, P', \gamma', L', Z', q)\in \Delta}|P'|+1) + |L| + |Z|\\ 
		& = && \sum_{(D, P, \alpha, L,Z, D')\in \Delta_{det}} \ \ \sum_{(p, P', \gamma', L', Z', q)\in \Delta}|\gamma'|+|P'|+2) + |L| + |Z| \\
		& = && |\Delta_{det}| \cdot \sum_{(p, P', \gamma', L', Z', q)\in \Delta}|\gamma'|+|P'|+2 \\
		& && + \sum_{(D, P, \alpha, L,Z, D')\in \Delta_{det}} |L| + |Z|
	\end{alignat*}
	We also know that the following holds.
	\begin{alignat*}{1}
		\sum_{(D, P, \alpha, L,Z, D')\in \Delta_{det}} |L| + |Z|\le &\sum_{(p, P', \gamma, L, Z', q')\in \Delta} \ \  \sum_{(D, P, \alpha L, Z', D') \in \Delta_{det}\mid p\in D} |L| + |Z'|\\
		& \le \sum_{(p, P', \gamma, L, Z', q')\in \Delta} 2^{|Q|-1+2|\Delta|}(|L| + |Z'|)
	\end{alignat*}
	We know that $|\Delta_{det}| = 2^{|Q|+2|\Delta|}$. Summing up:
	
	\begin{alignat*}{1}
		|\cT'| &\le 2^{|Q|} + 2^{|Q|+2|\Delta|} \sum_{(p, P', \gamma', L, Z', q)\in \Delta}|\gamma'|+|P'|+2) + 2^{|Q|-1+2|\Delta|}\sum_{(p, P', \gamma, L, Z', q')\in \Delta} (|L| + |Z'|)\\
		&\le 2^{|Q|} + 2^{|Q|+2|\Delta|}\sum_{(p, P', \gamma', L, Z', q)\in \Delta} (|\gamma'|+|P'| +|L| + |Z'|) \\
		&\le 2^{|Q|} + 2^{|Q|+2|\Delta|}\cdot|\cT|
	\end{alignat*}
	So finally, the size of the automaton is $\mathcal{O}(2^{|Q|} + 2^{|Q| + 2|\Delta|}\cdot|\cT|) \subseteq \mathcal{O}(2^{|Q| + 2|\Delta|}\cdot|\cT|)$.
\end{proof}

\subsection*{Proof of Theorem~\ref{theo:algo-synch-resets}}
In first place, note that all of the clock constraints in the automaton are composed of subformulas $z_i \sim c_i$, where $z_i$ is a clock and $c_i$ is a rational number. With the help of the results proven in \cite{AlurD94}, we can build a timed CEA $\cT'$ by replacing every $z_i \sim c_i$ in every clock constraint by $z_i \sim c_i \cdot d$, where $d$ is the lowest integer such that $c_i \cdot d$ is an integer for all $i$. That way, every comparison of the value of a clock $z_i$ with a number will be with an integer number. We will assume that the automaton will have the form of $\cT'$, as \cite{AlurD94} has proven the equivalence of both timed CEAs.

In second place, we will build the clock regions for the timed CEA $\cT'$ as they are presented in \cite{AlurD94}. First, we introduce a new notation. If $d$ is a rational number, then $\fract(d)$ is the fractional part of $d$, and $\lfloor d \rfloor$ is its integral part. For each $z\in \zset$, we name $c_z$ the greatest integer $c$ such that $z \sim c$ is a subformula of a clock constraint in $\cT'$. We introduce a new equivalence relationship over clock valuations, symbolized by $\approxeq$, where we say that for two clock valuations $\nu$ and $\nu'$, $\nu \approxeq \nu'$ if $\dom{\nu} = \dom{\nu'}$, and all of these hold:

\begin{itemize}
	\item For all $z\in \dom{\nu}$, we have that either $\lfloor \nu(z) \rfloor = \lfloor \nu'(z) \rfloor$, or both $\nu(z), \nu'(z) > c_z$.
	\item For all $z, z' \in \dom{\nu}$ with $\nu(z) \le c_z$ and $\nu(z') \le c_{z'}$, $\fract(\nu(z)) \le \fract(\nu(z'))$ iff  $\fract(\nu'(z)) \le \fract(\nu'(z'))$
	\item For all $z\in \dom{\nu}$ with $\nu(z) \le c_z$, $\fract(\nu(z)) = 0$ iff $\fract(\nu'(z)) = 0$.
\end{itemize}
The clock regions are the equivalence classes clock valuations over the relation $\approxeq$.  We know that the number of clock regions is bounded by $|\zset| \cdot 2^{|\zset|} \cdot \prod_{z \in \zset}(2c_z+2)$. Next, we will prove the $\textsc{PSPACE}$-completeness of the problem.

\begin{proof}[Proof of Theorem~\ref{theo:algo-synch-resets}]
	We will now prove that the problem is $\textsc{PSPACE}$-complete. For that matter, we will show that the opposite problem is $\textsc{NPSPACE}$-complete, that is, showing that there exist two runs such that the resets are different, by proving first $\textsc{NPSPACE}$-membership and then $\textsc{NPSPACE}$-hardness.
	
	We will start by proving $\textsc{NPSPACE}$-membership. Let $\cT$ be the input of a non deterministic turing machine. Then, we will write into the tape a first configuration $(q_0, q_0, R_0)$, where $R_0$ is the clock region that has all clocks in $0$.
	
	For the current configuration $(p, p', R)$ in the tape, we will then nondeterministically choose two transitions $(p, P, \gamma, L, Z, q)$ and $(p', P', \gamma', L', Z', q')$ such that $P \cap P' \ne \varnothing$, there exists a valuation $\nu\in R$ and $t \ge 0$ such that $\nu + t \models \gamma \land \gamma'$ and $L = L'$. If $Z \ne Z'$, we go to an accepting state, else we change the current configuration by $(q, q', R')$, where $R'$ is a clock region chosen nondeterministically such that there exists $\nu+t \in R'$.
	
	It is easy to see that the algorithm uses polynomial space, and it outputs the correct answer as $Z\ne Z'$ iff the timed CEA does not have synchronous resets. We have proven $\textsc{PSPACE}$-membership.
	
	We will then prove $\textsc{PSPACE}$-hardness. We will do so by reducing from the emptyness problem in time automata, a known problem that is $\textsc{PSPACE}$-complete \cite{AlurD94}. Suppose $\cA$ is a time automaton. We will transform it to a timed CEA by preserving its states, create a set of predicates such that every predicate is equivalent to reading a letter of the input word, and preserve its transitions and clocks. The only change will be on the markings of the automaton, on which every transition will mark a different variable $X_i$, for $1 \le i\le |\Delta|$. Afterwards, for every final state $p$ of the timed automaton, we will build two transitions from $p$ to $p$ reading the same predicate and having no clock constraints, but resetting different sets of clocks. In this way, if the final state is reachable in the original timed automaton, then also will be the two exiting transitions of the timed CEA, which will then not have synchronous resets. This way we have proven $\textsc{PSPACE}$-hardness and the problem is $\textsc{PSPACE}$-complete.
\end{proof}

\subsection*{Proof of Theorem~\ref{theo:windowed-sync}}
First, we will introduce the following lemma:
\begin{lemma}
	Let $\phi$ be a simple timed CEL formula. Then, there exists a single clock timed CEA with synchronous resets $\cT_r$ such that $\semaux{\cT_r} = \semaux{\phi}$ and, for every transition $\tau$ of $\cT_r$, $\tau$ resets the clock iff $\tau$ marks with at least one variable.
\end{lemma}
\begin{proof}
	We will prove that for every simple timed CEL formula $\phi$, there exists a single-clock synchronous reset timed CEA $\cT$ such that its clock $z_X$ is only reset in every marking transition of $\cT$. We will prove this result using induction.
	
	As a base case, if $\phi = R$, we do the construction of Proposition~\ref{pro:cel2cea}. We know by construction that the automaton does not have any clocks, then we add one clock $Z_X$ that resets in every marking transition (that is, the only one transition that exists). The equivalence of the automata is trivial.
	
	For the inductive case, we will demonstrate that for every operator the property holds. We assume that $\phi_1, \phi_2$ are is a simple timed CEL formula such that there is a timed CEA with synchronous resets and one clock, that is reset in every marking transition. We will prove it for the aggregated formula with every operator.
	
	For the automaton $\cT_{\phi_1 \kAS X}$, we will follow the construction in the proof of Proposition~\ref{pro:cel2cea} to create an automaton such that $\semaux{\cT_{\phi_1 \kAS X}} = \semaux{\phi_1 \kAS X}$, and we will prove that $\cT_{\phi_1 \kAS X}$ has synchronous resets. Let $\rho_1: (q_0, \nu_0) \xrightarrow{P_1, \gamma_1 / L_1, Z_1} \dots \xrightarrow{P_n, \gamma_n / L_n, Z_n} (q_n, \nu_n)$ and $\rho_2: (q_0', \nu_0') \xrightarrow{P_1', \gamma_1' / L_1, Z_1'} \dots \xrightarrow{P_n', \gamma_n' / L_n, Z_n'}(q_n, \nu_n)$ be partial runs of $\cT_{\phi_1 \kAS X}$ over $\Stream$. We know that the transitions of $\cT_\phi$ are preserved, only changing the markings to include the new variable $X$. Then, for two transitions $(q_{j-1}, P_j, \gamma_j, L_j, Z_j, q_j)$ and $(q_{j-1}', P_j', \gamma_j', L_j, Z_j', q_j')$, we know that $Z_j = Z_j'$. Furthermore, $\cT_{\phi_1 \kAS X}$ as the number of clocks and their resets don't change, we know that the property is preserved.
	
	For the automaton $\cT_{\phi_1 \kFILTER X[P]}$, the proof is analogous as the one for $\phi \kAS X$, but changing the predicates instead of the variable markings.
	
	For the automaton $\cT_{\phi_1\kOR \phi_2}$, we know that there exist $\cT_{\phi_1}$ and $\cT_{\phi_2}$ such that for every run, $z_{X}$, $z_{X'}$ are the clocks that reset in every marked transition. We also know that in the construction, there are no transitions from a state in $\cT_{\phi_1}$ to a state in $\cT_{\phi_2}$ and viceversa. In that way, we know that if a clock is reset in a run passing through the states of $\cT_{\phi_1}$, then it will not affect any run passing through the states of $\cT_{\phi_2}$. Using this fact, we can perform the construction in Proposition~\ref{pro:cel2cea} of $\phi_1'$ and $\phi_2'$, but assuming that $z_{X}=z_{X'}$, that is, the resulting automaton has only one clock, $z_{X}$, that is reset in every marking transition.
	
	For the automaton $\cT_{\phi_1 \kAND \phi_2}$, we do the construction of Proposition~\ref{pro:cel2cea}. We know that the timed CEA $\cT_{\phi_1}$ and $\cT_{\phi_2}$ have one clock each, that is, $z_{X}$, $z_{X'}$ (the clocks that are reset in every marking transition). Then, let $\rho$ and $\rho'$ be two partial runs of $\cT_{\phi_1 \kAND \phi_2}$ is $(q_0, q_0')$ over the stream $\Stream$. The first run is $\rho: ((p_0, q_0), \nu_0) \xrightarrow{P_1, \gamma_1 / L_1, Z_1} \dots \xrightarrow{P_i, \gamma_i / \varnothing , Z_i} \dots\xrightarrow{P_n, \gamma_i / L_n, Z_n} ((p_0', q_0'))$ and the second is $\rho': ((p_0, q_0), \nu_0) \xrightarrow{P_1', \gamma_1' / L_1, Z_1'} \dots \xrightarrow{P_i, \gamma_i' / \varnothing , Z_i'} \dots\xrightarrow{P_n, \gamma_i' / L_n, Z_n'} ((p_0', q_0'))$. We know that $Z_i = Z_i'$, with its content being $\{z_{X}, z_{X'}\}$ if $L \ne \varnothing$ and $\varnothing$ if $L = \varnothing$. This automaton has synchronous resets, because for every index of the runs, $Z_i = Z_i'$. Moreover, as $z_{X}$ is reset in the same indexes as $z_{X'}$, we can simplify the number of clocks to two by replacing every instance of the clock $z_{X'}$ by $z_{X}$. Finally, the resulting automaton has one clock $z_{X}$, where $z_{X}$ is reset in every marking transition.
	
	For the automaton $\cT_{\phi_1 \sq \psi}$, we use the construction in Proposition~\ref{pro:cel2cea}. We then analyze a run of $\cT_{\phi_1 \sq \psi}$. That run is of the form $(q_0, \nu_0) \xrightarrow{P_1, \gamma_1 / L_1, Z_1} \dots \xrightarrow{P_i, \gamma_i / L_i, Z_i} \dots \xrightarrow{P_j, \gamma_j / L_j, Z_j} \dots \xrightarrow{P_n, \gamma_n / L_n, Z_n} (q_n, \nu_n)$. We know that $Z_k = \{z_X\}$ if $L_k \ne \varnothing$ else $\varnothing$ for $k \le i$, $Z_{k} = \varnothing$ for all $i<k<j$ and finally $Z_{k} = \{z_X'\}$ if $L_k\ne\varnothing$ else $\varnothing$ for all $j \le k$. We also know that the clock constraints after the index $i$ don't use the clock $z_X$. This way, we can change all occurences of clock $z_X'$ by the clock $z_X$, thus having the same language. The resulting timed CEA resets the clock $z_X$ at every marking transition, so it is easy to see that it has synchronous resets.
	
	The proof is analog for the formulas $\phi_1 \ssq \phi_2$.
	
	For the automaton $\cT_{\phi_1 \xsq{I} \psi}$, we use the construction in Proposition~\ref{pro:cel2cea}. We then analyze a run of $\cT_{\phi_1 \xsq{I} \psi}$. That run is of the form $(q_0, \nu_0) \xrightarrow{P_1, \gamma_1 / L_1, Z_1} \dots \xrightarrow{P_i, \gamma_i / L_i, Z_i} \dots \xrightarrow{P_j, \gamma_j / L_j, Z_j} \dots \xrightarrow{P_n, \gamma_n / L_n, Z_n} (q_n, \nu_n)$. We know that $Z_k = \{z_X\}$ if $L_k \ne \varnothing$ else $\varnothing$ for $k < i$, $Z_i = \{z_X, z_X'\}$ $Z_{k} = \varnothing$ for all $i<k<j$ and finally $Z_{i} = \{z_X''\}$ if $L_i\ne\varnothing$ else $\varnothing$ for all $j \le i$. We also know that the clock constraints after the index $i$ don't use the clock $z_X$ and after the index $j$ don't use the clock $z_X'$. This way, we can change all occurences of clock $z_X'$ and $z_X''$ by the clock $z_X$, thus having the same language. The resulting timed CEA resets the clock $z_X$ at every marking transition, so it is easy to see that it has synchronous resets.

	The proof is analog for the formula $\phi_1 \xssq{I}\phi_2$.
	
	For the automaton $\cT_{\phi_1 \ks}$, we use the construction in Proposition~\ref{pro:cel2cea}. We then analyze a run of $\cT_{\phi_1 \ks}$. That run is of the form $(q_0, \nu_0) \xrightarrow{P_1, \gamma_1 / L_1, Z_1} \dots \xrightarrow{P_i, \gamma_i / L_i, Z_i} \dots \xrightarrow{P_j, \gamma_j / L_j, Z_j} \dots \xrightarrow{P_n, \gamma_n / L_n, Z_n} (q_n, \nu_n)$. We know that the construction does not add any new clocks, and that by construction the clocks are preserved. Then, the new automaton resets the clock $z_X$ in only every marking transition, so it has synchronous resets.
	
	The proof is analog for automaton $\cT_{\phi_1 \sks}$.
	
	For the automaton $\cT_{\phi_1 \xks{I}}$, we use the construction in Proposition~\ref{pro:cel2cea}. We know that the new added clock $z_X'$ is reset in the transitions from $p$ to $q_{new}$, and that it is a marking transition. Then, we also know that the clock $z_X$ is also reset in that transition, and that the following transitions before the clock check, i.e., the transition from $q_{new}$, do not reset either the clock $z_X$ or the clock $z_X'$. By construction, there are no other transitions before the clock check. That way, We know that in those transitions, the valuation of the clock $z_X$ is the same that the valuation of the clock $z_X'$, so we can replace the clock $z_X'$ by the clock $z_X$, leaving a single clock in the automaton that is reset in every marking transition, so the automaton has synchronous resets.
	
	The proof is analog for automaton $\cT_{\phi \timesks{I}}$.
\end{proof}

Now we will proceed to prove the theorem by proving a stronger proposition: a windowed CEL formula can be translated to a timed CEA with two clocks that always resets one in the first transition and the other in every marking transition.

\begin{proof}
	We will prove this result by induction. As a base case, we know that every simple timed CEL formula is equivalent to a timed CEA with synchronous resets and one clock that resets in every marking transition. If we add one clock that is never compared in a clock condition, but is reset in all starting transitions, we achieve an equivalent timed CEA that follows the property. For the base case of CEL formulas, we know that the equivalent automata doesn't have any clocks. We can add two clocks, one that is reset in every marking transition and the other that is reset in only the first transitions. It is trivial that adding these clocks doesn't change the language accepted by the automaton, because the automaton does not have any clock conditions. For the inductive case, assume that for windowed timed CEL formulas $\phi_1$ and $\phi_2$, there exist $\cT_{\phi_1}$ and $\cT_{\phi_2}$ such that $\semaux{\phi_1} = \semaux{\cT_{\phi_1}}$ and $\semaux{\phi_2} = \semaux{\cT_{\phi_2}}$ and the property is satisfied. %

	We will prove that there exist timed CEAs that satisfy the property for any formula
	\[
	\phi \in \{\phi_1 \kAS X, \phi_1 \kFILTER X[P], \phi_1 \kOR \phi_2, \phi_1 \kAND \phi_2,  \twin{\phi_1}{\sim c}\}
	\] 
	such that $\semaux{\cT_{\phi}} = \semaux{\phi}$. Remember that the initial state has no incoming transitions, as this will allow us to reset clock $z_N$ only in the first transition of the run.
	
	For the automaton $\cT_{\phi_1 \kAS X}$, we will follow the construction in the proof of Proposition~\ref{pro:cel2cea} to create an automaton such that $\semaux{\cT_{\phi_1 \kAS X}} = \semaux{\phi_1 \kAS X}$, and we will prove that $\cT_{\phi_1 \kAS X}$ has synchronous resets. Let $\rho_1: (q_0, \nu_0) \xrightarrow{P_1, \gamma_1 / L_1, Z_1} \dots \xrightarrow{P_n, \gamma_n / L_n, Z_n} (q_n, \nu_n)$ and $\rho_2: (q_0', \nu_0') \xrightarrow{P_1', \gamma_1' / L_1, Z_1'} \dots \xrightarrow{P_n', \gamma_n' / L_n, Z_n'}(q_n, \nu_n)$ be partial runs of $\cT_{\phi_1 \kAS X}$ over $\Stream$. We know that the transitions of $\cT_\phi$ are preserved, only changing the markings to include the new variable $X$. Then, for two transitions $(q_{j-1}, P_j, \gamma_j, L_j, Z_j, q_j)$ and $(q_{j-1}', P_j', \gamma_j', L_j, Z_j', q_j')$, we know that $Z_j = Z_j'$. Furthermore, $\cT_{\phi_1 \kAS X}$ as the number of clocks and their resets don't change, we know that the property is preserved.
	
	For the automaton $\cT_{\phi_1 \kFILTER X[P]}$, the proof is analogous as the one for $\phi \kAS X$, but changing the predicates instead of the variable markings.
	
	For the automaton $\cT_{\phi_1\kOR \phi_2}$, we know that there exist $\cT_{\phi_1}$ and $\cT_{\phi_2}$ such that for every run, $z_N$ and $z_{N'}$ are the clocks that reset at the start of the run and $z_{X}$, $z_{X'}$ are the clocks that reset in every marked transition. We also know that in the construction, there are no transitions from a state in $\cT_{\phi_1}$ to a state in $\cT_{\phi_2}$ and viceversa. In that way, we know that if a clock is reset in a run passing through the states of $\cT_{\phi_1}$, then it will not affect any run passing through the states of $\cT_{\phi_2}$. Using this fact, we can perform the construction in Proposition~\ref{pro:cel2cea} of $\phi_1'$ and $\phi_2'$, but assuming that $z_{X}=z_{X'}$ and $z_{N}=z_{N'}$, that is, the resulting automaton has only two clocks, the first clock $z_{N}$ that is reset in the first transition and the second clock $z_{X}$ that is reset in every marking transition.
	
	For the automaton $\cT_{\phi_1 \kAND \phi_2}$, we do the construction of Proposition~\ref{pro:cel2cea}. We know that the timed CEA $\cT_{\phi_1}$ and $\cT_{\phi_2}$ have two clocks each, that is, $z_{N}$, $z_{N'}$ (the clocks that are reset in the first transition) and $z_{X}$, $z_{X'}$ (the clocks that are reset in every marking transition). Then, let $\rho$ and $\rho'$ be two partial runs of $\cT_{\phi_1 \kAND \phi_2}$ is $(q_0, q_0')$ over the stream $\Stream$. The first run is $\rho: ((p_0, q_0), \nu_0) \xrightarrow{P_1, \gamma_1 / L_1, Z_1} \dots \xrightarrow{P_i, \gamma_i / \varnothing , Z_i} \dots\xrightarrow{P_n, \gamma_i / L_n, Z_n} ((p_0', q_0'))$ and the second is $\rho': ((p_0, q_0), \nu_0) \xrightarrow{P_1', \gamma_1' / L_1, Z_1'} \dots \xrightarrow{P_i, \gamma_i' / \varnothing , Z_i'} \dots\xrightarrow{P_n, \gamma_i' / L_n, Z_n'} ((p_0', q_0'))$. We know that $Z_1 = Z_1' = \{z_{N}, z_{X}, z_{N'}, z_{X'}\}$, and that $Z_i = Z_i'$, with its content being $\{z_{X}, z_{X'}\}$ if $L \ne \varnothing$ and $\varnothing$ if $L = \varnothing$. This automaton has synchronous resets, because for every index of the runs, $Z_i = Z_i'$. Moreover, as $z_{N}$ is reset in the same indexes as $z_{N'}$ and $z_{X}$ is reset in the same indexes as $z_{X'}$, we can simplify the number of clocks to two by replacing every instance of the clock $z_{X'}$ by $z_{X}$ and of $z_{N'}$ by $z_{N}$. Finally, the resulting automaton has two clocks $z_{N}$ and $z_{X}$, where $z_{N}$ is reset in the first transition and $z_{X}$ is reset in every marking transition.
	
	For the automaton $\cT_{\twin{\phi_1}{I}}$, we do the construction of Proposition~\ref{pro:cel2cea}. Let $\rho: ((p_0, q_0), \nu_0) \xrightarrow{P_1, \gamma_1 / L_1, Z_1} \dots \xrightarrow{P_i, \gamma_i / \varnothing , Z_i} \dots\xrightarrow{P_n, \gamma_i / L_n, Z_n} ((p_n, q_n))$ and $\rho': ((p_0', q_0'), \nu_0) \xrightarrow{P_1', \gamma_1' / L_1, Z_1'} \dots \xrightarrow{P_i, \gamma_i' / \varnothing , Z_i'} \dots\xrightarrow{P_n, \gamma_i' / L_n, Z_n'} ((p_n', q_n'))$ be runs of $\cT_{\twin{\phi_1}{I}}$ over $\Stream$. If $z_{N}$ and $z_{X}$ are the clocks of $\cT_{\phi_1}$, then know that $Z_1 = Z_1' = \{z_{N}, z_{N'}, z_{X}\}$ and $Z_i = Z_i'$ with the value $\{z_{X}\}$ if $L_i\ne\varnothing$ and $\varnothing$ if $L_i = \varnothing$. Then, the timed CEA has synchronous resets. Furthermore, the transitions resets $z_{N}$ iff it resets $z_{N'}$, then the timed CEA such that every instance of $z_{N'}$ is replaced by $z_{N}$ is equivalent. Finally, the resulting automaton has two clocks $z_{N}$ and $z_{X}$, where $z_{N}$ is reset in the first transition and $z_{X}$ is reset in every marking transition.

\end{proof}
 	
	\section{Proofs from Section~\ref{sec:evaluation}} \label{app:evaluation}

\subsection*{Proof of Theorem~\ref{theo:algorithm}}
The proof of Theorem~\ref{theo:algorithm} requires the definition of some data structures to be able to solve the problem efficiently. First, we will introduce the \emph{Clock-Aware Enumerable Compact Set} (CAECS), the data structure that will store the intermediate outputs of the algorithm, to then define the methods to update it efficiently, and finally use it in the algorithm at the end. Note that this proof is made only for the lower equal case of the monotonic single clock deterministic timed CEA. That is because the greater equal case is analog, and is easily constructed with minimal changes to the framework proposed.

\subsection{The data structure}

Our data structure is called a \emph{Clock-Aware Enumerable Compact Set} (CAECS), extending the structures presented in \cite{MunozR22} and later used in \cite{BucchiGQRV22}. Specifically, a CAECS is a structure $\ecs = (\Omega, V, \dleft, r, \lambda)$, where $V$ is a finite set of nodes, $\dleft: V \rightarrow V$ and $r: V\rightarrow V$ are the left and right partial functions, and $\lambda: V \rightarrow \Omega$ is a labeling function. We assume that $\ecs$ forms an directed acyclic graph.

A CAECS has six types of nodes: (1) bottom nodes, (2) extended nodes, (3) union nodes, (4) reset nodes, (5) clock-check nodes, and (6) empty nodes. The graphical representation of these nodes is depicted in Figure \ref{fig:node-types}. For any node $n$ with its corresponding type, we define $\lambda(n)$ as follows:
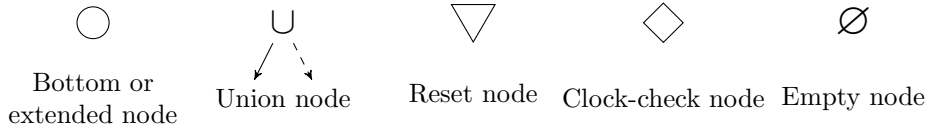
\begin{figure}[t]
	\centering
	\begin{tikzpicture}[->,
		>=stealth',
		initial text= {},
		initial distance= {2.5mm}
		]
		
		\begin{scope}[xshift=0cm]
			\node[extnode] (i) at (0, 0) {};
			\node[below of=i, text width=70pt,align=center] {Bottom or extended node};
		\end{scope}
		
		\begin{scope}[xshift=2.5cm]
			\node[unionnode] (i) at (0, 0) {};
			\node[anynode] (l) at (-0.5, -1) {};
			\node[anynode] (r) at (0.5, -1) {};
			\draw[dashed] (i) to (r);
			\draw[solid] (i) to (l);
			\node[below of=i] {Union node};
		\end{scope}
		
		\begin{scope}[xshift=5cm]
			\node[resetnode] (i) at (0, 0) {};
			\node[below of=i] {Reset node};
		\end{scope}
		
		\begin{scope}[xshift=7.5cm]
			\node[clockchecknode] (i) at (0, 0) {};
			\node[below of=i] {Clock-check node};
		\end{scope}
		
		\begin{scope}[xshift=10cm]
			\node[emptynode] (i) at (0, 0) {};
			\node[below of=i] {Empty node};
		\end{scope}
	\end{tikzpicture}
	\caption{CAECS node representations. In union nodes, the solid arrow represents the left child and the dashed arrow represents the right child.}
	\label{fig:node-types}
\end{figure}

\begin{itemize}
	\item If $n$ is a bottom node, then  $\lambda(n) = (i, t_i)$, where $i\in \bbN$ represents the index of the stream, and $t_i\in \bbR$ represents the absolute time when the open timed complex event starts. Also, $\dleft(n)$ and $\dright(n)$ are not defined. 
	\item If $n$ is an extended node, then  $\lambda(n) = (i, L)$, where $i$ represents the index of the stream, and $L$ represents the marked variables in that position for that nodes, $\dleft(n)\in V$ and $\dright(n)$ is not defined. 
	\item If $n$ is an union node, then  $\lambda(n) = \cup$, and $\dleft(n), \dright(n)\in V$. 
	\item If $n$ is a reset node, then  $\lambda(n) = t$, where $t$ represents the absolute time where the reset is performed. In this case, only $\dleft(n)$ is defined. 
	\item If $n$ is a clock-check node, then  $\lambda(n) = (t_0, c)$, where $t_0$ is the absolute time where the node was added, and $c$ is the maximum time allowed for the clock valuation. In this case, only $\dleft(n)$ is defined.
	\item If $n$ is an empty node, then  $\lambda(n) = \varnothing$.
\end{itemize}

For the semantics of the nodes of the CAECS, we will introduce a new representation of complex events. For a complex event $C = (i, j, \mu)$, instead of using the mapping $\mu: \xset \rightarrow 2^{[i..j]}$ between variables and indexes, we will use the representation of \emph{indexed complex event} $C = (i,j, \iota)$ where $\iota$ is the inverse  of $\mu$, namely, mapping indexes to set of variables. More specific, $\iota: [i..j] \rightarrow 2^\xset$ is a function such that $X \in \iota(k)$ iff $k \in \mu(X)$ for every $X \in \xset$ and $k \in [i..j]$. Furthermore, to optimize the representation, we assume that $\iota$ only maps $k \in [i..j]$ whenever $\iota(k) \neq \emptyset$. That said, those two definitions are equivalent and interchangeable, but the latter will be easier to maintain in the later presented algorithm, while having no increment on its time complexity. We will refer to the indexed complex events as complex events from now on. Moreover, an \emph{open timed complex event} is a pair $(i, \iota)$ where $i \in \bbN$ represents the index where the complex event starts and $\iota: \{i, i+1, \dots\}\rightarrow 2^\xset$ represents the mapping between indexes of the stream and variables marked in those indexes. The open timed complex event is called \textit{open} because it lacks an end index $j$: only when the timed CEA reaches a final state, the timed complex event is closed and assigned an end index. That said, the \emph{auxiliary semantics} $\semauxop{n}$ of a node $n \in V$ of the CAECS is a set of triples $(i, \iota, t)$ where $(i, \iota)$ is an open timed complex event and $t \in \bbQplus$ such that:
\begin{itemize}
	\item If $n$ is a bottom node with $\lambda(n) = (i,t_i)$, then $\semauxop{n} = \{(i, \varnothing, t_i)\}$, where $\varnothing$ is a function that maps every index to the empty set.
	\item If $n$ is an extended node with $\lambda(n) = (i',L)$, then $\semauxop{n} = \{(i, \iota', t)\mid \exists (i, \iota, t)\in \semauxop{\dleft(n)}.\, \iota' = \iota \cup \{(i', L)\}\}$ where we assume that $i' \notin \dom{\iota'}$ for every $(i, \iota', t) \in \sem{\ell(n)}$.
	\item If $n$ is a union node, then $\semauxop{n} = \semauxop{\dleft(n)}\cup \semauxop{\dright(n)}$. 
	\item If $n$ is a reset node with $\lambda(n) = t$, then $\semauxop{n} = \{(i, \iota, t) \mid (i, \iota, t') \in \semauxop{\dleft(n)}\}$. 
	\item If $n$ is a clock-check node with $\lambda(n) = (t_0, c)$, then $\semauxop{n} = \{(i, \iota, t) \in \semauxop{\dleft(n)}\mid t_0 - t \le c\}$.
	\item If $n$ is a empty node, then $\semauxop{n} = \varnothing$.
\end{itemize} 
In other words, the auxiliary semantics $\semauxop{n}$ of node $n$ contains all the timed open complex events that can be created from that node. 

Given a node $n$ of $\ecs$, the main goal of using an CAECS $\ecs$ is to enumerate the set of open timed complex event from $n$ defined as:
\[
\sem{n} = \{(i, \iota) \mid (i, \iota, t) \in \semauxop{n}\}
\]
Towards this goal, we need to impose several conditions on the CAECS. The first condition is that it must be \textit{time-ordered}, condition that will be defined after some preliminary definitions. For a node $n$, we define its \textit{maximum-reset}, denoted $\maxreset(n)$. It is defined as $\maxreset(n) = \max\{t\mid (i, \iota, t)\in \semauxop{n}\}$. The value of the \textit{maximum-reset} holds the greater absolute time where the clock was resetted in any of the possible runs passing through the node $n$. That said, a union node $n$ is \textit{time-ordered} if $\maxreset(\dleft(n)) \ge \maxreset(\dright(n))$. This simplifies the traversal of the CAECS, as we always check the left node, and only check the right node if the active clock conditions are satisfied on its \textit{maximum-reset}. Note that as the timed CEA is monotonic, this will allow us to filter the results by ignoring all of the nodes that do not follow the $\le$ constraint. For the greater-equal case, the construction is analog but ordered inversely.

The second condition is that the CAECS $\ecs$ must be \emph{$k$-bounded} for some $k\in \bbN$. Define the (left) \emph{output-depth} of a node $n$, written as $\odepth(n)$, as follows: if $n$ is a bottom or extended node, $\odepth(n) = 0$. Otherwise, $\odepth(n) = \odepth(\dleft(n)) + 1$. The output depth tell us how many non-output nodes (i.e., union, reset, or clock-check nodes) we need to traverse to the left before we find an output node (i.e., bottom or extended) that, therefore, produces part of the output. Then $\ecs$ is \emph{$k$-bounded} if $\odepth(n) \le k$ for every node $n$.

The third condition is that the CAECS $\ecs$ must be \textit{duplicate-free}, which holds if, for every union node $n$ in $\ecs$, it holds that $\{(i, \iota)\mid (i, \iota, t)\in \semauxop{\ell(n)}\} \cap \{(i, \iota) \mid (i, \iota, t)\in \semauxop{r(n)}\} = \emptyset$. This condition ensures that no union node will have the same output in two different branches.

\paragraph{Methods for managing the data structure}
To correctly manage the CAECS, we need to ensure that every operation that we perform on nodes of the CAECS will preserve a bounded $\odepth$ of the structure. For that matter, we need to check some preliminary definitions and properties.

In an instance of a CAECS structure, the reset and clock-check nodes behave in a particular way. More in detail, a sequence of reset and clock-check nodes can be simplified if they meet some particular conditions. This will be important when managing the data structure, because it allows to simplify that sequence by a new sequence of at most 2 nodes. To simplify the explanation and algorithms of the next parts, we make the following definition.

A \textit{clock-aware gadget} $g$ is a structure composed of zero to two nodes. The structure has an \textit{entry} node, named as $\entry(g)$, which corresponds to the first node inside the structure (the first node of the directed subgraph), and an \textit{exit} node, named as $\exit(g)$, which corresponds to the first node after the structure (the node pointed at by the last one of the directed subgraph). The special case is when $g$ has no nodes, in that case, $\entry(g) = \exit(g)$. We also define $\nodes(g)$ as the ordered list of nodes of which is composed the gadget, starting from the entry (if the structure has any). We will simplify the notation by saying that if $g$ is a clock-aware gadget, $\semauxop{g} = \semauxop{\entry(g)}$.

There are then five possible forms for a clock-aware gadget: the first is a void gadget (a gadget with no nodes), the second is a reset gadget (a gadget with only a reset node), the third is a clock-check gadget (a gadget with only a clock-check node), the fourth is an empty gadget (a gadget with only an empty node) and the fifth is a composed gadget (a gadget with a reset node as its entry, pointing to a clock-check node). The gadgets are listed in Figure~\ref{fig:gadget-types}.
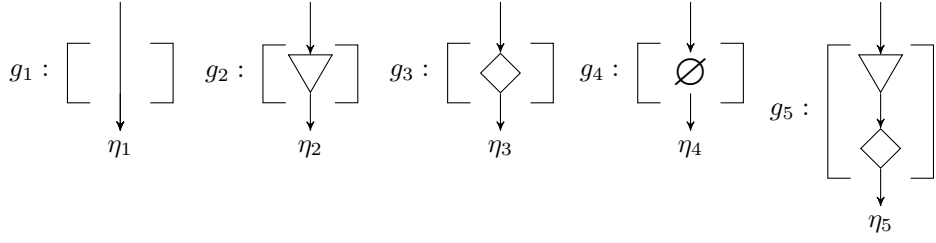
\begin{figure}[t]
	\centering
	\begin{tikzpicture}[->,
		>=stealth',
		initial text= {},
		initial distance= {2.5mm}
		]
		
		\begin{scope}[xshift=0cm]
			\node[rootnode] (i) at (0, 0) {};
			\node[minimum size=15pt] (g) at (0,-1) {};
			\node (e) at (0, -2) {$\eta_1$};
			
			\draw[solid] (i) to (e);
			\draw[solid] (g) to (e);

			\node[fit=(g)] (xyzbox) {};
			
			\draw [-, black] (xyzbox.north west) [square left brace] to node[left] {$g_1\,\colon$} (xyzbox.south west);
			\draw [-, black] (xyzbox.north east) [square right brace] to (xyzbox.south east) ;
		\end{scope}
		
		\begin{scope}[xshift=2.5cm]
			\node[rootnode] (i) at (0, 0) {};
			\node[resetnode] (g) at (0,-1) {};
			\node (e) at (0, -2) {$\eta_2$};
			
			\draw[solid] (i) to (g);
			\draw[solid] (g) to (e);

			\node[fit=(g)] (xyzbox) {};
			
			\draw [-, black] (xyzbox.north west) [square left brace] to node[left] {$g_2\,\colon$} (xyzbox.south west);
			\draw [-, black] (xyzbox.north east) [square right brace] to (xyzbox.south east) ;
		\end{scope}
		
		\begin{scope}[xshift=5cm]
			\node[rootnode] (i) at (0, 0) {};
			\node[clockchecknode] (g) at (0,-1) {\,};
			\node (e) at (0, -2) {$\eta_3$};
			
			\draw[solid] (i) to (g);
			\draw[solid] (g) to (e);

			\node[fit=(g)] (xyzbox) {};
			
			\draw [-, black] (xyzbox.north west) [square left brace] to node[left] {$g_3\,\colon$} (xyzbox.south west);
			\draw [-, black] (xyzbox.north east) [square right brace] to (xyzbox.south east) ;
		\end{scope}
		
		\begin{scope}[xshift=7.5cm]			%
			\node[rootnode] (i) at (0, 0) {};
			\node[emptynode] (g) at (0,-1) {};
			\node (e) at (0, -2) {$\eta_4$};
			
			\draw[solid] (i) to (g);
			\draw[solid] (g) to (e);

			\node[fit=(g)] (xyzbox) {};
			
			\draw [-, black] (xyzbox.north west) [square left brace] to node[left] {$g_4\,\colon$} (xyzbox.south west);
			\draw [-, black] (xyzbox.north east) [square right brace] to (xyzbox.south east) ;
		\end{scope}
		
		\begin{scope}[xshift=10cm]			%
			\node[rootnode] (i) at (0, 0) {};
			\node[resetnode] (g1) at (0,-1) {};
			\node[clockchecknode] (g2) at (0,-2) {\,};
			\node (e) at (0, -3) {$\eta_5$};
			
			\draw[solid] (i) to (g1);
			\draw[solid] (g1) to (g2);
			\draw[solid] (g2) to (e);

			\node[fit=(g1) (g2)] (xyzbox) {};
			
			\draw [-, black] (xyzbox.north west) [square left brace] to node[left] {$g_5\,\colon$} (xyzbox.south west);
			\draw [-, black] (xyzbox.north east) [square right brace] to (xyzbox.south east) ;
		\end{scope}
	\end{tikzpicture}
	\caption{Clock-aware gadget types. The gadget is the  $g_1$ corresponds to a void gadget, $g_2$ corresponds to a reset gadget, $g_3$ corresponds to a clock-check gadget, $g_4$ corresponds to an empty gadget, $g_5$ corresponds to a composed gadget. In each gadget version, the entry node is the one pointed by the top-down arrow, and the exit node is the $\eta_i$ below.}
	\label{fig:gadget-types}
\end{figure}

We also will present the following definitions. If $u_1$ and $u_2$ are two non-union nodes such that $\uleft(u_1) = u_2$, then we will say that $u_1$ and $u_2$ are \textit{in series}, and we will define $u_1 \circ u_2$ as a new clock-aware gadget such that $\semauxop{u_1\circ u_2} = \semauxop{u_1}$. Likewise, if $g_1$ and $g_2$ are two clock-aware gadgets such that $\exit(g_1) = \entry(g_2)$, we will say that $g_1$ and $g_2$ are \textit{in series}, and we will define $g_1 \circ g_2$ as a new clock-aware gadget such that $\semauxop{g_1\circ g_2} = \semauxop{g_1}$.

\begin{lemma} \label{lemma:clock-aware-nodes}
	Let $u_1$ and $u_2$ be reset or clock-check nodes \textit{in series}, and $u$ a node such that $\dleft(u_2) = u$. We can construct a new clock-aware gadget $g_{1,2}$ having $\exit(g_{1,2}) = u$ such that $\semauxop{g_{1,2}} = \semauxop{u_1 \circ u_2}$.
\end{lemma}
\begin{proof}
	In this proof, we will assume that $u_i'$ is a copy of the node $u_i$.
	\begin{itemize}
		\item Let $u_1$, $u_2$ be reset nodes having $\lambda(u_1) = t_1$ and $\lambda(u_2) = t_2$. Then, by definition, $\semauxop{u_2} = \{(i, \iota, t_2) \,\mid\, (i, \iota, t) \in \semauxop{u}\}$, and $\semauxop{u_1} = \{(i, \iota, t_1) \,\mid\, (i, \iota, t) \in \semauxop{u_2}\} =\{(i, \iota, t_1) \,\mid\, (i, \iota, t) \in \semauxop{u}\}$. As such, we can construct $g_{1,2}$ as being a reset gadget with $\nodes(g_{1,2}) = [u_1']$.
		\item Let $u_1$ and $u_2$ be clock-check nodes having $\lambda(u_1) = (t_1, W_1)$ and $\lambda(u_2) = (t_2, W_2)$. Then, by definition, $\semauxop{u_2} = \{(i, \iota, t)\in \semauxop{u} \,\mid\, t_2 - t \le W_2 \}$, and $\semauxop{u_1} = \{(i, \iota, t)\in \semauxop{u_2} \,\mid\, t_1 - t \le W_1\} = \{(i, \iota, t)\in \semauxop{u} \,\mid\, t_1 - W_1 \le t \,\land\, t_2 - W_2 \le t\}$. We will assume w.l.o.g. that $t_1 \ge t_2$, and will distinguish three cases: if $t_1-W_1 \le t_2-W_2$ (one time window is contained within the other), then $g_{1,2}$ is a new clock-check gadget, having $\nodes(g_{1,2}) = [u_2']$. If $t_1-W_1 \in (t_2-W_2, t_2]$ (the time windows have an intersection), then $g_{1,2}$ is a new clock-check gadget, having $\nodes(g_{1,2}) = [v]$ and $\lambda(v) = (t_2, W_1 - (t_1-t_2))$. In every other case (the time windows don't overlap), $g_{1,2}$ is the empty gadget.
		\item Let $u_1$ be a clock-check node having $\lambda(u_1) = (t_1, c)$ and $u_2$ be a reset node having $\lambda(u_2) = t_2$. Then $\semauxop{u_2} =  \{(i, \iota, t_2) \,\mid\, (i, \iota, t) \in \semauxop{u}\}$, and $\semauxop{u_1} = \{(i, \iota, t)\in \semauxop{u_2} \,\mid\, t_1 - t \le c\} = \{(i, \iota, t_2) \,\mid\, (i, \iota, t) \in \semauxop{u} \land t_1 - t_2 \le c \}$. If $t_1-t_2 \le c$, $g_{1,2}$ is a reset gadget having $\nodes(u) = [u_2']$, and in every other case, $g_{1,2}$ is an empty gadget.
		\item Let $u_1$ be a reset node and $u_2$ be a clock-check node. Then $g_{1,2}$ is a composed node such that $\nodes(g_{1,2}) = [u_1', u_2']$.
	\end{itemize}
	
\end{proof}

We have proven that any two reset or clock-check nodes \textit{in series} can be replaced by a single new clock-aware gadget having its same semantics. Now we will prove a stronger property: any two clock-aware gadgets \textit{in series} can be replaced by a single new clock-aware gadget having its same semantics.

\begin{lemma}
	If $g_1$ and $g_2$ are clock-aware gadgets \textit{in series} and $u$ is a node such that $\exit(g_2) = u$, then we can construct a new clock-aware gadget $g_{1,2}$ having $\exit(g_{1,2}) = u$ such that $\semauxop{g_{1,2}} = \semauxop{g_1 \circ g_2}$.
\end{lemma}

\begin{proof}
	Let $g_1$ and $g_2$ are clock-aware gadgets \textit{in series} and $u$ a node such that $\exit(g_2) = u$. We will prove that we can construct $g_{1,2}$ for any type of clock-aware gadgets.
	\begin{itemize}
		\item Let $g_1$ be an void gadget w.l.o.g. Then $g_{1,2} = g_2$.
		\item Let $g_1$ or $g_2$ be empty gadgets. Then $g_{1,2}$ is a empty gadget too.
		\item Let $g_1$ and $g_2$ be reset or clock-check gadgets. This case follows directly after Lemma \ref{lemma:clock-aware-nodes}.
		\item Let $g_1$ be a reset or clock-check gadget and $g_2$ be a composed gadget w.l.o.g. This case follows directly after iteratively applying Lemma \ref{lemma:clock-aware-nodes} to each pair of nodes \textit{in series} until the resulting clock-aware gadget is single and minimal.
		\item Let $g_1$ and $g_2$ be composed gadgets. This case is solved the same way as the previous one.
	\end{itemize}
\end{proof}

We need to ensure some properties on each node so that the $\odepth$ of every node of the CAECS is bounded by $k$ when we apply the operations. For that matter, we define a \textit{safe node}, and will enumerate the conditions that make a node $u$ a safe node. If $u$ is a bottom, extended or empty node, it is always a \textit{safe node}. If $u$ is a union node, then it is a safe node if $\odepth(u) \le 9$. If $u$ is a part of a clock-aware gadget $g$, then it is a safe node if $\exit(g)$ is either a safe union node, a bottom node or an extended node, that means, there cannot be any clock-aware gadgets \textit{in series} in the CAECS. Additionally, if $u$ is a clock-check node with $\lambda(u) = (t_0, c)$, then it is required that $t_0 - \maxreset(\uleft(u)) \le c$. Finally, if a node has a empty node as a child, then it is not safe, even if it meets the other conditions required to be safe. If every node in the CAECS is a \textit{safe node}, then the $\odepth$ of every node of the CAECS is bounded by $k=11$. Note that $k \leq 11$ which is bounded.

To allow for a \textit{output-linear} enumeration algorithm, we must provide operations to update our CAECS while ensuring that every node remains a safe node after the update. The operations are the following:
\[
\begin{array}{c}
	b \leftarrow \dbottom(i, t_i) \ \  \ \ \ \ \onode \leftarrow \dextend(n, j, L) \ \ \ \ \ \  \unode \leftarrow \dunion(n_1, n_2)\\
	b \leftarrow \dreset(n, t) \ \  \ \ \ \ \onode \leftarrow \dclockcheck(n, t_0, c)
\end{array}
\]
We want to ensure that the CAECS continues to have only safe nodes after running any operation. Furthermore, we will guarantee that all the nodes of the CAECS after using these methods not only are safe, but also they follow one of two node structures, as depicted in Figure~\ref{fig:safe-nodes}. The first structure is a clock-aware gadget with an extended node as its exit node. The second structure is a clock-aware gadget with an union node as its exit node, that has the first structure as a left child and a node with $\odepth$ of at most $6$ as a right child. All the nodes in a safe structure are required to be safe.

Aiming towards that goal, the operations are implemented as follows. They are based on the construction of the ECS in \cite{BucchiGQRV22}. The first method, $\dbottom(i, t_i)$, returns a new node $v$ such that $\semauxop{v} = \{(i, \varnothing, t_i)\}$, by adding a new bottom node to the CAECS. The second method, $\dextend(n, j, L)$, returns a new node $v$, such that $\semauxop{v} = \{(i, \iota', t)\mid \exists (i, \iota, t)\in \semauxop{n}.\, \iota' = \iota \cup \{(j, L)\}\}$, by creating the extended node $v$ and linking it to node $n$. The third method, $\dunion(n_1, n_2)$ requires that $\maxreset(n_1) = \maxreset(n_2)$, and returns a new node $v$ such that $\semauxop{v} = \semauxop{n_1}\cup \semauxop{n_2}$. This method is depicted in Figure~\ref{fig:union-details}. The fourth method, $\dreset(n, t)$ returns a new node $v$ such that $\semauxop{v} = \{(i, \iota, t) \mid (i, \iota, t') \in \semauxop{n}\}$, by getting the largest clock-aware gadget from node $n$, named $g_2$ and creating a new reset gadget $g$ such that $\semauxop{g} = \semauxop{g_1 \circ g_2}$. The fifth method, $\dclockcheck(t_0, c)$ returns a new node $v$ such that $\semauxop{v} = \{(i, \iota, t) \in \semauxop{n}\mid t_0 - t \le c\}$, by getting the largest clock-aware gadget from node $n$, named $g_2$ and creating a new clock-check gadget $g$ such that $\semauxop{g} = \semauxop{g_1 \circ g_2}$. It is very important to note that if $t_0-\maxreset(n) > c$, this method returns an empty node, because this ensures that the enumeration of any node in an union-list (introduced later) will have at least one output. We detail the methods in Algorithm \ref{alg:update-caecs}. All of these methods require that the node inputs are not empty nodes. 

\begin{lemma}
	If the node inputs to the methods $\dbottom(i, t_i)$, $\dextend(n, j, L)$, $\dunion(n_1, n_2)$, $\dreset(n, t)$, $\dclockcheck(n, t_0,c)$ are roots to a safe structure, the nodes returned by the methods are also a root to a safe structure.
\end{lemma}

\begin{proof}
	We will go through the different cases and prove each one.
	\begin{itemize}
		\item If the method is $\dbottom(i, t_i)$, then the outputted node is a bottom node with no children. This output is the root of a type 1 safe structure.
		\item If the method is $\dextend(n, j)$, then if $n$ is the root of a safe structure, all of its nodes are safe. When adding a new extended node $v$ pointing to $n$, $v$ is a safe node by definition and so are all of the other nodes by assumption. Also, the node $v$ is the root of a type 1 safe structure.
		\item If the method is $\dunion(n_1, n_2)$, $n_1$ and $n_2$ are required to have the same max-reset. then there are 3 cases: if $n_1$ and $n_2$ are type 1 structures (the first case of Algorithm~\ref{alg:update-caecs}), then the result is the structure \textbf{(a)} of Figure~\ref{fig:safe-nodes}, and the result is the root of a type 2 safe structure. If $n_1$ is a root of a type 1 safe structure and $n_2$ is the root of a type 2 safe structure w.l.o.g. (the third case of Algorithm~\ref{alg:update-caecs}), then the result is the structure \textbf{(b)} of Figure~\ref{fig:safe-nodes}, and the result is the root of a type 2 safe structure. If $n_1$ and $n_2$ are roots of type 2 safe structures (the second case of Algorithm~\ref{alg:update-caecs}), then w.l.o.g. the result is the structure \textbf{(c)} of Figure~\ref{fig:safe-nodes}, and the result is the root of a type 2 safe structure. The only change is when $\maxreset(n_1) < \maxreset(n_2)$, and in that case the left and right childs of the third union are inverted, but the guarantees are the same.
		\item If the method is $\dreset(n, t)$ or $\dclockcheck(n, t_0, c)$, then whether $n$ is the root of a type 1 safe structure or a  root of a type 2 safe structure, then the reset or clock-check node is composed with the gadget at the root of the input safe structure and the result is the root of a safe structure of the same type as $n$. Additionally, if the method is $\dclockcheck(n, t_0, c)$ and $t_0 - \maxreset(n) > c$, then an empty node is returned, thus enforcing the condition that every safe clock-check node must have at least one output.
	\end{itemize}
\end{proof}

\begin{figure}[t]
	\centering
	\begin{tikzpicture}[->,
		>=stealth',
		initial text= {},
		initial distance= {2.5mm}]		
		
		\begin{scope}[xshift=0cm]
			\node (g) at (0, 0) {\bracketnode{g}};
			\node[extnode, accepting below] (e) at (0,-1) {};
			
			\draw (g) to (e);
			\node[left of=g] {$s_1\,\colon$};
			
		\end{scope}
		
		\begin{scope}[xshift=2.5cm]
			\node (g1) at (0, 0) {\bracketnode{g}};
			\node[unionnode] (u) at (0, -1) {};
			\node (g2) at (-0.5,-2) {\bracketnode{g}};
			\node[extnode, accepting below] (e1) at (-0.5,-3) {};
			\node[anynode] (e2) at (0.5,-2) {$\eta_1$};

			\draw[solid] (g1) to (u);
			\draw[solid] (u) to (g2);
			\draw[solid] (g2) to (e1);
			\draw[dashed] (u) to (e2);
			\node[left of=g1] {$s_2\,\colon$};
		\end{scope}
	\end{tikzpicture}
	\caption{Node structures returned by any CAECS update method. Each \bracketnode{g} figure is a clock-aware gadget. The node returned is the root of the tree. $s_1$ corresponds to the first node structure, and $s_2$ corresponds to the second safe node structure. The node $\eta_1$ is a node with $\odepth \le 6$.
	}
	\label{fig:safe-nodes}
\end{figure}
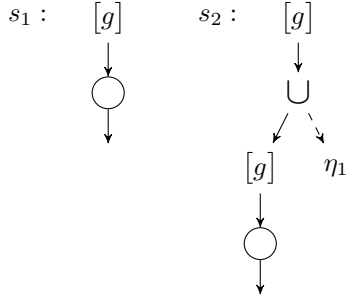

\begin{figure}[t]
	\centering
	\begin{tikzpicture}[->,
		>=stealth',
		initial text= {},
		initial distance= {2.5mm}]		
		
		\begin{scope}[xshift=0cm]
			\node[unionnode, label=west:{$r\,\colon$}] (u) at (0, -2.5) {};
			\node[label=west:{$n_1\,\colon$},inner xsep=0] (g1) at (-0.75,-3.5) {\bracketnode{g_1}};
			\node[label=west:{$n_2\,\colon$},inner xsep=0] (g2) at (0.75,-3.5) {\bracketnode{g_2}};
			\node[extnode, accepting below] (e1) at (-0.75,-4.5) {};
			\node[extnode, accepting below] (e2) at (0.75,-4.5) {};
			
			\draw[solid] (u) to (g1);
			\draw[dashed] (u) to (g2);
			\draw[solid] (g1) to (e1);
			\draw[solid] (g2) to (e2);
			
			\node at (0, -6) {\textbf{(a)}};
		\end{scope}
		
		\begin{scope}[xshift=3cm]
			\node[unionnode, label=west:{$r\,\colon$}] (u1) at (0, -1) {};
			
			\node[label=west:{$n_1\,\colon$},inner xsep=0] (g1) at (-1,-2) {\bracketnode{g_1}};
			\node[extnode, accepting below] (e1) at (-1,-3) {};
			
			\node[unionnode] (u2) at (1.25, -1.25) {};
			\node[label=west:{$n_2\,\colon$},inner xsep=0] (g2) at (2,-2.125) {\textcolor{gray}{\bracketnode{g_2}}};
			\node (g23) at (0,-2.75) {\bracketnode{g_{2,3}}};
			
			\node[grayunionnode] (u3) at (2, -3) {};
			\node (g22) at (3, -2) {\bracketnode{g_2'}};
			
			\node (n) at (3,-3.5) {$\eta$};
			\node (g3) at (1,-3.5) {\textcolor{gray}{\bracketnode{g_3}}};
			\node[extnode, accepting below] (e2) at (1, -4.5) {};
			
			\draw[solid] (u1) to (g1);
			\draw[solid] (g1) to (e1);
			\draw[dashed] (u1) to (u2);
			\draw[solid] (u2) to (g23);
			\draw[dashed] (u2) to (g22);
			\draw[solid, gray] (g2) to (u3);
			\draw[dashed] (u3) to (n);
			\draw[solid] (g22) to (n);
			\draw[solid] (g23) to (e2);
			\draw[solid, gray] (u3) to (g3);
			\draw[solid, gray] (g3) to (e2);
			
			\node at (1, -6) {\textbf{(b)}};
			
		\end{scope}

		\begin{scope}[xshift=9cm]
			\node[unionnode, label=west:{$r\,\colon$}] (u1) at (0, 0) {};
			\node[unionnode] (u2) at (1, -0.5) {};
			\node[unionnode] (u3) at (2, -1) {};
			\node (g12) at (-2,-2.5) {\bracketnode{g_{1,2}}};
			\node[extnode, accepting below] (e1) at (-1,-4) {};

			\node [label=west:{$n_1\,\colon$},inner xsep=0] (g1) at (-1,-0.5) {\textcolor{gray}{\bracketnode{g_1}}};
			\node[grayunionnode] (u4) at (-1, -2) {};
			\node (g2) at (-1,-3) {\textcolor{gray}{\bracketnode{g_2}}};
			\node (g11) at (0,-2) {\bracketnode{g_1'}};
			\node (n1) at (0, -3) {$\eta_1$};
			
			\node (g34) at (1,-3) {\bracketnode{g_{3,4}}};
			\node (g33) at (3,-1.75) {\bracketnode{g_3'}};
			
			\node (n2) at (3,-3.625) {$\eta_2$};
			
			\node [label=west:{$n_2\,\colon$},inner xsep=0] (g3) at (2.25,-1.75) {\textcolor{gray}{\bracketnode{g_3}}};
			\node[grayunionnode] (u5) at (2.25, -2.75) {};
			\node (g4) at (1.75,-3.625) {\textcolor{gray}{\bracketnode{g_4}}};
			\node[extnode, accepting below] (e2) at (1.75,-4.5) {};

			\draw[solid] (u1) to (g12);
			\draw[solid] (g12) to (e1);
			\draw[dashed] (u1) to (u2);
			\draw[dashed] (u2) to (u3);
			\draw[solid, gray] (g1) to (u4);
			\draw[solid, gray] (u4) to (g2);
			\draw[solid, gray] (g2) to (e1);
			\draw[dashed, gray] (u4) to (n1);
			\draw[solid] (g11) to (n1);
			\draw[solid] (u3) to (g11);
			\draw[solid] (u2) to (g34);
			\draw[dashed] (u3) to (g33);
			\draw[solid] (g33) to (n2);
			\draw[solid, gray] (g3) to (u5);
			\draw[dashed, gray] (u5) to (n2);
			\draw[solid, gray] (u5) to (g4);
			\draw[solid, gray] (g4) to (e2);
			\draw[solid] (g34) to (e2);
			
			\node at (0.25, -6) {\textbf{(c)}};
			
		\end{scope}
	\end{tikzpicture}
	\caption{Result of the union operation over nodes $n_1$ and $n_2$. The method returns the node $r$ in each case. Some cases are omitted when $\maxreset(\uright(\exit(n_1))) < \maxreset(\uright(\exit(n_2)))$, but that case is very similar to the ones depicted. Here, $\eta$, $\eta_1$ and $\eta_2$ are safe nodes with $\odepth \le 6$. Also, $g_{i,j}$ is a gadget resulting of the composing of gadgets $g_i$ and $g_j$, and $g_i'$ is a copy of the gadget $g_i$. You can check that the result of the union operation only contains safe nodes and that it returns a safe structure.}
	\label{fig:union-details}
\end{figure}
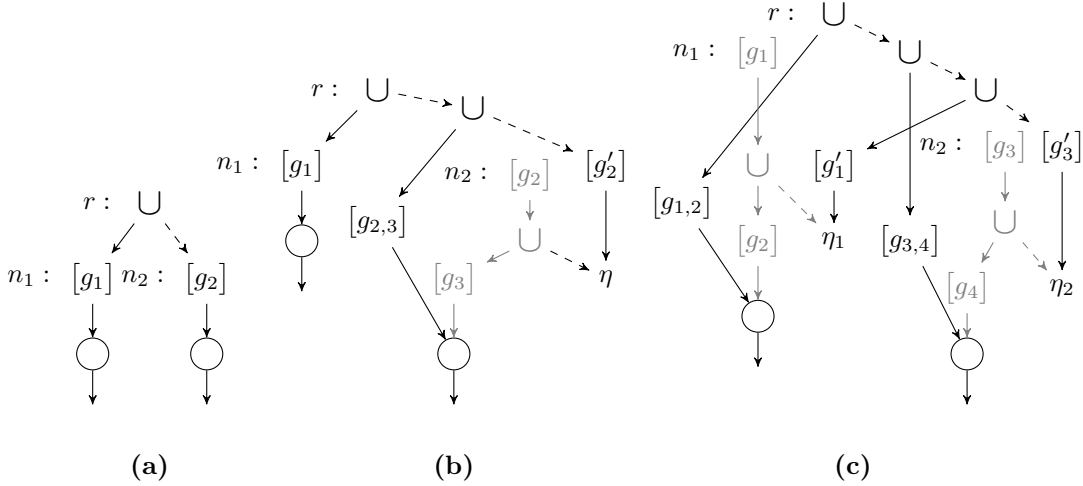

\algdef{SE}[DOWHILE]{Do}{doWhile}{\algorithmicdo}[1]{\algorithmicwhile\ #1}%
\algnewcommand{\LineComment}[1]{\State \(\triangleright\) #1}

\begin{algorithm*}[t]
	\caption{Methods for updating the CAECS. All of these methods return safe nodes. In this algorithm, $\getgadget(n)$ returns the largest valid gadget starting from the node $n$, and $\mergegadgets(g_1, g_2)$ returns a new gadget corresponding to $g_1\circ g_2$.}\label{alg:update-caecs}
	\smallskip
	\begin{varwidth}[t]{0.80\textwidth}
		\begin{algorithmic}[1]

			\LineComment{Creates a new bottom node}
			\Procedure{{\dbottom}}{$i$, $t_i$}%
			\State $b \gets \createbottom(i, t_i)$
			\State \Return $b$
			\EndProcedure	
			
			\State
			
			\LineComment{Returns the union of two nodes}
			\LineComment{Requires that $\max(n_1) = \max(n_2)$}
			\Procedure{{\dunion}}{$n_1$, $n_2$}
			\State $g_1 \gets \getgadget(n_1)$
			\State $g_2 \gets \getgadget(n_2)$
			\If {$\left(\begin{array}{l}
					\disbottom(\exit(g_1))\lor\\
					\disextend(\exit(g_1))\\
					\land \disbottom(\exit(g_2))\lor\\ \disextend(\exit(g_2))
				\end{array}\right)$}
			\State $u \gets \createunion()$
			\State $\dleft(u) \gets \entry(g_1)$
			\State $\dright(u) \gets \entry(g_2)$
			\State \Return $u$
			\ElsIf {$\left(\begin{array}{c}
					\disunion(\exit(g_1)) \\
					\land \disunion(\exit(g_2))
				\end{array}\right)$}
			\State $u_1, u_2, u_3 \gets \createunions()$
			\State $g_3 \gets \getgadget(\dleft(\exit(g_1)))$
			\State $g_4 \gets \getgadget(\dleft(\exit(g_2)))$
			\State $g_5 \gets \getgadget(\dright(\exit(g_1)))$
			\State $g_6 \gets \getgadget(\dright(\exit(g_2)))$
			\State $g_{1+3} \gets \mergegadgets(g_1, g_3)$
			\State $g_{2+4} \gets \mergegadgets(g_2, g_4)$
			\State $g_{1+5} \gets \mergegadgets(g_1, g_5)$
			\State $g_{2+6} \gets \mergegadgets(g_2, g_6)$
			\If{$\left(\begin{array}{c}
					\maxreset(\entry(g_{5}))\\
					 < \maxreset(\entry(g_{6}))
				\end{array}\right)$} \label{line:maxswap}
			\State $\swap(g_5, g_6)$
			\EndIf
			\State $\dright(u_1) \gets u_2$
			\State $\dright(u_2) \gets u_3$
			\State $\exit(g_{1+3}) \gets \exit(g_3)$
			\State $\exit(g_{2+4}) \gets \exit(g_4)$
			\State $\exit(g_{1+5}) \gets \exit(g_5)$
			\State $\exit(g_{2+6})\gets \exit(g_6)$
			\State $\dleft(u_1) \gets \entry(g_{1+3})$
			\State $\dleft(u_2) \gets \entry(g_{2+4})$
			\State $\dleft(u_3) \gets \entry(g_{1+5})$
			\State $\dright(u_3) \gets \entry(g_{2+6})$
			
			\State \Return $u_1$
			
			\algstore{myalg}

		\end{algorithmic}
	\end{varwidth}
	\begin{varwidth}[t]{0.80\textwidth}
		\begin{algorithmic}[1]
			
			\algrestore{myalg}
			\Else
			\If{$\disunion(\exit(g_1))$}
			\State $\swap(g_1, g_2)$
			\EndIf
			\State $u_1,u_2 \gets \createunions()$
			
			\State $g_3 \gets \getgadget(\dleft(\exit(g_2)))$
			\State $g_4 \gets \getgadget(\dright(\exit(g_2)))$
			
			\State $g_{2+3}\gets \mergegadgets(g_2, g_3)$
			\State $g_{2+4}\gets \mergegadgets(g_2, g_4)$
			\State $\dleft(u_1) \gets u_2$
			
			\State $\exit(g_{2+3}) \gets \exit(g_3)$
			\State $\exit(g_{2+4}) \gets \exit(g_4)$
			
			\State $\dleft(u_1) = \entry(g_1)$
			\State $\dleft(u_2) = \entry(g_{2+3})$
			\State $\dright(u_2) = \entry(g_{2+4})$
			\State \Return $u_1$
			
			\EndIf
			\EndProcedure
			
			\State

			\LineComment{Creates a new extended node}
			\Procedure{{\dextend}}{$n$, $j$, $L$}
			\State $e \gets \createextend(j, L)$
			\State $\dleft(e) \gets n$
			\State \Return $e$
			\EndProcedure
			
			\State
			
			\LineComment{Adds a new reset node}
			\Procedure{{\dreset}}{$n$, $t$}
			\State $g_1 \gets \createresetgadget(t)$
			\State $g_2 \gets \getgadget(n)$
			\State $g_{1+2}\gets \mergegadgets(g_1, g_2)$
			\State $\exit(g_{1+2}) \gets \exit(g_2)$
			\State \Return $\entry(g_{1+2})$
			\EndProcedure
			
			\State
			
			\LineComment{Adds a new clock-check node}
			\Procedure{{\dclockcheck}}{$n$, $t_0$, $c$}
			\If{$t_0 - \maxreset(n) > c$} \label{line:maxcheck}
			\State $n' \gets \createempty()$
			\State $\uleft(n') \gets n$
			\State \Return $n'$
			\EndIf
			\State $g_1 \gets \createclockcheckgadget(t_0, c)$
			\State $g_2 \gets \getgadget(n)$
			\State $g_{1+2}\gets \mergegadgets(g_1, g_2)$
			\State $\exit(g_{1+2}) \gets \exit(g_2)$
			\State \Return $\entry(g_{1+2})$
			\EndProcedure

		\end{algorithmic}
	\end{varwidth}
	\smallskip
\end{algorithm*}

\paragraph{Union-lists and their operations}
Now we define the second data structure for the streaming evaluation algorithm. 
We introduce the notion of \textit{union-list}, a structure to store the pointers from nodes of the CAECS to states of the automaton while reading the stream. The \textit{union-list} is a list $u_0\dots u_n$ of nodes of the CAECS. It is required that the first element of the union-list is the root of a type 1 structure, and that all the nodes comply with $\maxreset(u_0) \ge \maxreset(u_j)$ for all $j < n$ and $\maxreset(u_i)>\maxreset(u_{i+1})$ for all $1 \le i \le n-1$, meaning that the \textit{union-list} is ordered by strictly decreasing max-reset. Note that as the input automaton is monotonic, the order of the union list will allow us to easily discard nodes that will not satisfy a time constraint. Only the lower equal case is demonstrated as the greater equal case is analog, but ordered inversely.

The \textit{union-list} has 5 main methods to perform its updates: $\newul(u_0)$, $\insertul(ul, u)$, $\mergeul(ul)$, $\ulclockcheck(\ulist, t_0, c)$ and $\ulreset(\ulist, t)$. The first method, $\newul(u_0)$ creates a new \textit{union-list} with the node $u_0$. The node $u_0$ is required to be a type 1 structure to comply with the previously stated requirement. The second method, $\insertul(ul, u)$ inserts node $u$ in the \textit{union-list} so that the max-resets continue to be ordered. If there already is a node $u'\in ul$ such that $\maxreset(u') = j$, then the node $u'$ is replaced by the node returned by $\dunion(u, u')$. The third method, $\mergeul(ul)$, returns a new node $u''$ such that $\semauxop{u''} = \semauxop{u_0}\cup \dots \cup \semauxop{u_n}$, where $u_0\dots u_n$ are the nodes of the \textit{union-list}. More in detail, if the \textit{union-list} has only one element, it returns that element, and if it has $n>1$ elements it creates $n-1$ union nodes and places the \textit{union-list} nodes so that the time-ordered property is followed (as shown in Figure~\ref{fig:merge-ul}). The fourth method, $\ulclockcheck(\ulist, t_0, c)$, applies $\dclockcheck(n, t_0, c)$ to every node $n$ of the union-list, and returns a new union-list with the nodes in the same order, but filtering out the nodes that become empty nodes. If every node of the union list is an empty node after performing $\ulclockcheck$, then the union-list is deleted. It is important to note that as the union-list is time-ordered, then only the nodes to the right can be left as empty. It is also easy to note that the resulting union-list is also time-ordered, and the first node is the root of a type 1 structure.
The fifth method, $\ulreset(\ulist, t)$ applies $\dreset(n, t)$ to every node $n$ of the union-list. Then it creates a new union-list with the following procedure: first, it instantiates the union-list with the first resulting node as its first node (as it is guaranteed to have $\odepth \le 2$). Afterwards, it calls $\dunion$ repeatedly over the following nodes (as they have the same $\maxreset$), and leaves the resulting node as the second node of the new union-list. Finally, it returns the new union-list. That union-list is time-ordered because all of its nodes have the same $\maxreset$, and it is easy to see that its first node is the root of a type 1 structure (it contains no unions).
\begin{figure}[t]
	\centering
	\begin{tikzpicture}[->,
		>=stealth',
		initial text= {},
		initial distance= {2.5mm}]		
		
		\begin{scope}[xshift=0cm]
			\node[unionnode] (j0) at (0,0) {};
			\node[unionnode] (j1) at (1, -0.5) {};
			\node[] (ji) at (2, -1) {\dots};
			\node[unionnode] (jn) at (3, -1.5) {};
			\node[] (u0) at (-1, -0.5) {$u_0$};
			\node[] (u1) at (0, -1) {$u_1$};
			\node[] (un1) at (2, -2) {$u_{n-1}$};
			\node[] (un) at (4, -2) {$u_n$};
			
			\draw (j0) to (u0);
			\draw (j1) to (u1);
			\draw (jn) to (un1);
			\draw[dashed] (jn) to (un);
			\draw[dashed] (j0) to (j1);
			\draw[dashed] (j1) to (ji);
			\draw[dashed] (ji) to (jn);
			
		\end{scope}
		
	\end{tikzpicture}
	\caption{Node structure returned by $\mergeul(ul)$. The subtree $u_i$ corresponds to the $i$\textsuperscript{th} element of the \textit{union-list}.}
	\label{fig:merge-ul}
\end{figure}
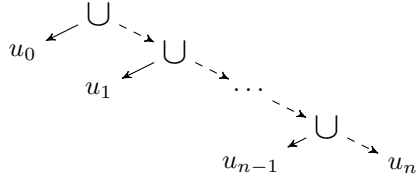

\begin{lemma}
	The result of $\mergeul(ul)$ is the root of a safe structure.
\end{lemma}

\begin{proof}
	We will prove this result by induction. If the union-list has only one element, then the result of $\mergeul(ul)$ is a safe structure, because the only element of $ul$ is the root of a safe structure.
	
	If the union-list has exactly two elements, then we know that the first element has an $\odepth$ of at most 2, and the second element has an odepth of at most $5$. Then, when adding a union node with the first element as a left child and the second element as a right child, the resulting structure is a type 2 safe structure.
	
	Assume that if $ul$ has at most $n\ge 2$ elements, then $\mergeul(ul)$ returns the root of a safe structure. Let's prove it for a length of $n+1$ elements.
	
	We know that $\mergeul(ul')$ returns the root of a safe structure, for $ul'$ the union-list containing the first $n$ elements of $ul$. We remove from the resulting structure the subtree corresponding to the $n$\textsuperscript{th} node, so that the last union has no longer a right child. We add a new union node, in its place, so its left child is the $n$\textsuperscript{th} node and its right child is the $n+1$\textsuperscript{th} node. This new tree is equal to the tree that is returned by $\mergeul(ul)$.
	
	We know that the $n$\textsuperscript{th} and the $n+1$\textsuperscript{th} node have an odepth of at most $5$ because they are roots of a safe structure. Then, the added union node has an odepth of at most $6$, and all the other nodes are already safe nodes (because they are safe structures). So the returned node is the root of a type 2 safe structure.
\end{proof}

\paragraph{Enumeration iterators}
If $\ecs$ is a CAECS with only safe nodes, and $n$ is a node of $\ecs$, we want to be able to enumerate the set $\sem{n}$, that is, we want to be able to enumerate all the open complex events starting from node $n$. We define the following problem (necessary to demonstrate Theorem~\ref{theo:algorithm}):

\begin{center}
		\begin{tabular}{rl}
			\hline\\[-2ex]
			\textbf{Problem:} & $\EnumCAECS$\\
			\textbf{Input:} & A CAECS $\ecs$, and a node $n$ of the CAECS\\
			\textbf{Output:} & Enumerate $\sem{n}$
			\\\\[-2.2ex]
			\hline
		\end{tabular}
\end{center}

To ensure that we are able to enumerate the complex events from a node of the CAECS efficiently, we define Lemma~\ref{theo:output-linear}, that ensures that if a CAECS is $k$-bounded, then the listing of the results that the CAECS represents can be done in output-linear fashion.

\begin{lemma} \label{theo:output-linear}
	Fix a value $k \ge 1$. Let $\ecs$ be an CAECS that is duplicate-free and $k$-bounded. For every node $n$ of $\ecs$, the set $\sem{n}$ can be enumerated with output-linear delay. 
\end{lemma}
\begin{proof}
	Recall the definition of output-linear delay: if $C_1,\dots, C_n$ are all the complex events that will be outputted sequentially in an enumeration phase, then there exists a constant $d$ such that it will take $d\cdot |C_i|$ time to output the event $C_i$.
	
	We know that for every result $C_i$ there exists a path from $n$ to a bottom node that encodes that result. As $\ecs$ is $k$-bounded, we know that the number of nodes to get from any node to a bottom or extended node will be at most $k$. So, following just one path from $n$ to a bottom node, we will take at most $k$ steps to arrive to the next extended node. Then, the time complexity of outputting the event $C_i$ is $\mathcal{O}(k\cdot |C_i|) = \mathcal{O}(|C_i|)$, which is output-linear.

	Note that every clock-check node is required to be able to output at most one valid open complex event, so if there is an union node $u$ in the path from $n$ to the bottom node $b$ of result $C_1$, then the right path must either output at least one result, or there must exist a clock-check node $c$ before that union such that the right path of the union does not follow the clock constraint. In that case, then we can just avoid the right path, avoiding steps that will not output any event and maintaining the output-linear delay. As such, we can perform a DFS over the $\ecs$ starting from node $n$ to output all of the results with output-linear delay. Finally, as the $\ecs$ is duplicate-free, then we will not have to take extra time at filtering the duplicates in the results.
	
\end{proof}

We have proven that if the CAECS is $k$-bounded then it is possible to output all its results in output-linear fashion. Towards that end, we will define iterators for every type of node, similar to those defined in \cite{MunozR22}, to enumerate the (open) complex events from every node with output-linear delay. The iterators are depicted in Algorithm~\ref{alg:node-iterators}, which we will explain next. We defined enumerators for every type of node of a CAECS (except the empty node). Each type of iterator provide with three main methods: the method \textsc{create} initializes the iterator with variables that are necessary for it to function. The method \textsc{next} tells us whether there is a complex event sharing the same prefix up until the node pointed at by the iterator, which has not been outputted yet, and if so, it switches the pointers to the next event. The method \textsc{print} outputs the current complex event that the iterators are capturing. Also, the union iterator contains a special method, \textsc{traverse}, that takes some nested union nodes and returns a list with all of the non-union nodes pointed at by the union nodes (it flattens a sequence of union nodes into a list). As a result, when passing a node of a CAECS to the \textsc{create} function of its respective iterator, it creates a structure that allows to output all the complex events that begin in that node, using the \textsc{next} and \textsc{print} methods.
\begin{algorithm*}[t]
	\caption{Enumeration over a node $u$ from some CAECS $\ecs$. Based on the enumeration algorithm in \cite{MunozR22} where $\ustack$ is a stack of nodes.
	}\label{alg:node-iterators}
	\smallskip
	\begin{varwidth}[t]{0.80\textwidth}
		\begin{algorithmic}[1]
			\LineComment{Bottom node iterator}
			\Procedure{{create}}{$v, t_0, c$}
			\State $\hasnext \gets {\bf true}$
			\EndProcedure	
			
			\State
			
			\Procedure{{next}}{}
			\If{$\hasnext = \bf{true}$}
			\State $\hasnext \gets \bf{false}$
			\State \Return $\bf{true}$
			\EndIf
			\State \Return $\bf{false}$
			\EndProcedure
			
			\State
			
			\Procedure{{print}}{}
			\LineComment{Doesn't do anything}
			\EndProcedure
			
			\State
			
			\LineComment{Extended node iterator}
			\Procedure{{create}}{$v, t_0, c$}
			\State $u \gets \textproc{create}(\uleft(v), t_0, c)$
			\EndProcedure	
			
			\State
			
			\Procedure{{next}}{}
			\If{$u.\textproc{next}$}
			\State \Return $\bf{true}$
			\EndIf
			\State \Return $\bf{false}$
			\EndProcedure
			
			\State
			
			\Procedure{{print}}{}
			\State ${\bf print}(\lambda(v))$
			\State $u.\textproc{print}$
			\EndProcedure
			
			\State
			\LineComment{Clock-check node iterator}
			
			\Procedure{{create}}{$v, t_0, c$}
			\State Assume $\lambda(v) := (t_0', c')$
			\If{$t_0 = {\bf{null}}\,\lor\, c = \bf{null}$}				
			\State $u \gets \textproc{create}(\uleft(v), t_0', c')$
			\Else
			\If{$t_0 - c \le t_0' - c'$} \label{line:time-window-contain} \label{line:clock-check-1}
			\State $u \gets \textproc{create}(\uleft(v), t_0', c')$
			\ElsIf{$t_0'-c' < t_0-c \le t_0'$} \label{line:clock-check-2}
			\State $u \gets \textproc{create}(\uleft(v), t_0', c-(t_0-t_0'))$ \label{line:time-window-intersection}
			\EndIf
			
			\EndIf
			\EndProcedure	
			
			\State
			\Procedure{{next}}{}
			\If{$u.\textproc{next}$}
			\State \Return $\bf{true}$
			\EndIf
			\State \Return $\bf{false}$
			\EndProcedure
			
			\State
			
			\Procedure{{print}}{}
			\State $u.\textproc{print}$
			\EndProcedure
			\algstore{myalg}
		\end{algorithmic}
	\end{varwidth}
	\begin{varwidth}[t]{0.80\textwidth}
		\begin{algorithmic}[1]
			\algrestore{myalg}
			\LineComment{Reset node iterator}
			
			\Procedure{{create}}{$v, t_0, c$}
			\State $u \gets \textproc{create}(\uleft(v), \bf{none}, \bf{none})$
			\EndProcedure	
			
			\State
			
			\Procedure{{next}}{}
			\If{$u.\textproc{next}$}
			\State \Return $\bf{true}$
			\EndIf
			\State \Return $\bf{false}$
			\EndProcedure
			
			\State
			
			\Procedure{{print}}{}
			\State $u.\textproc{print}$
			\EndProcedure
			
			\State
			\LineComment{Union node iterator}
			\Procedure{{create}}{$v, t_0, c$}
			\State $\ustack\gets\text{\sf push}(\ustack, v)$
			\State $\ustack\gets\textsc{traverse}(\ustack, t_0, c)$
			\State $u \gets \textsc{create}(\text{\sf top}(\ustack))$
			\EndProcedure	
			
			\State
			
			\Procedure{{next}}{}
			\If{$u.\textsc{next} = \bf{false}$}
			\State $\ustack\gets\text{\sf pop}(\ustack)$
			\If{$\text{\sf length}(\ustack) = 0$}
			\State \Return $\bf{false}$
			\ElsIf{$\disunion(\text{\sf top}(\ustack))$}
			\State $\ustack\gets\textsc{traverse}(\ustack, t_0, c)$
			\EndIf
			\State $u \gets \textsc{create}(\text{\sf top}(\ustack))$
			\State $u.\textsc{next}$
			\EndIf
			\State \Return $\bf{true}$
			\EndProcedure
			
			\State
			
			\Procedure{{print}}{}
			\State $u.\textsc{print}$
			\EndProcedure
			
			\State
			
			\Procedure{{traverse}}{$\ustack, t_0, c$}
			\While{$\disunion(\text{\sf top}(\ustack))$}
			\State $u \gets \text{\sf top}(\ustack)$
			\State $\ustack \gets \text{\sf pop}(\ustack)$
			\If{$t_0-\maxreset(\uright(u)) \le c$}
			\State $\ustack\gets\text{\sf push}(\ustack, \uright(u))$
			\EndIf
			\State $\ustack\gets\text{\sf push}(\ustack, \uleft(u))$
			\EndWhile
			\State \Return $\ustack$
			\EndProcedure
			
			\State
			
			\Procedure{{enumerate}}{$v$}
			\State $u \gets \textsc{create}(v, {\bf{null}}, {\bf{null}})$
			\While{$u.\textsc{next} = \bf{true}$}
			\State $u.\textsc{print}$
			\EndWhile
			\EndProcedure
		\end{algorithmic}
	\end{varwidth}
	\smallskip
\end{algorithm*}

\paragraph{Evaluation algorithm}
We will now present the algorithm to perform the update of the data structure and the enumeration of the results. We want to be able to update the structure in constant time in data complexity, and to enumerate in output-linear delay. We already have an algorithm that allows for the latter, so we will present an algorithm that makes the former possible. 

The algorithm, presented in Algorithm~\ref{alg:evaluation}, will take the following steps: for every new label and timestamp from the stream (line~\ref{line:evaluation-while}), it will add a new bottom node to the data structure (line~\ref{line:evaluation-add-bottom}), then it will execute the transitions from every active state and store the results in a new data structure (lines~\ref{line:evaluation-trans-start}~through~\ref{line:evaluation-trans-end}). More in detail, for every union-list corresponding to a different active state, it will perform the merge of the union-list to have a single node that has an equivalent output to all of the nodes of the union-list (line~\ref{line:evaluation-merge}), then it will perform all the existing transitions from the active state, adding to the CAECS the extended, clock-check and reset nodes as needed. It is important to note that the first transitions that are performed are the transitions that contain the reset of the clock, because that way we can ensure that the union-list will be time-ordered.
The correctness of Algorithm~\ref{alg:evaluation} follows by induction over the stream (similarly than in~\cite{BucchiGQRV22}).

As all the operations are proven to return safe structures, then the CAECS is safe at all times and the enumeration in the output phase will be output-linear, as proven in Lemma~\ref{theo:output-linear}. Also, we want to ensure that the update-time of the algorithm is constant. For that, we will first demonstrate that the size of the union-lists is always bounded, and furthermore that the update time is constant. We will start by the boundedness of the union-lists in the following lemma:

\begin{lemma}
	If the deterministic timed CEA that is given as input to the Algorithm~\ref{alg:evaluation} is 
	\[
	\cT = (Q, \pset, \xset, \zset, \Delta, q_0, F),
	\] 
	then the size of every union-list in that algorithm is bounded by $|Q|+2$.
\end{lemma}

\begin{proof}
	In first place, if $\ulist = u_1,\dots,u_n$  we define $\maxset(ul) = \{\maxreset(u_i)\,|\, i \le n\}$, that contains the max-resets of every node on an union list, and we define $\maxreset(ul) = \max\{\maxset(ul)\}$ the maximum max-reset in an union list. We then can define $\gmaxset{j} = \{\maxreset(T_j[q])\,|\,q\in\ckeys(T_j)\}$ for a hash table $T_j$ of union lists at reading label in position $j$ on the stream.
	
	After these initial definitions, we will try to prove the following claim: let $j<i$, for $t_k\in \maxset(T_i[q])$, if $t_k \le t_j$, then $t_k \in \gmaxset{j}$. In other words, if one of the maximum resets of the union list of state $q$ in position $i$ is the timestamp at position $k$, then it must be the maximum reset of any union list of the previous indexes $j<i$. Intuitively, if $t_k \in \maxset(T_i[q])$ then there must exist a transition from $p$ to $q$ without reset of the clock such that $t_k$ was the maximum reset of $T_{i-1}[q]$, and so on. This is because when we execute the transitions from index $i-1$ to index $i$, the new union list can only nodes with a maximum reset that was the maximum reset of a union list in index $i-1$, or the new value of the timestamp $t_i$. We will prove this result by induction.
	
	As the base case, let $j = i-1$, $t_k\in \maxset(T_i[q])$ such that $t_k \le t_j$. Then there is a transition $(p, P, \gamma, L, Z, q)\in \Delta$ such that $t_k = \maxreset(T_{i-1}[p])$. Then $t_k\in \gmaxset{i-1} = \gmaxset{j}$.	For the inductive case, we suppose that if $t_k \le t_{j+1}$, then $t_k\in \gmaxset{j+1}$. We will prove that if $t_k\le t_j$, then $t_k\in \gmaxset{j}$. As $t_k \le t_{j} \le t_{j+1}$, then $t_k\in \gmaxset{j+1}$. We then know that there exists a state $q$ such that $t_k \in \maxset(T_{j+1}[q])$. Then there exists a transition $(p, \ell, \gamma, L, Z, q)\in \Delta$ such that $t_k = \maxreset(T_j[p])$. Therefore, $t_k \in \gmaxset{j}$.
	
	We proved that for $t_k\in \maxset(T_i[q])$, if $j < i$ and $t_k \le t_j$, then $t_k \in \gmaxset{j}$. We will now prove that $|T_i[q]| \le |Q| + 2$ for all $i \ge 1$ and all $q\in Q$. Intuitively, every max reset in $T_i[q]$ must be the max reset of a state in a previous index. Then, the maximum number of values is constrained to the maximum number of max resets of the previous index, that is, the number of states. Let $t_j = \maxreset(T_i[q])$. If $j<i$, then $t_j \ge t_k$ for all $t_k\in \maxset(T_i[q])$. Then for all $t_k$ in $\maxset(T_i(q))$, $t_k\in \gmaxset{j}$. As such, $|T_i[q]| = |\maxset(T_i[q])| \le |\gmaxset{j}| \le |Q|$. If $j = i$, we will start by defining $T_i[q] = u_1u_2\dots u_n$. We know that $\maxreset(u_1) = t_i$. We also know that $\maxreset(u_2) \le \maxreset(u_1)$, and $\maxreset(u_3) < \maxreset(u_2)$, so $\maxreset(u_3) < t_i$. Let $t_m=\maxreset(u_3)$. Then for every $l \ge 3$, we know that $\maxreset(u_l) \le t_m < t_i$, so based on the previous result, $\maxreset(u_l) \in \gmaxset{m}$. Finally, $\{t_k\in\maxset(T_i[q]) \,|\, t_k\le t_m\} \subseteq \gmaxset{j}$,  and $|\{t_k\in \maxset(T_i[q]) \,|\, t_k > t_m\}| = 2$, therefore $|\maxset(T_i[q])| \le |\gmaxset{j}| + 2 \le |Q| + 2$
\end{proof}

Having already demonstrated the boundedness of the union-lists, it only left to prove that the update time of the CAECS at reading a new event $e$ of the stream is constant, i.e., it only depends on $|\cT|$ and $|e|$. We will present the following lemma:

\begin{lemma}
	The update time of the CAECS at reading a new label of the stream is bounded by $\mathcal{O}(|Q|(|Q|+|\Delta|)|e|)$, where $e$ is the next event returned by $\yield(\Stream)$.
\end{lemma}

\begin{proof}
	The number of calls to the for loops in the procedure \textproc{ExecTrans} is bounded by the number of transitions and the time to check if the event is part of the predicate, that is, $|\Delta|\cdot |e|$. In the worst case, the first case of the for loop of $\textproc{ExecTrans}$ is $\mathcal{O}(|Q|)$, as the union-list has at most $|Q| + 2$ elements and the only operation that is not $\mathcal{O}(1)$ is \textproc{Add} that is called once. The second case is $\mathcal{O}(3|Q|) = \mathcal{O}(|Q|)$, because $\ulclockcheck$, $\mergeul$ and $\textproc{Add}$ are all $\mathcal{O}(|Q|)$. For $\textproc{ExecResetTrans}$, we have the same time complexity in the for loop, so the total complexity of each iteration of the for loop is $\mathcal{O}(|Q|)$.
	
	Then, in the worst case, we call $|Q|$ times the procedures \textproc{ExecTrans} and \textproc{ExecResetTrans}, having in each one a call to $\mergeul$, adding a complexity of $\cO(|Q|^2)$. As the for loop is called $\mathcal{O}(|Q|)$ times, this adds also a complexity of $\mathcal{O}(|Q|\cdot |\Delta|)$. Finally, the ordering of the keys in line~\ref{line:enumeration-ord-keys} is at most $\mathcal{O}(|Q|\log|Q|)$. So at the end, the time complexity of the update phase of the algorithm is $\mathcal{O}(|Q|\log|Q| + |Q|(|Q|+|e|\cdot|\Delta|)) \subseteq \mathcal{O}(|Q|(|Q| + |e|\cdot |\Delta|)) \subseteq \mathcal{O}(|\cT|^2\cdot |e|)$. This complexity only depends on the size of the automaton and of the event.
\end{proof}

\begin{algorithm*}[t]
	\caption{Evaluation of an deterministic timed CEA $\cT$  over a timed stream $\Stream$.}\label{alg:evaluation}
	\smallskip
	\begin{varwidth}[t]{0.66\textwidth}
		\begin{algorithmic}[1]
			\Procedure{Evaluation}{$\cT,\Stream$}
			\State $j \gets -1$
			\State $T \gets \varnothing$
			\State $\Delta_Z \gets \{(p, P, \gamma, L, Z, q)\in \Delta\,|\,Z\ne\varnothing\}$
			\State $\Delta_\varnothing \gets \{(p, P, \gamma, L, \varnothing, q)\in \Delta\}$
			\While{$(e, t) \gets \yield(\Stream)$}\label{line:evaluation-while}
			\State $j \gets j+1$
			\State $T' \gets \varnothing$
			\State $\ulist \gets \linit(\dbottom(j, t))$\label{line:evaluation-add-bottom}
			\State $\textproc{ExecTrans}(q_0, \ulist, e, j, t, \Delta_Z)$ \label{line:evaluation-trans-start}
			\For{$p \in \ckeys(T)$}
			\State $\textproc{ExecTrans}(p, T[p], e, j, t, \Delta_Z)$
			\EndFor
			\State $\textproc{ExecTrans}(q_0, \ulist, e, j, t, \Delta_\varnothing)$
			\For{$p \in \cordkeys(T)$} \label{line:enumeration-ord-keys}
			\State $\textproc{ExecTrans}(p, T[p], e, j, t, \Delta_\varnothing)$ \label{line:evaluation-trans-end}
			\EndFor
			\State $T \gets T'$
			\State $\textproc{Output}(j)$
			\EndWhile
			\EndProcedure
			\State
			\Procedure{Add}{$q, \nnode, \ulist$}
			\If{$\disempty(\nnode)$}
			\State \Return
			\EndIf
			\If{$q \in \ckeys(T')$}
			\State $T'[q] \gets \linsert(T'[q], \nnode)$
			\Else
			\State $T'[q] \gets \ulist$
			\EndIf
			\State \Return
			\EndProcedure
			\State
			\Procedure{Output}{$j$}
			\For{$p \in \cordkeys(T)$}
			\If{$p \in F$}
			\State $\nnode \gets \lmerge(T[p])$
			\State $\textproc{Enumerate}(\nnode, j)$
			\EndIf
			\EndFor
			\EndProcedure
			\algstore{myalg}
		\end{algorithmic}
	\end{varwidth}
	\begin{varwidth}[t]{0.80\textwidth}
		\begin{algorithmic}[1]
			\algrestore{myalg}
			\Procedure{ExecTrans}{$p, \ulist, e, j, t, \Delta_O$}
			\State $\nnode \gets \lmerge(\ulist)$ \label{line:evaluation-merge}
			\For{$(p, P, \gamma, L, Z, q)\in \Delta_O$ such that $e \in P$}
			\If{$L \ne \varnothing$}
			\State $\nnode' \gets \dextend(\nnode, j, L)$
			\If{$\gamma \ne \TRUE$}
			\State Assume $\gamma := X \le c$
			\State $\nnode' \gets \dclockcheck(\nnode', t, c)$
			\EndIf
			\If{$Z \ne \varnothing \,\land\, \disempty(\nnode')$}
			\State $\nnode' \gets \dreset(\nnode', t)$
			\EndIf
			\If{$\neg\, \disempty(n')$}
			\State $\ulist' \gets \newul(\nnode')$
			\State $\textproc{Add}(q, \nnode', \ulist')$
			\EndIf
			\Else
			\State $\ulist' \gets \ulist$
			\If{$\gamma \ne \TRUE$}
			\State Assume $\gamma := X \le c$
			\State $\ulist' \gets \ulclockcheck(\ulist', t, c)$
			\EndIf
			\If{$Z \ne \varnothing \,\land\, \neg\, \disempty(\ulist')$}
			\State $\ulist' \gets \ulreset(\ulist', t)$
			\EndIf
			\If{$\neg\, \disempty(\ulist')$}
			\State $\nnode' \gets \mergeul(\ulist')$
			\State $\textproc{Add}(q, \nnode', \ulist')$
			\EndIf
			\EndIf
			\EndFor
			\State \Return
			\EndProcedure
		\end{algorithmic}
	\end{varwidth}
	\smallskip
\end{algorithm*} 	

\end{document}